\documentclass[10pt,onecolumn,journal]{IEEEtran}
\usepackage{indentfirst,amsmath}
\usepackage{color}

\usepackage{amsfonts,amsmath,amssymb,amsthm}
\usepackage{latexsym,amscd}
\usepackage{amsbsy}
\usepackage{graphicx,multirow,bm}
\usepackage{algorithm}
\usepackage{algorithmicx}
\usepackage{algpseudocode}
\usepackage{xcolor}
\usepackage{tikz}
\usepackage{diagbox}
\usepackage{cite}
\usepackage{arydshln}
\usepackage{bbm}

\newtheorem{Theorem}{Theorem}

\newtheorem{Definition}{Definition}
\newtheorem{Example}{Example}
\newtheorem{Property}{Property}
\newtheorem{Remark}{Remark}
\newtheorem{Lemma}{Lemma}

\newtheorem*{Fact}{Fact}

\newcommand{\tabincell}[2]{\begin{tabular}{@{}#1@{}}#2\end{tabular}}

\DeclareMathAlphabet{\mathpzc}{OT1}{pzc}{m}{it}

\begin{document}
	\title{Optimal Repair/Access MDS Array Codes with Multiple Repair Degrees
			\author{Yi Liu, Jie Li, \IEEEmembership{Member,~IEEE}, and Xiaohu Tang, \IEEEmembership{Senior Member,~IEEE}}
		\thanks{Y. Liu and X.H. Tang are with the Information Security and National Computing Grid Laboratory, Southwest Jiaotong University, Chengdu, 610031, China (e-mail: yiliu.swjtu@outlook.com, xhutang@swjtu.edu.cn).}
		\thanks{J. Li was with the Hubei Key Laboratory of Applied Mathematics, Faculty of Mathematics and Statistics, Hubei University, Wuhan 430062, China (e-mail: jieli873@gmail.com).}}
	\date{}
	\maketitle

	\begin{abstract}
		In the literature, most of the known high-rate $(n,k)$ MDS array codes with the optimal repair property only support a single repair degree (i.e., the number of helper nodes contacted during a repair process)  $d$, where $k\le d\le n-1$. However, in practical storage systems, the number of available nodes changes frequently. Thus, it is preferred to construct $(n,k)$ MDS array codes with multiple repair degrees and the optimal repair property for all nodes. To the best of our knowledge, only two MDS array codes have such properties in the literature, which were proposed by Ye and Barg (IEEE Trans. Inform. Theory, 63(10), 2001-2014, 2017). However, their sub-packetization levels are relatively large. In this paper, we present a generic construction method that can convert some MDS array codes with a single repair degree into the ones with multiple repair degrees and optimal repair property for a set of nodes, while the repair efficiency/degrees of the remaining nodes can be kept.  
		As an application of the generic construction method, an explicit construction of high-rate MDS array code with multiple repair degrees and the optimal access property for all nodes is obtained over a small finite field. Especially, the sub-packetization level  is much smaller than that of the two codes proposed by Ye and Barg concerning the same parameters $n$ and $k$.
	\end{abstract}
	
	
	\section{Introduction}
	Distributed storage systems, such as those run by  Hadoop, Google Colossus, Microsoft Azure \cite{Micro}, OceanStore \cite{ocean}, Total Recall \cite{total}, and DHash++ \cite{Dhash}, are widely used in not only large-scale data centers but also peer-to-peer storage settings. Currently, deployed distributed storage systems are formed of thousands of individual nodes, where the node failures are 
normal. Therefore, in order to ensure reliability, a certain amount of redundant data should be stored in the distributed storage system as well. Conventionally, distributed storage systems use replications to produce redundant data, such as HDFS \cite{HDFS}. However, due to the large storage consumption of exact replicas, there is a trend for distributed storage systems to migrate from replications to erasure codes \cite{Micro}. Compared with the former, erasure codes can offer higher reliability at the same redundancy level and thus have been extensively deployed in distributed storage systems.
	
	Among families of erasure codes, maximum distance separable (MDS) codes provide optimal trade-off between fault-tolerance and storage overhead. By distributing the codeword across distinct storage nodes, in the case of node failures, the missing data can be recovered from the data at some surviving nodes, which are named \textit{helper nodes}. During the repair process, efficient operation of the system requires minimizing the \textit{repair bandwidth}, which is defined as the amount of data downloaded to repair a failed node \cite{Dimakis}.
	
	It was proved in \cite{Dimakis} that for an $(n,k)$ MDS code with code length $n$ and dimension $k$, the recovery of a single failed node from $d$ helper nodes should download at least a  fraction $\frac{1}{d-k+1}$ of the data stored in each of the helper nodes, i.e., the repair bandwidth $\gamma(d)$ satisfies
	\begin{eqnarray}\label{Eqn_bound_on_gamma}
		\gamma(d)\ge \frac{d}{d-k+1}N,
	\end{eqnarray}
	where $d\in [k:n)$ and $N$ are called the \textit{repair degree} and \textit{sub-packetization level}, respectively. Particularly, the code is referred to as an \textit{array code} if $N>1$ \cite{array codes}. In the literature, most existing MDS codes are designed as a kind of array codes to achieve the lower bound in \eqref{Eqn_bound_on_gamma}. In this paper, we also focus on MDS array codes.
	
For $d\in [k:n)$, rewrite $d=k+\theta-1$, where $2\le \theta\le  n-k$. For an $(n,k)$ MDS array code, if the repair bandwidth attains the lower bound in \eqref{Eqn_bound_on_gamma} when repairing a failed node by connecting $d$ helper nodes, we say that the code has the \textit{$\theta$-optimal repair property} for this node. More generally, given any $m$ ($m\ge 2$) positive integers $\delta_0,\delta_1,\cdots,\delta_{m-1}$ with $2 \le \delta_0<\delta_1<\cdots<\delta_{m-1} \le r=n-k$, if a node has the $\delta_z$-optimal repair property for all $0\le z < m$, we say that this node has the \textit{$\delta_{[0:m)}$-optimal repair property}, where $\delta_{[0:m)}=\{\delta_0,\delta_1,\cdots,\delta_{m-1}\}$. Besides the repair bandwidth, some other metrics also need to be optimized in practice. In general, during the process of repairing a failed node, a symbol downloaded from one helper node can be a  linear combination of several symbols at this node, and the amount of data accessed can be more than that transmitted. When repairing a failed node by connecting $d=k+\theta-1$ helper nodes, if the amount of data accessed from the helper nodes also meets the lower bound in \eqref{Eqn_bound_on_gamma}, we say that the $(n,k)$ MDS array code has the \textit{$\theta$-optimal access property} for this node. Similarly, we also say one node has the \textit{$\delta_{[0:m)}$-optimal access property}  if it has $\delta_z$-optimal access property for all $0\le z<m$.  Actually, the optimal access property implies the optimal repair property, but not vice versa. In this sense, the optimal access property can be viewed as an enhanced property of the optimal repair property. 
	
Up to now, for $k>n/2$ (i.e., the high-rate regime),  some explicit constructions of MDS array codes which support a single repair degree and with the $\theta$-optimal repair property have been proposed, where $2\le \theta\le r$. Among them, most constructions are limited to the case of $\theta=r$,  i.e.,  repairing a failed node requires connecting all the $n-1$ surviving nodes, where some of the notable works are \cite{zigzag,extended_zigzag,Long_arxiv,repair_parity_zigzag,new_modified_zigzag,invariant_subspace,all_nodes,Hadamard_strategy,Sasidharan_Kumar,coupled_layer,RES}. Only a few known explicit constructions of MDS array codes support $d<n-1$ (i.e., $\theta<r$) \cite{Barg_1,coupled,kumar}, however, they either have a large sub-packetization level (e.g., the two codes proposed in \cite{Barg_1}) or have restrictions on the choices of the parameter $d$ or equivalently $\theta$ (e.g., the MDS array codes proposed in \cite{coupled,kumar}), where the two MDS array codes proposed in \cite{Barg_1} are respectively called YB code 1 and  YB code 2. Particularly, in this paper, the MDS array code with the optimal access property and optimal sub-packetization level proposed in \cite{kumar} is called VBK code. More recently, in \cite{t-transformation}, Liu \textit{et al.} proposed an explicit construction of high-rate $(n,k)$ MDS array code with the $\theta$-optimal access property for all nodes, where the sub-packetization level is $\theta ^{\lceil\frac{n}{2} \rceil}$, which is between that of the YB codes 1, 2 \cite{Barg_1} and the ones proposed in \cite{coupled,kumar}. 
	
Constructions of high-rate MDS array codes with multiple repair degrees were first proposed by Ye and Barg \cite{Barg_1} in 2017, where two $(n,k)$ MDS array codes with all nodes having $\delta_{[0:m)}$-optimal repair property for any subset $\delta_{[0:m)}$ of $[2:r]$ were proposed. Both codes have sub-packetization levels $\delta^n$, where
	\begin{eqnarray}\label{eqn value delta}
		\delta=\mathrm{lcm}(\delta_0,\delta_1,\cdots,\delta_{m-1}).
	\end{eqnarray}  Specifically, the parity-check matrices of the two MDS array codes are based on diagonal matrices and permutation matrices.
	For convenience, we refer to the one based on diagonal matrices as YB code 3 and the other one as YB code 4. To the best of our knowledge,  YB codes 3 and 4 are the only two known high-rate MDS array codes with the $\delta_{[0:m)}$-optimal repair property for all nodes in the literature. However, their sub-packetization levels are relatively large.

	In this paper, we aim to construct high-rate $(n,k)$ MDS array codes that have  $\delta_{[0:m)}$-optimal repair property for all nodes and a lower sub-packetization level. By this motivation, we provide a generic construction method that can convert some known MDS array codes with the $\delta_0$-optimal repair property into another MDS array code, which makes a set of nodes possessing the $\delta_{[0:m)}$-optimal repair property, and simultaneously  preserves the optimal repair/access property for the remaining nodes. By applying this generic construction method multiple times, an algorithm is proposed that can construct $(n,k)$ MDS array codes with the $\delta_{[0:m)}$-optimal repair property for all nodes from a class of special $(n,k)$ MDS array codes with the $\delta_0$-optimal repair property for all nodes. As application of the algorithm to VBK code in \cite{kumar}, we obtain an explicit high-rate $(n,k)$ MDS array code $\mathbbmss{G}$ which has the $\delta_{[0:m)}$-optimal access property for all nodes. Specifically, the new code $\mathbbmss{G}$ has a sub-packetization level $\delta^{\lceil\frac{n}{\delta_0}\rceil}$ for $\delta_0=2,3,4$, which is much smaller than that of YB codes 3 and 4, where $\delta$ is defined in \eqref{eqn value delta}. When $\delta_0>4$, consider the new code $\mathbbmss{G}$ with $(\{4\}\cup \delta_{[0:m)})$-optimal access property for all nodes, it not only has a smaller sub-packetization level $(\mathrm{lcm}(4,\delta))^{\lceil \frac{n}{4} \rceil}$ than that of YB codes 3 and 4, but also can support one more repair degree than YB codes 3 and 4. 
	
	The remainder of the paper is organized as follows. Section II reviews some necessary preliminaries. Section III proposes the generic construction method and its asserted properties. Section IV gives the algorithm of this method and a new  explicit construction of high-rate MDS array code which is obtained by means of this algorithm. Section V gives comparisons of key parameters among the MDS array code proposed in this paper and YB codes 3, 4. Finally, Section VI concludes this paper.

	\section{Preliminaries}\label{section S2}
	
	In this section, we introduce the MDS property and optimal repair property of MDS array codes, and a series of special partitions for a given standard basis set.
	
	\subsection{Structure of MDS Array Codes}
	
	Let $\mathbb{F}_q$ be a finite field with $q$ elements where $q$ is a prime power. For two non-negative integers $a$ and $b$ with $a<b$, define $[a:b)$ and $[a:b]$ as two ordered sets $\{a,a+1,\cdots,b-1\}$ and $\{a,a+1,\cdots,b\}$, respectively. An $(n,k)$ array code encodes a file of size $\mathcal{M}=kN$ into $n$ fragments $\mathbf{f}_0,\mathbf{f}_1,\cdots,\mathbf{f}_{n-1}$, which are stored across $n$ nodes, respectively, where $\mathbf{f}_i=(f_{i,0},f_{i,1},\cdots,f_{i,N-1})^{\top}$ is a column vector of length $N$ over $\mathbb{F}_q$, and $\top$ denotes the transpose operator.  	

	In this paper, $(n,k)$ array codes are assumed to be defined in the following parity-check form:
	\begin{equation}\label{Eqn Code C}
		\underbrace{\left(
			\begin{array}{cccc}
				A_{0,0} & A_{0,1} & \cdots & A_{0,n-1}\\
				A_{1,0} & A_{1,1} & \cdots & A_{1,n-1}\\
				\vdots & \vdots &\ddots & \vdots\\
				A_{r-1,0}&A_{r-1,1}&\cdots&A_{r-1,n-1}\\
			\end{array}
			\right)}_{block~ matrix~ A}
		\left(
		\begin{array}{c}
			\mathbf{f}_0\\\mathbf{f}_1\\\vdots\\\mathbf{f}_{n-1}\\
		\end{array}
		\right)=\mathbf{0}_{rN}
	\end{equation}
	where $r=n-k$. Throughout this paper, $\mathbf{0}_{N}$ (resp. $\mathbf{0}_{N\times M})$ denotes the zero column of length $N$ (resp. matrix of order $N\times M$), and will be abbreviated as $\mathbf{0}$ in the sequel if the length (resp. order) is clear. In \eqref{Eqn Code C}, the $rN\times rN$ matrix $A$ is called the \textit{parity-check matrix} of the code, which can be simplified as $$A=(A_{t,i})_{t\in [0:r),i\in [0:n)}$$ to indicate the block entries. Note that for each $t\in [0:r)$, $\sum\limits_{i=0}^{n-1}A_{t,i}\mathbf{f}_i=\mathbf{0}$ contains $N$ equations, for convenience, we say that $\sum\limits_{i=0}^{n-1}A_{t,i}\mathbf{f}_i=\mathbf{0}$ is the $t$-th \textit{parity-check group} (PCG), where $A_{t,i}$ is an $N\times N$ matrix  over $\mathbb{F}_q$. 
	
	An $(n,k)$ code is said to have the MDS property if the original file can be reconstructed by connecting any $k$ out of the $n$ nodes, i.e., the data stored in any set of $r=n-k$ nodes can be obtained by the remaining $k$ nodes.
	
	\begin{Lemma}[\cite{Barg_1}]
		An $(n,k)$ array code defined by \eqref{Eqn Code C} has the MDS property if and only if the block matrix
		\begin{eqnarray*}
			\left(
			\begin{array}{cccc}
				A_{0,i_0} & A_{0,i_1} & \cdots & A_{0,i_{r-1}}\\
				A_{1,i_0} & A_{1,i_1} & \cdots & A_{1,i_{r-1}}\\
				\vdots & \vdots &\ddots & \vdots\\
				A_{r-1,i_0}&A_{r-1,i_1}&\cdots&A_{r-1,i_{r-1}}\\
			\end{array}
			\right)
		\end{eqnarray*}
		of order $rN$ is nonsingular for any $r$-subset $\{i_0,i_1,\cdots,i_{r-1}\}\subset [0:n)$. 
	\end{Lemma}
	
	\subsection{Optimal Repair Property}\label{section sec2-1}
	
	An $(n,k)$ MDS code with the $\theta$-optimal repair property is preferred, i.e., any failed node can be repaired by downloading $\frac{N}{\theta}$ symbols from each of the $d=k+\theta-1$ helper nodes. In this paper, similarly to that in \cite{t-transformation}, when repairing a failed node $i\in [0:n)$, the $\frac{N}{\theta}$ symbols downloaded from each helper node $j\in \mathcal{H}$ is represented by $R_{i,\theta}\mathbf{f}_j$, where $\mathcal{H}$ denotes the indices set of the $d$ helper nodes and the $\frac{N}{\theta}\times N$ matrix  $R_{i,\theta}$ of full rank is called the \textit{$\theta$-repair matrix} of node $i$. In addition, the code is preferred to have the $\theta$-optimal access property, i.e., when repairing a failed node $i\in [0:n)$, the amount of accessed data attains the lower bound in \eqref{Eqn_bound_on_gamma}. Clearly, node $i$ has the $\theta$-optimal access property if the $\theta$-repair matrix $R_{i,\theta}$ satisfies that each row has only one nonzero element. 
	
	Obviously, some linear independent equations
	should be chosen  out of those $rN$ parity-check equations in  \eqref{Eqn Code C}  to regenerating a failed node $i\in [0:n)$.  Precisely, for any $t\in [0:r)$, we get $\frac{N}{\theta}$  linear independent equations from the $t$-th PCG of \eqref{Eqn Code C} by multiplying it with an $\frac{N}{\theta} \times N$ matrix $S_{i,\theta}$ of rank $\frac{N}{\theta}$, where $S_{i,\theta}$ is called the \textit{$\theta$-select matrix} of node $i$. As a consequence, the following linear system of equations (LSE) are available,
	\begin{eqnarray}\label{Eqn linear equations}
		\underbrace{\left(\begin{array}{c}
				S_{i,\theta} A_{0,i}\\
				S_{i,\theta} A_{1,i}\\
				\vdots\\
				S_{i,\theta} A_{r-1,i}
			\end{array}\right)\mathbf{f}_i}_{\mathrm{useful~ data}}+\sum_{j\in \mathcal{D}}\underbrace{\left(\begin{array}{c}
				S_{i,\theta} A_{0,j}\\
				S_{i,\theta} A_{1,j}\\
				\vdots\\
				S_{i,\theta} A_{r-1,j}
			\end{array}\right)\mathbf{f}_j}_{\mathrm{intermediate~ data}}
		+\sum_{j\in \mathcal{H}}\underbrace{\left(\begin{array}{c}
				S_{i,\theta} A_{0,j}\\
				S_{i,\theta} A_{1,j}\\
				\vdots\\
				S_{i,\theta} A_{r-1,j}
			\end{array}\right)\mathbf{f}_j}_{\mathrm{interference~ by}~\mathbf{f}_{j}}
		=\mathbf{0},
	\end{eqnarray}
	where  $\mathcal{D}=[0:n)\backslash(\mathcal{H}\cup \{i\})$ is the index set of the $r-\theta$ nodes which are not connected, particularly $\mathcal{D}=\emptyset$ if  $\theta=r$.
	
	Therefore, the optimal repair property indicates that the interference terms caused by $\mathbf{f}_{j}$ have to be cancelled by the downloaded data $R_{i,\theta} \mathbf{f}_{j}$ from node $j\in \mathcal{H}$, i.e.,
	\begin{eqnarray*}
		\mbox{Rank} \left(\left(
		\begin{array}{c}
			R_{i,\theta} \\
			S_{i,\theta} A_{0,j} \\
			\vdots\\
			S_{i,\theta} A_{r-1,j} \\
		\end{array}
		\right)\right) =\frac{N}{\theta}
	\end{eqnarray*}
for all  $j\in \mathcal{H}$ and further for all $j\in [0:n)\backslash\{i\}$ since $\mathcal{H}$ is an arbitrary $d$-subset of $[0:n)$,
	which means that
	\begin{eqnarray*}
		\mbox{Rank} \left(\left(
		\begin{array}{c}
			R_{i,\theta} \\
			S_{i,\theta} A_{t,j} \\
		\end{array}
		\right)\right) =\frac{N}{\theta} \textrm{~for~all~} j\in [0:n)\backslash\{i\} \textrm{~and~}t\in [0:r).
	\end{eqnarray*}
That is, there exists an $\frac{N}{\theta}\times \frac{N}{\theta}$ matrix $\tilde{A}_{t,j,i,\theta}$ such that
	\begin{eqnarray}\label{repair node requirement3}
		S_{i,\theta}A_{t,j}=\tilde{A}_{t,j,i,\theta}R_{i,\theta}  \textrm{~for~}  j\in [0:n)\backslash\{i\} \textrm{~and~}t\in [0:r).
	\end{eqnarray}
	
	Let $\mathcal{D}=\{j_0,j_1,\cdots,j_{r-\theta-1}\}$, by substituting  \eqref{repair node requirement3} into \eqref{Eqn linear equations}, together with the data $R_{i,\theta}\mathbf{f}_j$ downloaded from each helper node $j\in \mathcal{H}$,  \eqref{Eqn linear equations} can be rewriten as
	\begin{eqnarray}\label{Eqn linear equations eq 1}
		\left(\begin{array}{cccc}
			S_{i,\theta} A_{0,i} & \tilde{A}_{0,j_0,i,\theta} & \cdots & \tilde{A}_{0,j_{r-\theta-1},i,\theta}\\
			S_{i,\theta} A_{1,i}  & \tilde{A}_{1,j_0,i,\theta} & \cdots & \tilde{A}_{1,j_{r-\theta-1},i,\theta}\\
			\vdots & \vdots & \vdots & \vdots\\
			S_{i,\theta} A_{r-1,i}  & \tilde{A}_{r-1,j_0,i,\theta} & \cdots & \tilde{A}_{r-1,j_{r-\theta-1},i,\theta}
		\end{array}\right)\left(\begin{array}{c}
			\mathbf{f}_i\\
			R_{i,\theta}\mathbf{f}_{j_0}\\
			\vdots\\
			R_{i,\theta}\mathbf{f}_{j_{r-\theta-1}}
		\end{array}\right)=-\underbrace{\sum_{j\in \mathcal{H}}\left(\begin{array}{c}
				\tilde{A}_{0,j,i,\theta} \\
				\tilde{A}_{1,j,i,\theta}\\
				\vdots\\
				\tilde{A}_{r-1,j,i,\theta}
			\end{array}\right)R_{i,\theta}\mathbf{f}_j}_{\mathrm{known~data}},
	\end{eqnarray}
	where the term on the right hand side (RHS) of \eqref{Eqn linear equations eq 1} is  determined by the downloaded data.  
	It is clear that there are $N+\frac{N}{\theta}|\mathcal{D}|=N+\frac{N}{\theta}(r-\theta)=\frac{rN}{\theta}$ unknown variables with $\frac{rN}{\theta}$ equations in \eqref{Eqn linear equations eq 1}. Then we have the following result.
	\begin{Lemma}\label{lem the requirement of obtaining fi}
		For given $\theta\in [2:r]$ and $i\in [0:n)$, if node $i$ has the $\theta$-optimal repair property, then the $\frac{rN}{\theta}\times \frac{rN}{\theta}$ coefficient matrix of \eqref{Eqn linear equations eq 1} is nonsingular for any $(r-\theta)$-subset $\{j_0,j_1,\cdots,j_{r-\theta-1}\}\subset[0:n)\backslash\{i\}$.
	\end{Lemma}

	\subsection{Standard Basis Sets}
	For any two positive integer $s,w\geq 2$, let $\{e_0,e_1,\cdots, e_{s^w-1}\}$ be the standard basis of $\mathbb{F}_q^{s^w}$, i.e.,
	\begin{equation}\label{Eqn_SB}
		e_a=(0,\cdots,0,1,0,\cdots,0),\,\,a\in [0:s^w),
	\end{equation}
	with only the $a$-th entry being nonzero.
	
	Given $a \in [0:s^w)$, denote $(a_{w-1},a_{w-2},\cdots,a_0)$ as its $s$-ary expansion, i.e.,
	\begin{eqnarray}\label{Eqn a_expansion}
		a=a_{w-1}s^{w-1}+a_{w-2}s^{w-2}+\cdots+a_0,
	\end{eqnarray}
where $a_i$ is the $i$-th element in the $s$-ary expansion of $a$. Throughout this paper, we do not distinguish the integer $a$ and its $s$-ary expansion if the context is clear.
	
	Based on \eqref{Eqn a_expansion}, we further define some subsets of the standard basis set $\{e_0,e_1,\cdots,e_{s^w-1}\}$ as
	\begin{eqnarray}\label{Eqn Vt}
		V_{i,u}=\{e_a|a_i=u, a\in [0:s^w)\}, & 0\le i<w, u\in [0:s).
	\end{eqnarray}
	\begin{Example}
		When $w=3$, $s=2$, Table \ref{table partition} gives $V_{i,u}$, $0\le i< 3$ and $0\le u < 2$.
		\begin{table*}[htbp]
			\centering
			\caption{An illustrative example of $V_{i,u}$}
			\label{table partition}
			\begin{tabular}{|c|c|c|c|c|c|c|c|}
				\hline $i$ & 0 & 1 & 2 & $i$ & 0 & 1 & 2\\
				\hline \multirow{4}{*}{$V_{i,0}$ }&$e_0$&$e_0$ & $e_0$ &\multirow{4}{*}{$V_{i,1}$ }&$e_1$&$e_2$ & $e_4$\\  & $e_2$&$e_1$&$e_1$&&$e_3$&$e_3$&$e_5$\\&$e_4$&$e_4$&$e_2$&& $e_5$&$e_6$&$e_6$\\&$e_6$&$e_5$&$e_3$&&$e_7$&$e_7$&$e_7$\\
				\hline
			\end{tabular}
		\end{table*}
	\end{Example}
	
	For easy of notation,  we also denote by $V_{i,u}$ the $s^{w-1}\times s^w$  matrix, whose rows are formed by vectors $e_a$
	in its corresponding sets,  such that $a$  is sorted in ascending order. For example, when $s=2$ and $w=3$, $V_{1,0}$ can be viewed as a $4\times 8$ matrix
	\begin{eqnarray*}
		V_{1,0}=\left(e_0^{\top}~ e_1^{\top} ~e_4^{\top}~ e_5^{\top}\right)^{\top}.
	\end{eqnarray*}
	
\subsection{Notations}	
	Throughout this paper, the following notations are used.
	\begin{itemize}
		\item For a matrix $Q$, define $Q(u,:)$, $Q(:,v)$ and $Q(u,v)$ as its $u$-th row vector, its $v$-th column vector and the entry in row $u$ and column $v$.
		\item For a matrix $Q$, define $\mbox{blkdiag}(Q,Q,\cdots,Q)_{t}$ as a block diagonal matrix with $Q$ occurring $t$ times.
		\item The symbols $\%$ denotes the modulo operation.
		\item For any positive $a$, denote $I_a$ the identity matrix of order $a$.
	\end{itemize}
	
	\section{A Generic Construction Method}\label{section MDS array code with dz optimal property}

	
	In this section, we propose a method that can transform an $(n,k)$ MDS array code with the $\delta_0$-optimal repair property for all nodes into a new $(n,k)$ MDS array code with the $\delta_{[0:m)}$-optimal repair property for a set of $\rho~(1\le \rho\le \delta_0)$ goal nodes (GNs), where these $\rho$ GNs are required to satisfy some specific conditions and $\delta_{[0:m)}\subseteq [2:r]$,  while keeping the repair property of the other $n-\rho$ remainder nodes (RNs) intact. Specifically, given an $(n,k)$ base code, let $\mathcal{G}$ be the set of indices of the $\rho$ GNs which we wish to endow with the $\delta_{[0:m)}$-optimal repair property.
		
	\subsection{The Generic Construction Method}\label{Section the generic SA}
In this subsection, we propose the generic construction method, which utilizes a known $(n,k)$ MDS array code $\mathbbmss{C}_0$ with sub-packetization level $\alpha N$ over $\mathbb{F}_q$ and $\delta_0$-optimal repair property as the base code, where $\alpha\ge 1$, $\delta_0\mid N$. Let $(A_{t,i})_{t\in [0:r),i\in [0:n)}$ be the parity-check matrix of base code $\mathbbmss{C}_0$ while the $ \frac{\alpha N}{\delta_z}\times \alpha N$ matrices $R_{i,\delta_z}$ and $S_{i,\delta_z}$, $i\in [0:n)$ respectively denote the $\delta_z$-repair matrix and $\delta_z$-select matrix of  base code $\mathbbmss{C}_0$ if it also has $\delta_z$-optimal repair property for some $\delta_z$ with $z\ge 1$.  
	For convenience, throughout this paper, we always set $N'=\frac{N}{\delta_0}$.

The following example shows an MDS array code that possesses the $\delta_0$-optimal repair property and will be chosen as the base code	throughout the examples of this paper.

			\begin{Example}\label{Exp set Piz for (1,2,3)}
The $(16,10)$ MDS array code in \cite{kumar} has sub-packetizaton level $N=2^8$ and $\delta_0$-optimal repair property for all nodes, and also satisfies some other specific properties which will be illustrated later, where $\delta_0=2$. It can be chosen as the base code $\mathbbmss{C}_0$,  the $\delta_0$-repair matrix $R_{i,\delta_0}$ and $\delta_0$-select matrix $S_{i,\delta_0}$ are defined by
			\begin{eqnarray*}
				R_{i,\delta_0}=S_{i,\delta_0}=V_{\lfloor\frac{i}{2}\rfloor,i\%2},\, 0\le i<16,
			\end{eqnarray*}
where $V_{j,0},V_{j,1}$ ($0\le j<8$) are given in \eqref{Eqn Vt}. 
\end{Example}

	Define
	\begin{eqnarray}\label{eqn notation l}
		l_z=\left\{\begin{array}{ll}
			\frac{\delta}{\delta_z}, & \mathrm{if~} 0\le z<m,\\
			0, & \mathrm{if~}z=m,
		\end{array}\right.
	\end{eqnarray}
	where  $\delta$ is defined in \eqref{eqn value delta}. 
	
	The generic construction method is then carried out through the following two steps.\\
	\textbf{\textit{Step 1: An intermediate MDS array code $\mathbbmss{C}_1$ by space sharing $l_0$ instances of code $\mathbbmss{C}_0$.}}
	
Construct an intermediate MDS array code $\mathbbmss{C}_1$ with sub-packetization level $l_0 \alpha N$ by space sharing $l_0$ instances of the base code $\mathbbmss{C}_0$. Specifically, for each instance $a\in [0:l_0)$, the $t$-th PCG is of the form
	\begin{eqnarray*}
		\sum\limits_{i\in \mathcal{G}}A_{t,i}\mathbf{f}_i^{(a)}+\sum\limits_{i\in [0:n)\backslash(\mathcal{G}\cup \mathcal{R})}A_{t,i}\mathbf{f}_i^{(a)}+\sum\limits_{j\in \mathcal{R}}A_{t,j}\mathbf{f}_j'^{(a)}=\mathbf{0},\,t\in [0:r),
	\end{eqnarray*}
where $\mathcal{R}$ denotes any given $r$-subset of $[0:n)\backslash\mathcal{G}$, $\mathbf{f}_i^{(a)}$ and $\mathbf{f}_j'^{(a)}$ respectively denote the data stored at nodes $i$ and $j$ of an instance of the code $\mathbbmss{C}_0$ for $i\in [0:n)\backslash\mathcal{R}$,  $j\in \mathcal{R}$, and $a\in [0:l_0)$. \\
	\textbf{\textit{Step 2: Construct code $\mathbbmss{C}_2$  by appending some data of each goal node to the   PCGs  of $\mathbbmss{C}_1$}}
	
	Based on code $\mathbbmss{C}_1$, we construct the desired storage code $\mathbbmss{C}_2$ by appending the data $\mathbf{P}_{t,i}^{(a)}(i\in \mathcal{G})$ called \textit{appended-data} to the $t$-th PCG of instance $a$ of $\mathbbmss{C}_1$, which leads to new parity-check equations and means that the data stored at some $r$ nodes will be modified. By convention, we assume that the data $\mathbf{f}_j'^{(a)}$ stored at node $j\in \mathcal{R}$ of instance $a\in [0:l_0)$ is changed to $\mathbf{f}_j^{(a)}$ and the data stored at the other nodes is unchanged. Then the $t$-th PCG of instance $a$ of new code $\mathbbmss{C}_2$ is given by 
	\begin{eqnarray}\label{eqn t-PCG of code C3}
		\sum\limits_{i\in \mathcal{G}}
		(A_{t,i}\mathbf{f}_i^{(a)}+\mathbf{P}_{t,i}^{(a)})+\sum\limits_{j\in [0:n)\backslash(\mathcal{G}\cup \mathcal{R})}
		A_{t,j}\mathbf{f}_j^{(a)}+\sum\limits_{j\in \mathcal{R}}A_{t,j}\mathbf{f}_j^{(a)}=\mathbf{0},& a\in [0:l_0),t\in [0:r),
	\end{eqnarray}
	where
	\begin{itemize}
		\item  [P0.] The appended-data $\mathbf{P}_{t,i}^{(a)}$  is to be designed as a linear combination of $\mathbf{f}_i^{(l_z)},\mathbf{f}_i^{(l_z+1)},\cdots,\mathbf{f}_i^{(l_0-1)}$ if $l_{z+1}\le a<l_z$ with $z\in [1:m)$ and  $\mathbf{P}_{t,i}^{(a)}=\mathbf{0}$ if $l_1\le a<l_0$ for $t\in [0:r),i\in \mathcal{G}$.
	\end{itemize}
	
	Obviously, the new code $\mathbbmss{C}_2$ maintains the MDS property of base code $\mathbbmss{C}_0$.
	\begin{Theorem}\label{theorem the MDS property of code C3}
		The new $(n,k)$ code $\mathbbmss{C}_2$ maintains the MDS property of code $\mathbbmss{C}_0$.
	\end{Theorem}
	\begin{proof}
		The new code $\mathbbmss{C}_2$ possesses the  MDS property if the data stored in any $r$ out of $n$ nodes can be obtained by the remaining $k$ nodes.  Let $i_0,i_1,\cdots,i_{r-1}$  be the indices of those $r$ nodes.
		For any $a\in [0:l_0)$,  we can obtain
		\begin{eqnarray}\label{eqn MDS for resultant code C3}
			\sum\limits_{v=0}^{r-1}A_{t,i_v}\mathbf{f}_{i_v}^{(a)}+\sum\limits_{v=0,i_v\in \mathcal{G}}^{r-1}\mathbf{P}_{t,i_v}^{(a)}=-\underbrace{\sum\limits_{v=0}^{k-1}A_{t,j_v}\mathbf{f}_{j_v}^{(a)}-\sum\limits_{v=0,j_v\in \mathcal{G}}^{k-1}\mathbf{P}_{t,j_v}^{(a)}}_{\mathrm{known~data}},\,\,t\in [0:r),
		\end{eqnarray}
		from \eqref{eqn t-PCG of code C3}, where $\{j_0,j_1,\cdots,j_{k-1}\}=[0:n)\backslash\{i_0,i_1,\cdots,i_{r-1}\}$ and the two terms on RHS of \eqref{eqn MDS for resultant code C3} are  determined by the data  stored at the remaining $k$  nodes.
We prove the MDS property by induction in the following.

i)		
		According to P0, $\mathbf{P}_{t,j}^{(a)}=\mathbf{0}$ for $a\in [l_1:l_0)$, $j\in \mathcal{G}$, thus we can directly obtain $\mathbf{f}_{i_v}^{(a)},v\in [0:r),a\in [l_1:l_0)$ from \eqref{eqn MDS for resultant code C3} by means of the MDS property of the code $\mathbbmss{C}_0$.
		
ii) Suppose that the data $\mathbf{f}_{i_v}^{(a)}, v\in [0:r)$, $a\in [l_z:l_0)$ have been obtained for some $z\in [1:m)$, then for $i\in \{i_0,i_1,\cdots,i_{r-1}\}\cap \mathcal{G}$, we can compute $\mathbf{P}_{t,i}^{(a)}$ ($a\in [l_{z+1}:l_z)$) from $\mathbf{f}_i^{(l_z)},\mathbf{f}_i^{(l_z+1)},\cdots,\mathbf{f}_i^{(l_0-1)}$ according to P0. That is, the second  term on left hand side (LHS) of \eqref{eqn MDS for resultant code C3} is known, then we are able to solve  $\mathbf{f}_{i_v}^{(a)}$, $a\in [l_{z+1}:l_z)$, $v\in [0:r)$ by means of the MDS property of the code $\mathbbmss{C}_0$.

By i) and ii), we thus can reconstruct $\mathbf{f}_{i_0}^{(a)},\mathbf{f}_{i_1}^{(a)},\cdots,\mathbf{f}_{i_{r-1}}^{(a)}$ for all $a\in [0:l_0)$, i.e., the data stored at the $r$  nodes.
	\end{proof}
	
	\subsection{The Precise Form of Appended-data}
	In this subsection, we first introduce two sets and analyze their properties, by which we further give the precise form of the appended-data $\mathbf{P}_{t,j}^{(a)}$.
	
	For a given $u\in [0:\delta_0)$, define an $\alpha N'\times \alpha N$ matrix  as 
	\begin{eqnarray}\label{eqn the matrix Phi}
		\Phi_{\alpha,u}=\mbox{blkdiag}(\Delta_u,\Delta_u,\cdots,\Delta_u)_\alpha,
	\end{eqnarray}
	where $\Delta_u$ is an $N'\times N$ matrix defined by
	\begin{eqnarray}\label{eqn the matrix Delta}
		\Delta_u=(\mathbf{0}_{N'\times N'},\cdots,\mathbf{0}_{N'\times N'},I_{N'},\mathbf{0}_{N'\times N'},\cdots,\mathbf{0}_{N'\times N'}),u\in [0:\delta_0)
	\end{eqnarray}
	with only the $u$-th block entry being nonzero matrix. 

	For any column vector $\mathbf{f}_i^{(a)}$ of length $\alpha N$, we divide it into $\delta_0$ equal parts $\mathbf{f}_i^{(a)}[0],\mathbf{f}_i^{(a)}[1],\cdots,\mathbf{f}_i^{(a)}[\delta_0-1]$, i.e., 
	\begin{eqnarray}\label{eqn vector b[u]}
		\mathbf{f}_i^{(a)}[u]=\Phi_{\alpha,u}\mathbf{f}_i^{(a)} \textrm{ for } u\in [0:\delta_0),i\in [0:n) \textrm{ and }a\in[0:l_0),
	\end{eqnarray}
	where $\mathbf{f}_i^{(a)}[u]$ is a column vector of length $\alpha N'$.
	
	For any two column vectors $\mathbf{f}_i^{(a)}[u]$ and $\mathbf{f}_i^{(b)}[v]$,  we say that $\mathbf{f}_i^{(a)}[u]\prec\mathbf{f}_i^{(b)}[v]$ if $a<b$ or $a=b$ and $u<v$, where $i\in [0:n)$, $a,b\in [0:l_0)$ and $u,v\in [0:\delta_0)$. For any $i\in \mathcal{G}$ and $j\in [1:m)$, define $\mathcal{P}_{i,j}$ as an ordered set with the set elements drawing from $\mathbf{f}_i^{(a)}[u]$, $u\in [0:\delta_0)$,  $a\in[0:l_0)$ and placed in ascending order w.r.t. $\prec$, which are generated through the following Algorithm \ref{Alg the set Pij}.  
	\begin{algorithm}[htbp]
		\caption{The way to generate set $\mathcal{P}_{i,j}$, $i\in \mathcal{G}$ and $j\in [1:m)$} \label{Alg the set Pij}
		\begin{algorithmic}[1]
			\Ensure $\mathcal{P}_{i,j}$, $i\in \mathcal{G},j\in [1:m)$, whose elements are column vectors of length $\alpha N'$.
			\For{$j=1$; $j<m$; $j++$}
			\State Set $\mathcal{P}_{i,j}=\{\mathbf{f}_{i}^{(a)}[u]|a\in [l_j:l_{j-1}),u\in [0:\delta_0)\}$;
			\EndFor
			\State Dividing $\mathcal{P}_{i,1}$ into $l_1$ disjoint subsets $\mathcal{P}_{i,1}^{(0)},\mathcal{P}_{i,1}^{(1)},\cdots,\mathcal{P}_{i,1}^{(l_1-1)}$ of equal size.
			\For{$j=2$; $j<m$; $j++$}
			\State $\mathcal{P}_{i,j}:=\mathcal{P}_{i,j}\cup\bigcup\limits_{a=l_j}^{l_{j-1}-1}(\mathcal{P}_{i,1}^{(a)}\cup\cdots\cup\mathcal{P}_{i,j-1}^{(a)})$;
			\State Dividing $\mathcal{P}_{i,j}$ into $l_j$ disjoint subsets $\mathcal{P}_{i,j}^{(0)},\mathcal{P}_{i,j}^{(1)},\cdots,\mathcal{P}_{i,j}^{(l_j-1)}$ of equal size.
			\EndFor
		\end{algorithmic}  	
	\end{algorithm}

Strictly speaking, to ensure that Algorithm \ref{Alg the set Pij} is valid, one needs $\left|\mathcal{P}_{i,j}\right |=(\delta_j-\delta_{j-1})l_j$ for $i\in \mathcal{G}$ and $j\in [1:m)$, which will be shown in P3.

		\begin{Example}\label{Exp set Piz for (2,3)}
Based on the  base code in Example  \ref{Exp set Piz for (1,2,3)}, suppose the goal is to obtain a $(16,10)$ MDS array code $\mathbbmss{C}_2$ having $\{2,3\}$-optimal repair property for the first two nodes, i.e., $\rho=2$ and $\mathcal{G}=[0:2)$. Here $\alpha=1,m=2$, $\delta_0=2,\delta_1=3$, then $l_0=3,l_1=2$ by \eqref{eqn notation l}. By means of  Algorithm \ref{Alg the set Pij}, the sets $\mathcal{P}_{i,j}$ and $\mathcal{P}_{i,j}^{(a)}$, $j\in [1:2)$, $a\in [0:2)$ corresponding to GN $i\in \mathcal{G}$ are
			\begin{eqnarray*}
				\mathcal{P}_{i,1}=\{\mathbf{f}_i^{(2)}[0],\mathbf{f}_i^{(2)}[1]\},\,\,\mathcal{P}_{i,1}^{(0)}=\{\mathbf{f}_i^{(2)}[0]\},\,\,\mathcal{P}_{i,1}^{(1)}=\{\mathbf{f}_i^{(2)}[1]\}.
			\end{eqnarray*}
		\end{Example}	
	\begin{Example}\label{Exp set Piz}
		
		Based on the  base code in Example  \ref{Exp set Piz for (1,2,3)}, suppose the goal is to obtain a $(16,10)$ MDS array code $\mathbbmss{C}_2$ with $\{2,3,4,6\}$-optimal repair property for the first two nodes, i.e., $\rho=2$ and $\mathcal{G}=[0:2)$.
		In this case, $\alpha=1,m=4$, $\delta_0=2,\delta_1=3,\delta_2=4,\delta_4=6$, then $l_0=6,l_1=4,l_2=3,l_3=2$ by \eqref{eqn notation l}.
		By means of  Algorithm \ref{Alg the set Pij}, the sets $\mathcal{P}_{i,j}$ and $\mathcal{P}_{i,j}^{(a)}$, $j\in [1:4)$, $a\in [0:4)$ of GN $i\in \mathcal{G}$ are given in Table \ref{table sets of P} and Table \ref{table sets of P'} respectively.
		\begin{table}[htbp]
			\begin{center}
				\caption{The sets $\mathcal{P}_{i,j}$  of GN $i$ for the code in Example \ref{Exp set Piz}, where $j\in [1:4)$}\label{table sets of P}
				\begin{tabular}{|c|c|c|c|}
					\hline $j$& 1 & 2 & 3 \\
					\hline $\mathcal{P}_{i,j}$ & $\{\mathbf{f}_{i}^{(4)}[0],\mathbf{f}_{i}^{(4)}[1],\mathbf{f}_{i}^{(5)}[0],\mathbf{f}_{i}^{(5)}[1]\}$ & $\{\mathbf{f}_{i}^{(3)}[0],\mathbf{f}_i^{(3)}[1],\mathbf{f}_i^{(5)}[1]\}$ & $\{\mathbf{f}_i^{(2)}[0],\mathbf{f}_i^{(2)}[1],\mathbf{f}_i^{(5)}[0],\mathbf{f}_i^{(5)}[1]\}$\\
					\hline
				\end{tabular}
			\end{center}
		\end{table}	
		\begin{table}[htbp]
			\begin{center}
				\caption{The sets $\mathcal{P}_{i,j}^{(a)}$ of GN $i$ for the code in Example \ref{Exp set Piz}, where $j\in [1:4)$, $a\in [0:4)$ and $a<l_j$}\label{table sets of P'}
				\begin{tabular}{|c|c|c|c|}
					\hline \diagbox{$a$}{$j$}& 1 & 2 & 3 \\
					\hline 0 & $\{\mathbf{f}_{i}^{(4)}[0]\}$ & $\{\mathbf{f}_{i}^{(3)}[0]\}$ & $\{\mathbf{f}_i^{(2)}[0],\mathbf{f}_i^{(2)}[1]\}$\\
					\hline 1 & $\{\mathbf{f}_{i}^{(4)}[1]\}$ & $\{\mathbf{f}_{i}^{(3)}[1]\}$ & $\{\mathbf{f}_i^{(5)}[0],\mathbf{f}_i^{(5)}[1]\}$\\
					\hline 2 & $\{\mathbf{f}_{i}^{(5)}[0]\}$ & $\{\mathbf{f}_{i}^{(5)}[1]\}$ &\\
					\hline 3 & $\{\mathbf{f}_{i}^{(5)}[1]\}$ & & \\
					\hline
				\end{tabular}
			\end{center}
		\end{table}
	\end{Example}

	According to  Algorithm \ref{Alg the set Pij}, we have the following  properties, whose proofs are given in Appendix \ref{appen the proof of property 1}.

	\begin{Property}\label{Pro the number of elements of Piz}
		Given $i\in \mathcal{G}$ and  $j\in [1:m)$,
		\begin{itemize}
			\item [P1.]  $\mathcal{P}_{i,1}\cup\mathcal{P}_{i,2}\cup \cdots \cup \mathcal{P}_{i,j}=\bigcup\limits_{a=0}^{l_j-1}(\mathcal{P}_{i,1}^{(a)}\cup\mathcal{P}_{i,2}^{(a)}\cup\cdots\cup\mathcal{P}_{i,j}^{(a)})=\{\mathbf{f}_i^{(a)}[u]|a\in  [l_j:l_0),u\in [0:\delta_0)\}$; 
			\item [P2.] When $j\ge 2$, 
			$\mathcal{P}_{i,j} \subseteq \{\mathbf{f}_i^{(a)}[u]|a\in [l_j:l_z),u \in [0:\delta_0)\} \cup \bigcup\limits_{a=l_{j}}^{l_z-1}(\mathcal{P}_{i,1}^{(a)}\cup \cdots \cup \mathcal{P}_{i,z}^{(a)})$ for all $z\in [1:j)$; 
			\item [P3.]  $|\mathcal{P}_{i,j} |=(\delta_j-\delta_{j-1})l_j$ and $|\mathcal{P}_{i,j}^{(a)} |=\delta_j-\delta_{j-1}$  for $a\in [0:l_j)$.
		\end{itemize}
	\end{Property}

	Now, we present the precise form of appended-data based on the sets $\mathcal{P}_{i,j}^{(a)}$ for $i\in \mathcal{G},1\le j<m,a\in [0:l_j)$. For convenience of notation, we also denote $\mathcal{P}_{i,j}^{(a)}$, $i\in \mathcal{G},1\le j<m,a\in [0:l_j)$ the column vector of length $(\delta_j-\delta_{j-1})\alpha N'$, which is formed by its elements in ascending order. Then for $i\in \mathcal{G}$, $t\in [0:r)$ and $a\in [0:l_0)$, we define $\mathbf{P}_{t,i}^{(a)}$ as
	\begin{eqnarray}\label{eqn_specific_form_of_P}
		\mathbf{P}_{t,i}^{(a)}=\left\{
		\begin{array}{ll}
			\sum\limits_{j=1}^{w}(K_{t,i,\delta_{j-1}-\delta_0},K_{t,i,\delta_{j-1}-\delta_0+1},\cdots,K_{t,i,\delta_j-\delta_0-1})\mathcal{P}_{i,j}^{(a)}, & \mathrm{if~}a\in [l_{w+1}:l_w),w\in [1:m),\\
			\mathbf{0}, &\mathrm{otherwise},
		\end{array}
		\right.
	\end{eqnarray}
	where the $\alpha N \times \alpha N'$ matrices  $K_{t,i,v}$ for $t\in [0:r),i\in \mathcal{G},v\in [0:\delta_{m-1}-\delta_0)$ are called  \textit{key matrcies} 
	of node $i$. According to P1, for $i\in \mathcal{G}$, $t\in [0:r)$, given $z\in [1:m)$ and $a\in [l_{z+1}:l_z)$, the appended-data $\mathbf{P}_{t,i}^{(a)}$ defined by \eqref{eqn_specific_form_of_P} is a linear combination of $\mathbf{f}_i^{(l_z)},\mathbf{f}_i^{(l_z+1)},\cdots,\mathbf{f}_i^{(l_0-1)}$, i.e., P0 holds. That is, $\mathbf{P}_{t,i}^{(a)}$  is well defined for \eqref{eqn t-PCG of code C3}.
	
	
	The following two examples illustrate the whole process of our method.
		\begin{Example}\label{Exp (16,10) MDS array code C3 (2,3)}
			Following up from Example \ref{Exp set Piz for (2,3)}, by applying the generic construction method, we can obtain a $(16,10)$ MDS array code $\mathbbmss{C}_2$ with $\{2,3\}$-optimal repair property for the first two nodes, which is defined by the following parity-check equations:
			\begin{eqnarray*}\label{eqn t-th PCG of (16,10) MDS array code C3 (2,3)}
				\left(\hspace{-2mm}\begin{array}{l}
					A_{t,0}\mathbf{f}_0^{(0)}+\zeta_0^tV_{0,0}^\top\mathbf{f}_{0}^{(2)}[0]\\		
					A_{t,0}\mathbf{f}_0^{(1)}+\zeta_0^tV_{0,0}^\top\mathbf{f}_{0}^{(2)}[1]\\
					A_{t,0}\mathbf{f}_0^{(2)}
				\end{array}\hspace{-2mm}\right)+\left(\hspace{-2mm}\begin{array}{l}
					A_{t,1}\mathbf{f}_1^{(0)}+\zeta_0^tV_{0,1}^\top\mathbf{f}_{1}^{(2)}[0]\\		
					A_{t,1}\mathbf{f}_1^{(1)}+\zeta_0^tV_{0,1}^\top\mathbf{f}_{1}^{(2)}[1]\\
					A_{t,1}\mathbf{f}_1^{(2)}
				\end{array}\hspace{-2mm}\right)+\sum\limits_{j=2}^{15}\left(\hspace{-2mm}\begin{array}{c}
					A_{t,j}\mathbf{f}_j^{(0)}\\
					A_{t,j}\mathbf{f}_j^{(1)}\\
					A_{t,j}\mathbf{f}_j^{(2)}
				\end{array}\hspace{-2mm}\right)=\mathbf{0},&t\in [0:6),
			\end{eqnarray*}
		where the $N\times \frac{N}{2}$ key matrices are
			\begin{eqnarray*}
				K_{t,0,0}=\zeta_0^tV_{0,0}^\top,K_{t,1,0}=\zeta_0^tV_{0,1}^\top,&t\in [0:r)
			\end{eqnarray*}
for $\zeta_0\in \mathbb{F}_q$, and the appended-data are
\begin{equation*}
\mathbf{P}_{t,i}^{(0)}= \zeta_0^tV_{0,i}^\top\mathbf{f}_{i}^{(2)}[0], ~\mathbf{P}_{t,i}^{(1)}= \zeta_0^tV_{0,i}^\top\mathbf{f}_{i}^{(2)}[1] \mbox{~~for~~}i=0,1.
\end{equation*}
		\end{Example}
	
	\begin{Example}\label{Exp (16,10) MDS array code C3}
		Following up from Example \ref{Exp set Piz}, through the generic construction method, we can obtain a $(16,10)$ MDS array code $\mathbbmss{C}_2$ with $\{2,3,4,6\}$-optimal repair property for the first two nodes, which is defined by the following parity-check equations:
		\begin{eqnarray}\label{eqn t-th PCG of (16,10) MDS array code C3}
			\sum\limits_{i=0}^{1}\left(\hspace{-2mm}\begin{array}{l}
				A_{t,i}\mathbf{f}_i^{(0)}+\zeta_0^tV_{0,i}^\top\mathbf{f}_{i}^{(4)}[0]+\zeta_1^tV_{0,i}^\top\mathbf{f}_{i}^{(3)}[0]+\zeta_2^tV_{0,i}^\top\mathbf{f}_{i}^{(2)}[0]+ \zeta_3^tV_{0,i}^\top\mathbf{f}_{i}^{(2)}[1]\\		
				A_{t,i}\mathbf{f}_i^{(1)}+\zeta_0^tV_{0,i}^\top\mathbf{f}_{i}^{(4)}[1]+\zeta_1^tV_{0,i}^\top\mathbf{f}_{i}^{(3)}[1]+\zeta_2^tV_{0,i}^\top\mathbf{f}_{i}^{(5)}[0]+ \zeta_3^tV_{0,i}^\top\mathbf{f}_{i}^{(5)}[1]\\
				A_{t,i}\mathbf{f}_i^{(2)}+\zeta_0^tV_{0,i}^\top\mathbf{f}_{i}^{(5)}[0]+\zeta_1^tV_{0,i}^\top\mathbf{f}_{i}^{(5)}[1]\\
				A_{t,i}\mathbf{f}_i^{(3)}+\zeta_0^tV_{0,i}^\top\mathbf{f}_{i}^{(5)}[1]\\
				A_{t,i}\mathbf{f}_i^{(4)}\\
				A_{t,i}\mathbf{f}_i^{(5)}
			\end{array}\hspace{-2mm}\right)+\sum\limits_{j=2}^{15}\left(\hspace{-2mm}\begin{array}{c}
				A_{t,j}\mathbf{f}_j^{(0)}\\
				A_{t,j}\mathbf{f}_j^{(1)}\\
				A_{t,j}\mathbf{f}_j^{(2)}\\
				A_{t,j}\mathbf{f}_j^{(3)}\\
				A_{t,j}\mathbf{f}_j^{(4)}\\
				A_{t,j}\mathbf{f}_j^{(5)}\\
			\end{array}\hspace{-2mm}\right)=\mathbf{0}
		\end{eqnarray}
		for $t \in [0:6)$, where the $N\times \frac{N}{2}$ key matrices are 
		\begin{eqnarray*}
				K_{t,0,v}=\zeta_v^tV_{0,0}^\top,K_{t,1,0}=\zeta_v^tV_{0,1}^\top,&t\in [0:r),v\in [0:4),
			\end{eqnarray*}	
with $\zeta_0,\zeta_1,\zeta_2$ and $\zeta_3$ being four distinct elements in $\mathbb{F}_q$.
		\end{Example}
	
\subsection{Repair Property}
	In this subsection, we show that GN $i$ possesses the $\delta_{[0:m)}$-optimal repair property and RN $j$ maintains the same optimal repair property as that of base code for all $i\in \mathcal{G}$ and $j\in [0:n)\backslash\mathcal{G}$.  Particularly,
	if node $i$ in the base code $\mathbbmss{C}_0$ has $\delta_z$-optimal repair property for $z\in [0:m)$, then let 
the   $\frac{\alpha N}{\delta_z}\times \alpha N$ full-rank matrices $R_{i,\delta_z}$ and $S_{i,\delta_z}$ denote the $\delta_z$-repair matrix and $\delta_z$-select matrix, respectively.

	Consider the repair of node $i\in [0:n)$ by connecting $d_z=k+\delta_z-1$ surviving nodes where $z\in [0:m)$.  Let the $\frac{l_0\alpha N}{\delta_z}\times l_0\alpha N$ full-rank matrices $R_{i,\delta_z}'$ and $S_{i,\delta_z}'$ respectively be the $\delta_z$-repair matrix  and $\delta_z$-select matrix of node $i$ of code $\mathbbmss{C}_2$  given by
	\begin{eqnarray}\label{eqn the repair,select matrix of C3 for goal nodes}
		R_{i,\delta_z}'=\left\{\begin{array}{ll}
			\left(\begin{array}{c;{2pt/2pt}c}
				\underbrace{\begin{array}{ccc}
						R_{i,\delta_0} & &\\
						& \ddots &\\
						& & R_{i,\delta_0}
				\end{array}}_{l_z\times l_z} & \underbrace{\begin{array}{ccc}
						\mathbf{0}_{\alpha N'\times \alpha N} & \cdots & \mathbf{0}_{\alpha N'\times \alpha N}\\
						\vdots & \ddots & \vdots\\
						\mathbf{0}_{\alpha N'\times \alpha N} & \cdots & \mathbf{0}_{\alpha N'\times \alpha N}
				\end{array}}_{l_z\times (l_0-l_z)}
			\end{array}\right), & \textrm{if~}i\in\mathcal{G},\\
			\mbox{blkdiag}(R_{i,\delta_z},R_{i,\delta_z},\cdots,R_{i,\delta_z})_{l_0}, & \textrm{otherwise},
		\end{array}\right.
	\end{eqnarray}	
	and
	\begin{eqnarray}\label{eqn repair,select matrix of remainder nodes}		
		S_{i,\delta_z}'=\left\{\begin{array}{ll}
			\left(\begin{array}{c;{2pt/2pt}c}
				\underbrace{\begin{array}{ccc}
						S_{i,\delta_0} & &\\
						& \ddots &\\
						& & S_{i,\delta_0}
				\end{array}}_{l_z\times l_z} & \underbrace{\begin{array}{ccc}
						\mathbf{0}_{\alpha N'\times \alpha N} & \cdots & \mathbf{0}_{\alpha N'\times \alpha N}\\
						\vdots & \ddots & \vdots\\
						\mathbf{0}_{\alpha N'\times \alpha N} & \cdots & \mathbf{0}_{\alpha N'\times \alpha N}
				\end{array}}_{l_z\times (l_0-l_z)}
			\end{array}\right),& \textrm{if~}i\in\mathcal{G},\\
			\mbox{blkdiag}(S_{i,\delta_z},S_{i,\delta_z},\cdots,S_{i,\delta_z})_{l_0}, & \textrm{otherwise}.
		\end{array}\right.
	\end{eqnarray}

	In other words, when $d_z$ surviving nodes are connected, we use the equations obtained by multiplying $S_{i,\delta_s}$ on both sides of the equations in \eqref{eqn t-PCG of code C3} to recover the data stored at node $i$, i.e.,
	\begin{eqnarray*}\label{Eqn the product between Sidz and PCG of C3}
		S_{i,\delta_s}A_{t,i}\mathbf{f}_i^{(a)}+\sum\limits_{j=0,j\ne i}^{n-1}S_{i,\delta_s}A_{t,j}\mathbf{f}_j^{(a)}+\sum\limits_{j\in \mathcal{G}}S_{i,\delta_s}\mathbf{P}_{t,j}^{(a)}=\mathbf{0},& t\in [0:r)
	\end{eqnarray*}
	where $s=0,a\in [0:l_z)$ if $i\in \mathcal{G}$ and $s=z,a\in [0:l_0)$ otherwise. Substituting \eqref{repair node requirement3} into the above LSEs, we then get
	\begin{eqnarray}\label{eqn specified form of a-SLE}
		S_{i,\delta_s}A_{t,i}\mathbf{f}_i^{(a)}+\sum\limits_{j=0,j\ne i}^{n-1}\tilde{A}_{t,j,i,\delta_s}R_{i,\delta_s}\mathbf{f}_j^{(a)}+\sum\limits_{j\in \mathcal{G}}S_{i,\delta_s}\mathbf{P}_{t,j}^{(a)}=\mathbf{0}, &t\in [0:r),
	\end{eqnarray}
	where $\tilde{A}_{t,j,i,\delta_s}$ is an $\frac{\alpha N}{\delta_s}\times \frac{\alpha N}{\delta_s}$ matrix defined in \eqref{repair node requirement3}.
	
	First of all, we propose the repair procedure of GNs. To this end, node $i\in \mathcal{G}$ in base code $\mathbbmss{C}_0$ is required to satisfy  the following conditions.
\begin{itemize}
			\item [C1.] $S_{j,\delta_0}K_{t,i,v}=\mathbf{0}$ for any $i,j\in \mathcal{G}$ with $i\ne j$, $t\in [0:r)$ and $v\in [0:\delta_{m-1}-\delta_0)$;
			\item [C2.] For any $i\in \mathcal{G}$, $z\in [1:m)$ and $\mathcal{D}_z=\{j_0,j_1,\cdots,j_{r-\delta_z-1}\}\subset [0:n)\backslash\{i\}$, the $r\alpha N'\times r\alpha N'$ matrix
			\begin{eqnarray} \label{eqn matrix MijDz}
				M_{i,\mathcal{D}_z}=\left(\begin{array}{c;{2pt/2pt}ccc;{2pt/2pt}ccc}
					S_{i,\delta_0}A_{0,i}  & \tilde{A}_{0,j_0,i,\delta_0} & \cdots & \tilde{A}_{0,j_{r-\delta_z-1},i,\delta_0} & S_{i,\delta_0}K_{0,i,0} & \cdots & S_{i,\delta_0}K_{0,i,\delta_z-\delta_0-1} \\
					S_{i,\delta_0}A_{1,i}  & \tilde{A}_{1,j_0,i,\delta_0} & \cdots & \tilde{A}_{1,j_{r-\delta_z-1},i,\delta_0} & S_{i,\delta_0}K_{1,i,0} & \cdots & S_{i,\delta_0}K_{1,i,\delta_z-\delta_0-1}\\
					\vdots & \vdots & \ddots & \vdots & \vdots & \ddots & \vdots\\
					S_{i,\delta_0}A_{r-1,i}  & \tilde{A}_{r-1,j_0,i,\delta_0} & \cdots & \tilde{A}_{r-1,j_{r-\delta_z-1},i,\delta_0} & S_{i,\delta_0}K_{r-1,i,0} & \cdots & S_{i,\delta_0}K_{r-1,i,\delta_z-\delta_0-1}
				\end{array}\right)
			\end{eqnarray}
			is nonsingular over $\mathbb{F}_q$, 
\end{itemize}
where $K_{t,i,v}$ are key matrices of node $i$. Particularly, we also define a $r\alpha N'\times r\alpha N'$ matrix $M_{i,\mathcal{D}_0}$ as in \eqref{eqn matrix MijDz} with $z=0$.

	\textit{\textbf{The Repair Procedure of GNs}:} Assume that GN $i~(i\in\mathcal{G})$ fails and $d_z=k+\delta_z-1$ helper nodes are connected for any given $z\in [0:m)$, then node $i$ is repaired as follows:
	\begin{itemize}
		\item [1)] Download the data $\{R_{i,\delta_0}\mathbf{f}_j^{(b)}|b\in [0:l_z)\}$ from each helper node $j\in \mathcal{H}_z$, where $\mathcal{H}_z$ denotes the set of indices of the $d_z$ helper nodes.
		\item [2)] Choose linear system of equations \eqref{eqn specified form of a-SLE} for $s=0$ and $a=0,1,\cdots,l_z-1$ to solve the data stored at GN $i$, i.e.,
		\begin{eqnarray}\label{eqn a-SLE by using (6)}
			S_{i,\delta_0}A_{t,i}\mathbf{f}_i^{(a)}+\sum\limits_{j=0,j\ne i}^{n-1}\tilde{A}_{t,j,i,\delta_0}R_{i,\delta_0}\mathbf{f}_j^{(a)}+\sum\limits_{j\in \mathcal{G}}S_{i,\delta_0}\mathbf{P}_{t,j}^{(a)}=\mathbf{0}, & a\in [l_{w+1}:l_w),t\in[0:r)
		\end{eqnarray}
		for all $w\in [z:m)$ by noting $$\{0,1,\cdots,l_z-1\}=\bigcup\limits_{w=z}^{m-1}[l_{w+1}:l_w).$$ According to \eqref{eqn_specific_form_of_P}, by first applying  C1 and then substituting the downloaded data into \eqref{eqn a-SLE by using (6)}, we then obtain
		\begin{align}\label{eqn reduction form of a-SLE}
			\hspace{-1cm}\left(\hspace{-2mm}\begin{array}{c}
				S_{i,\delta_0}A_{0,i}\\
				S_{i,\delta_0}A_{1,i}\\
				\vdots\\
				S_{i,\delta_0}A_{r-1,i}
			\end{array}\hspace{-2mm}\right)\mathbf{f}_i^{(a)}+\sum\limits_{j\in \mathcal{D}_z}\left(\hspace{-2mm}\begin{array}{c}
				\tilde{A}_{0,j,i,\delta_0}\\
				\tilde{A}_{1,j,i,\delta_0}\\
				\vdots\\
				\tilde{A}_{r-1,j,i,\delta_0}
			\end{array}\hspace{-2mm}\right)R_{i,\delta_0}\mathbf{f}_{j_v}^{(a)}+\sum\limits_{u=1}^{z}\Gamma_u\mathcal{P}_{i,u}^{(a)}
			=-\sum\limits_{u=z+1}^{w}\Gamma_u\mathcal{P}_{i,u}^{(a)}+*, a\in [l_{w+1}:l_w)
		\end{align}
		for all $w\in [z:m)$, where $\mathcal{D}_z=[0:n)\backslash(\mathcal{H}_z\cup\{i\})$, symbol $*$ denotes a known vector that can be determined by the downloaded data in 1), here $*$ denotes
		\begin{eqnarray*}\label{Eqn the symbol *}
			-\sum\limits_{j\in\mathcal{H}_z}\left(\hspace{-2mm}\begin{array}{c}
				\tilde{A}_{0,j,i,\delta_0}\\
				\tilde{A}_{1,j,i,\delta_0}\\
				\vdots\\
				\tilde{A}_{r-1,j,i,\delta_0}
			\end{array}\hspace{-2mm}\right)R_{i,\delta_0}\mathbf{f}_{j}^{(a)},
		\end{eqnarray*}
		and the $r\alpha N'\times (\delta_u-\delta_{u-1})\alpha N'$ matrix 
		\begin{eqnarray*}
			\Gamma_u=\left(\begin{array}{cccc}
				S_{i,\delta_0}K_{0,i,\delta_{u-1}-\delta_0} & S_{i,\delta_0}K_{0,i,\delta_{u-1}-\delta_0+1} & \cdots &  S_{i,\delta_0}K_{0,i,\delta_{u}-\delta_0-1}\\
				S_{i,\delta_0}K_{1,i,\delta_{u-1}-\delta_0} & S_{i,\delta_0}K_{1,i,\delta_{u-1}-\delta_0+1} & \cdots &  S_{i,\delta_0}K_{1,i,\delta_{u}-\delta_0-1}\\
				\vdots & \vdots & \ddots & \vdots \\
				S_{i,\delta_0}K_{r-1,i,\delta_{u-1}-\delta_0} & S_{i,\delta_0}K_{r-1,i,\delta_{u-1}-\delta_0+1} & \cdots &  S_{i,\delta_0}K_{r-1,i,\delta_{u}-\delta_0-1}\\
			\end{array}\hspace{-2mm}\right),\,\,u\in [1:m).
		\end{eqnarray*}
		Let $\mathcal{D}_z=\{j_0,j_1,\cdots,j_{r-\delta_z-1}\}$, in matrix form, we can rewritten \eqref{eqn reduction form of a-SLE} as
		\begin{eqnarray}\label{eqn reduction matrix form of a-SLE}
		    M_{i,\mathcal{D}_z}\cdot\left(\begin{array}{c}
		         \mathbf{f}_i^{(a)}  \\
		         R_{i,\delta_0}\mathbf{f}_{j_0}^{(a)}\\
		         \vdots\\
		         R_{i,\delta_0}\mathbf{f}_{j_{r-\delta_z-1}}^{(a)}\\
		         \mathcal{P}_{i,1}^{(a)}\\
		         \vdots\\
		         \mathcal{P}_{i,z}^{(a)}\\
		    \end{array}\right)=-\sum\limits_{u=z+1}^{w}\Gamma_u\mathcal{P}_{i,u}^{(a)}+*, a\in [l_{w+1}:l_w),
		\end{eqnarray}
		where $M_{i,\mathcal{D}_z}$ is defined in \eqref{eqn matrix MijDz}.
		\item [3)] Recover the data stored at GN $i$ by sequentially solving \eqref{eqn reduction matrix form of a-SLE} when $w=z,z+1,\cdots,m-1$,
		\begin{itemize}
			\item [3-1)] Compute the first term on RHS of \eqref{eqn reduction matrix form of a-SLE} from the recovered data of GN $i$ if $w>z$.
			\item [3-2)] Recover the data $\mathbf{f}_i^{(a)}$ and $\mathcal{P}_{i,u}^{(a)}$ for $1\le u\le z$ from \eqref{eqn reduction matrix form of a-SLE}.
		\end{itemize}
	\end{itemize}
	
	\begin{Theorem}\label{Thr_repair_for_goal_node}
		GN $i\in \mathcal{G}$ of the new code $\mathbbmss{C}_2$ has the $\delta_{[0:m)}$-optimal repair/access property  if 
		RN $i$ of base code $\mathbbmss{C}_0$ satisfies C1 and C2.
	\end{Theorem}
	\begin{proof}
	   Let us consider the $\delta_z$-optimal repair property of GN $i$ for any $z\in [0:n)$, i.e., $d_z=k+\delta_z-1$ helper are connected. As shown in 1) of The Repair Procedure of GNs,  $\gamma(d_z)=\frac{d_z\cdot l_z\cdot \alpha N}{\delta_0}=\frac{d_z}{\delta_z}l_0\alpha N=\frac{d_z}{d_z-k+1}l_0\alpha N$ since $\mbox{Rank}(R_{i,\delta_0})=\frac{\alpha N}{\delta_0}$, $\delta_z=d_z-k+1$ and $\frac{l_z}{\delta_0}=\frac{l_0}{\delta_z}$ from \eqref{eqn notation l}, which attains the lower bound in \eqref{Eqn_bound_on_gamma}. Moreover, if node $i$ has the $\delta_0$-optimal access property for base code $\mathbbmss{C}_0$, i.e., the repair matrix $R_{i,\delta_0}$  has  only one nonzero element in each row, then by \eqref{eqn the repair,select matrix of C3 for goal nodes}  GN $i$ has the $\delta_z$-optimal access property in the new code $\mathbbmss{C}_2$. Thus, to prove this theorem, it suffices to show that 3) of The Repair Procedure of GNs can be executed for $w=z,z+1,\cdots,m-1$.
		
		For fixed $a$, according to P3, it is easy to see that there are
		\begin{eqnarray*}
			\alpha N+(r-\delta_z)\alpha N'+\sum\limits_{u=1}^{z}(\delta_u-\delta_{u-1})\alpha N'=\alpha\cdot \delta_0 N'+(r-\delta_0)\alpha N'=r\alpha N'
		\end{eqnarray*}
		unknown variables on LHS of the $r \alpha N'$ equations in \eqref{eqn reduction matrix form of a-SLE}. Note that the coefficient matrix on LHS of \eqref{eqn reduction matrix form of a-SLE} is nonsingular according to C2 if $z>0$ and Lemma \ref{lem the requirement of obtaining fi} if $z=0$ (the $\delta_0$-optimal repair property of code $\mathbbmss{C}_0$). That is, \eqref{eqn reduction matrix form of a-SLE} is solvable if the first term in its RHS  is known, i.e., 3-2) of The Repair Procedure of GNs can be executed. Then we only need to show the following claim.
		\begin{itemize}
			\item [\textbf{Claim:}] For any given $w\in [z:m)$, the first term $\sum\limits_{u=z+1}^{w}\Gamma_u\mathcal{P}_{i,u}^{(a)}$ on RHS of \eqref{eqn reduction matrix form of a-SLE} can be determined for all $a\in [l_{w+1}:l_w)$.
		\end{itemize}
We prove it by induction.

i) When $w=z$, Claim is obvious as $\sum\limits_{u=z+1}^{w}\Gamma_u\mathcal{P}_{i,u}^{(a)}=\mathbf{0}$. 

ii) Assume that Claim holds for all $w\in [z: v]$ where $z\le v<m-1$. Then, $\mathbf{f}_i^{(a)}$ and  $\mathcal{P}_{i,u}^{(a)}$ are available for all $1\le u\le z$ and $a\in [l_{v+1}:l_z)$, which together with P2 imply that the first term on RHS of \eqref{eqn reduction matrix form of a-SLE} for $w=v+1$ is already known. That is, Claim holds for $w=v+1$  and thus for all $z\le w<m$ by the induction. 

Then for $w\in [z:m)$,  3) of The Repair Procedure of GNs can be executed to the end, which means that the data $\mathbf{f}_i^{(a)}$ and $\mathcal{P}_{i,u}^{(a)}$ for all $1\le u\le z$ and $a\in [0:l_z)$ can be recovered. Finally by P1, we already have $\mathbf{f}_i^{(a)}$, $a\in [l_z:l_0)$ from the data in set $\bigcup\limits_{a=0}^{l_z-1}(\mathcal{P}_{i,1}^{(a)}\cup\cdots\cup\mathcal{P}_{i,z}^{(a)})$ if $z>0$, i.e., all the data stored at GN $i$ are regenerated, which finishes the proof.
	\end{proof}

	Example \ref{Exa_Repair} serves to visualize the ideas behind The Repair Procedure of GNs.
	
	\begin{Example}\label{Exa_Repair}
		Following up from Example \ref{Exp (16,10) MDS array code C3}, let us consider the repair of the first GN of the $(16,10)$ MDS array code $\mathbbmss{C}_2$ by connecting to $d_1=12$ helper nodes, i.e., we investigate the first node has $3$-optimal repair property.

		By using \eqref{eqn t-th PCG of (16,10) MDS array code C3} and \eqref{eqn reduction matrix form of a-SLE}, the procedure of repairing the first GN is shown in Table \ref{Table repair of goal node}.
		To save space, we only give unknown variables related to the first GN in \eqref{eqn reduction matrix form of a-SLE} for $a=0,1,2,3$, and show how to obtain the data stored at the first GN in 3-2) of The Repair Procedure of GNs.
		
		\begin{table*}[htbp]
			\centering
			\caption{The procedure of repairing the first node of the $(16,10)$ MDS array code $\mathbbmss{C}_2$ given in Example \ref{Exp (16,10) MDS array code C3} by connecting $12$ helper nodes}.\label{Table repair of goal node}
			\begin{tabular}{|c|c|c|c|c|c|}
				\hline  $w$ & $a$ & The unknown variables & The eliminated variables & The solved variables\\
				\hline  1& 3 & $\mathbf{f}_0^{(3)},\mathbf{f}_0^{(5)}[1]$ &  & $\mathbf{f}_0^{(3)},\mathbf{f}_0^{(5)}[1]$\\
				\hline  2& 2 & $\mathbf{f}_0^{(2)},\mathbf{f}_0^{(5)}[0],\mathbf{f}_0^{(5)}[1]$ & $\mathbf{f}_0^{(5)}[1]$ & $\mathbf{f}_0^{(2)},\mathbf{f}_0^{(5)}[0]$\\
				\hline  \multirow{2}{*}{3} & 1 & $\mathbf{f}_0^{(1)},\mathbf{f}_0^{(3)}[1],\mathbf{f}_0^{(4)}[1],\mathbf{f}_0^{(5)}[0],\mathbf{f}_0^{(5)}[1]$ & $\mathbf{f}_0^{(3)}[1],\mathbf{f}_0^{(5)}[0],\mathbf{f}_0^{(5)}[1]$ & $\mathbf{f}_0^{(1)},\mathbf{f}_0^{(4)}[1]$\\
				\cline{2-6}  &0 & $\mathbf{f}_0^{(0)},\mathbf{f}_0^{(3)}[0],\mathbf{f}_0^{(4)}[0],\mathbf{f}_0^{(2)}[0],\mathbf{f}_0^{(2)}[1]$ & $\mathbf{f}_0^{(3)}[0],\mathbf{f}_0^{(2)}[0],\mathbf{f}_0^{(2)}[1]$ & $\mathbf{f}_0^{(0)},\mathbf{f}_0^{(4)}[0]$\\
				\hline
			\end{tabular}
		\end{table*}
	\end{Example}

	Next, we examine the repair property of the RNs of code $\mathbbmss{C}_2$, which is the same as that of the base code.
		
	\textit{\textbf{The Repair Procedure of RNs}:} Assume that RN $i\in [0:n)\setminus\mathcal{G}$ fails and it has $\delta_z$-optimal repair property for base code $\mathbbmss{C}_0$. When $d_z=k+\delta_z-1$ helper nodes are connected to repair RN $i$, let $\mathcal{H}_z$ be the set of indices of the $d_z$ helper nodes, its stored data is repaired as follows:
	\begin{itemize}
		\item [1)] Download the data $\{R_{i,\delta_z}\mathbf{f}_j^{(b)}|b\in [0:l_0)\}$ from each helper node $j\in \mathcal{H}_z$.
		\item [2)] Choose the linearly system of equations \eqref{eqn specified form of a-SLE} for $s=z$ and $a=0,1,\cdots,l_0-1$ to solve the stored data at RN $i$, i.e.,
		\begin{eqnarray}\label{eqn LSE of repair remainder nodes}
			\left(\hspace{-2mm}\begin{array}{c}
				S_{i,\delta_z}A_{0,i}\\
				S_{i,\delta_z}A_{1,i}\\
				\vdots\\
				S_{i,\delta_z}A_{r-1,i}\\
			\end{array}\hspace{-2mm}\right)\mathbf{f}_i^{(a)}+\sum\limits_{v=0}^{r-\delta_z-1}\left(\hspace{-2mm}\begin{array}{c}
				\tilde{A}_{0,j_v,i,\delta_z}\\
				\tilde{A}_{1,j_v,i,\delta_z}\\
				\vdots\\
				\tilde{A}_{r-1,j_v,i,\delta_z}\\
			\end{array}\hspace{-2mm}\right)R_{i,\delta_z}\mathbf{f}_{j_v}^{(a)}+\sum\limits_{j\in \mathcal{G}}\left(\hspace{-2mm}\begin{array}{c}
				S_{i,\delta_z}\mathbf{P}_{0,j}^{(a)}\\
				S_{i,\delta_z}\mathbf{P}_{1,j}^{(a)}\\
				\vdots\\
				S_{i,\delta_z}\mathbf{P}_{r-1,j}^{(a)}\\
			\end{array}\hspace{-2mm}\right)=*,~a\in [l_{w+1}:l_w)
		\end{eqnarray}
		for all $w\in [0:m)$, where $\{j_0,j_1,\cdots,j_{r-\delta_z-1}\}=\mathcal{D}_z=[0:n)\backslash(\mathcal{H}_z\cup\{i\})$  and similar to \eqref{eqn reduction form of a-SLE}, $*$ denotes a known vector that can be determined by the downloaded data. In matrix form, \eqref{eqn LSE of repair remainder nodes} can be described as
		\begin{eqnarray}\label{eqn LSE of repair remainder nodes in matrix form}
			\setlength\arraycolsep{2pt}
			\left(\begin{array}{cccc}
				S_{i,\delta_z}A_{0,i} & \tilde{A}_{0,j_0,i,\delta_0} & \cdots & \tilde{A}_{0,j_{r-\delta_z-1},i,\delta_0}\\
				S_{i,\delta_z}A_{1,i} & \tilde{A}_{1,j_0,i,\delta_0} & \cdots & \tilde{A}_{1,j_{r-\delta_z-1},i,\delta_0}\\
				\vdots & \vdots & \vdots & \vdots\\
				S_{i,\delta_z}A_{r-1,i} & \tilde{A}_{r-1,j_0,i,\delta_0} & \cdots & \tilde{A}_{r-1,j_{r-\delta_z-1},i,\delta_0}\\
			\end{array}\right)\left(\begin{array}{c}
				\mathbf{f}_i^{(a)}\\
				R_{i,\delta_z}\mathbf{f}_{j_0}^{(a)}\\
				\vdots\\
				R_{i,\delta_z}\mathbf{f}_{j_{r-\delta_z-1}}^{(a)}
			\end{array}\right)+\sum\limits_{j\in \mathcal{G}}\left(\hspace{-1mm}\begin{array}{c}
				S_{i,\delta_z}\mathbf{P}_{0,j}^{(a)}\\
				S_{i,\delta_z}\mathbf{P}_{1,j}^{(a)}\\
				\vdots\\
				S_{i,\delta_z}\mathbf{P}_{r-1,j}^{(a)}\\
			\end{array}\hspace{-1mm}\right)=*,~a\in [l_{w+1}:l_w).
		\end{eqnarray}
		\item [3)] Recover the data stored at RN $i$ by sequentially solving \eqref{eqn LSE of repair remainder nodes in matrix form} for $w=0,1,\cdots,m-1$,
		\begin{itemize}
			\item [3-1)] Compute the second term on LHS of \eqref{eqn LSE of repair remainder nodes in matrix form} from the recovered data (or downloaded data) of GN $j\in \mathcal{G}$ if $w>0$.
			\item [3-2)] Recover the data $\mathbf{f}_i^{(a)}$ and $R_{i,\delta_z}\mathbf{f}_j^{(a)}$ for $j\in \mathcal{D}_z$ from \eqref{eqn LSE of repair remainder nodes in matrix form}.
		\end{itemize}
	\end{itemize}
	
	We now show that RN $i$ of the new code $\mathbbmss{C}_2$ has the same repair property as that of the base code $\mathbbmss{C}_0$, where $i\in [0:n)\backslash\mathcal{G}$.

		\begin{Lemma}\label{lem the important for the repair of remainder node}
		Given $z\in [0:m)$, $a\in [l_{w+1}:l_w)$ with $w\in [1:m)$, $t\in [0:r)$, $i\in [0:n)\backslash\mathcal{G}$ and $j\in \mathcal{G}$, the column vector $S_{i,\delta_z}\mathbf{P}_{t,j}^{(a)}$ can be computed from the data in set $\{R_{i,\delta_z}\mathbf{f}_j^{(b)}|b\in [l_w:l_0)\}$ if for  base code $\mathbbmss{C}_0$,  RN $i$  has the $\delta_z$-optimal repair property and  all GNs $j\in \mathcal{G}$ satisfy
		\begin{itemize}
			\item [C3.] $\mbox{Rank}\left(\left(\begin{array}{c}R_{i,\delta_z}\\S_{i,\delta_z}K_{t,j,v}\Phi_{\alpha,u}\end{array}\right)\right)=\frac{\alpha N}{\delta_z}$ for  $t\in [0:r)$, $u\in [0:\delta_0)$ and $v\in [0:\delta_{m-1}-\delta_0)$, where $\Phi_{\alpha,u}$ is the $\alpha N'\times \alpha N$ matrix defined by \eqref{eqn the matrix Phi}.
		\end{itemize} 
	\end{Lemma}
	
	\begin{proof}
		The proof is given in Appendix \ref{appen the proof of lem 3}.
	\end{proof}

	\begin{Theorem}\label{Thr_repair_for_remainder_nodes}
			Given $z\in [0:m)$, RN $i\in [0:n)\backslash\mathcal{G}$ of the new code $\mathbbmss{C}_2$ has the $\delta_z$-optimal repair/access property over $\mathbb{F}_q$ if for base code $\mathbbmss{C}_0$, RN $i$   has the $\delta_z$-optimal repair property  and  all GNs  satisfy C3. 
	\end{Theorem}

	\begin{proof}
		According to The Repair Procedure of RNs, $\gamma(d_z)=\frac{d_z}{\delta_z}l_0\alpha N=\frac{d_z}{d_z-k+1}l_0\alpha N$ due to $\mbox{Rank}(R_{i,\delta_z})=\frac{\alpha N}{\delta_z}$ and $\delta_z=d_z-k+1$, which attains the lower bound in \eqref{Eqn_bound_on_gamma}. In addition,  if node $i$ of the base code $\mathbbmss{C}_0$ has the $\delta_z$-optimal access property, i.e., the repair matrix $R_{i,\delta_z}$  also has only one nonzero element in each row, then by \eqref{eqn the repair,select matrix of C3 for goal nodes}  RN $i$ has the $\delta_z$-optimal access property in the new code $\mathbbmss{C}_2$.  Therefore, it is sufficient to show that 3) of The Repair Procedure of RNs can be executed under C3.  
		
		According to Lemma \ref{lem the requirement of obtaining fi}, by the $\delta_z$-optimal repair property of RN $i$ in  base code $\mathbbmss{C}_0$, the coefficient matrix of the first term on LHS of \eqref{eqn LSE of repair remainder nodes in matrix form} is nonsingular. Thus given $w\in [0:m)$ and $a\in [l_{w+1}:l_w)$, the data $\mathbf{f}_i^{(a)}$ and $R_{i,\delta_z}\mathbf{f}_{j_v}^{(a)}$, $v\in [0:r-\delta_z)$, can be repaired from \eqref{eqn LSE of repair remainder nodes in matrix form} if the following claim holds.
		\begin{itemize}
			\item [\textbf{Claim:}] Given $w\in [0:m)$, the second term on LHS of \eqref{eqn LSE of repair remainder nodes} for $a\in [l_{w+1}:l_w)$ is known.
		\end{itemize}
		
		By P0, $\mathbf{P}_{t,j}^{(a)}=\mathbf{0}$ for $a\in [l_1:l_0)$ and $j\in \mathcal{G}$, i.e., Claim holds for $w=0$. Suppose that Claim holds for all $0\le w\le s< m-1$, implying that we can obtain  data $R_{i,\delta_z}\mathbf{f}_{j_v}^{(a)}$ for $a\in [l_{s+1}:l_{0})$ and $v\in  [0:r-\delta_z)$. 		
		Therefore by Lemma \ref{lem the important for the repair of remainder node}, we are able to compute the second term on LHS of \eqref{eqn LSE of repair remainder nodes in matrix form} from $R_{i,\delta_z}\mathbf{f}_j^{(b)}$ for $b\in [l_{s+1}:l_0)$ and $j\in\mathcal{G}$, which have been either downloaded ($j\in \mathcal{H}_z\cap \mathcal{G}$) or repaired $(j\in \mathcal{D}_z\cap \mathcal{G}$).  That is, Claim also holds for $w=s+1$ and thus for all $0\le w<m$ by  the induction. This finishes the proof.
	\end{proof}

	Combining Theorems \ref{theorem the MDS property of code C3}-\ref{Thr_repair_for_remainder_nodes}, we then have the following results.
	
	\begin{Theorem}\label{Thr all conditions}
		By choosing an $(n,k)$ MDS array code $\mathbbmss{C}_0$ with the $\delta_0$-optimal repair/access property for all nodes and a set $\mathcal{G}$ of $\rho$ nodes  satisfying C1-C3,  the new $(n,k)$ MDS array code $\mathbbmss{C}_2$ has  the $\delta_{[0:m)}$-optimal repair/access property for these $\rho$ nodes over $\mathbb{F}_q$,  and preserves the $\delta_z$-optimal repair/access property for the other nodes where $z\in [0:m)$.
	\end{Theorem}
	
	Besides,  we have the following lemma whose proof is given in Appendix \ref{appen the proof of lem 4}. It is very useful when we recursively apply the construction method in the next section.
	
	\begin{Lemma}\label{Theorem the recursion result}
		Assume that  base code $\mathbbmss{C}_0$ has another set $\mathcal{G}'\subset [0:n)\backslash \mathcal{G}$ of $\rho$ nodes  satisfying C1-C3, whose key matrices  are $K_{t,i,v},t\in [0:r),i\in \mathcal{G}',v\in [0:\delta_{m-1}-\delta_0)$. Then, nodes in $\mathcal{G}'$ of new code $\mathbbmss{C}_2$  still satisfy C1-C3 with key matrices  of the form
		\begin{eqnarray}\label{eqn key matrix for recursion}
			K_{t,i,v}'=\mbox{blkdiag}(K_{t,i,v},K_{t,i,v},\cdots,K_{t,i,v})_{l_0},& i\in \mathcal{G}',t\in [0:r),v\in [0:\delta_{m-1}-\delta_0).
		\end{eqnarray}
	\end{Lemma}	
	
	\section{MDS Array Code Construction by Recursively Applying the Construction Method}
	In the previous section, we provided a construction method that can transform a specific $(n,k)$ MDS array code into a new $(n,k)$ MDS array code with the $\delta_{[0:m)}$-optimal repair property for a set of $\rho$ nodes, while the repair/access property of the remaining $n-\rho$ nodes are preserved. In this section,  by recursively applying the generic construction method, we propose a generic Algorithm \ref{Alg_repair_for_all_nodes}  that can build MDS array codes with the $\delta_{[0:m)}$-optimal repair property for all nodes. Specifically, by directly applying Algorithm \ref{Alg_repair_for_all_nodes} to VBK code \cite{kumar},  we get an MDS array code with the $\delta_{[0:m)}$-optimal access property for all nodes.
	
	\subsection{Generic Algorithm for Constructing MDS Array Code with the $\delta_{[0:m)}$-optimal Repair Property for All Nodes}\label{section algorithm for all nodes dz-optimal access}
	
	In this subsection, we introduce the generic algorithm based on a class of special MDS array codes, which is called \textit{transformable MDS} (TMDS) array codes.	
	\begin{Definition}\label{def T-MDS array code}
		An $(n,k)$ MDS array code defined in the form of \eqref{Eqn Code C} with the $\delta_0$-optimal repair property for all nodes is said to be a TMDS array code if there exists a partition $\mathcal{J}_0,\mathcal{J}_1,\cdots,\mathcal{J}_{\mu-1}$ of set $[0:n)$ such that the nodes in $\mathcal{J}_t$ of this code satisfy C1-C2 for $\alpha=1$, and C3 for $z=0,\alpha=1$.
	\end{Definition}
	
	
	\begin{Remark}
		If the value $\mu$ in Definition \ref{def T-MDS array code} is $n$, without loss of generality,  we always assume $\mathcal{J}_t=\{t\}$ for $t\in [0:\mu)$.
	\end{Remark}
		
	By means of the TMDS array code, we present a generic algorithm (Algorithm \ref{Alg_repair_for_all_nodes}) that can construct an MDS array code with the $\delta_{[0:m)}$-optimal repair property for all nodes by recursively using the construction method in Section \ref{section MDS array code with dz optimal property} $\mu$ times, where $\mu$ is the value given in Definition \ref{def T-MDS array code}. In the $(s+1)$-th round construction method, we choose code $\mathbbmss{Q}_{s}$ as the base code and denote the resultant code as $\mathbbmss{Q}_{s+1}$, where $0\le s<\mu$ and the key matrices of node $i\in [0:n)$ of code $\mathbbmss{Q}_{s}$ is defined as $K_{t,i,v}^{(s)},t\in [0:r),v\in [0:\delta_{m-1}-\delta_0)$.
	\begin{algorithm}[htbp]
		\caption{}\label{Alg_repair_for_all_nodes}
		\begin{algorithmic}[1]
			\Require An $(n,k)$ TMDS array code $\mathbbmss{Q}_0$ with sub-packetization $N$			
			\Ensure The desired code $\mathbbmss{Q}_\mu$ with the $\delta_{[0:m)}$-optimal repair property for all nodes, where its sub-packetization level is $l_0^\mu N$
			\For{$s=0$; $s < \mu$; $s++$}
			\State Set code $\mathbbmss{Q}_s$ of sub-packetization level $\alpha_s N$ as the base code, where $\alpha_s=l_0^s$
			\State Designate $\mathcal{J}_s$ given in Definition \ref{def T-MDS array code} as the set $\mathcal{G}$ in the construction method
			\State For $i\in \mathcal{G}$, setting the $\alpha_s N\times \alpha_s N'$ matrix $K_{t,i,v}=K_{t,i,v}^{(s)}$ in \eqref{eqn_specific_form_of_P}, where
			\begin{eqnarray}\label{Eqn key matrix for alg2}
				K_{t,i,v}^{(s)}=\mbox{blkdiag}(K_{t,i,v}^{(0)},K_{t,i,v}^{(0)},\cdots,K_{t,i,v}^{(0)})_{\alpha_s},t\in [0:r),v\in [0:\delta_{m-1}-\delta_0) 
			\end{eqnarray}
			\State Applying the construction method on code $\mathbbmss{Q}_s$ to generate a new code $\mathbbmss{Q}_{s+1}$ with sub-packetization level $l_0\alpha_s N$
			\EndFor
		\end{algorithmic}
	\end{algorithm}
	
	\begin{Theorem}\label{sec:modify adaptive:Thr the requirement for resultant code}
		By choosing an $(n,k)$ TMDS array code over $\mathbb{F}_q$ as base code, a new $(n,k)$ MDS array code generated from Algorithm \ref{Alg_repair_for_all_nodes} has  the $\delta_{[0:m)}$-optimal repair property for all nodes over $\mathbb{F}_q$, where the sub-packetization level of the new $(n,k)$ MDS array code is $l_0^\mu N=(\frac{\delta}{\delta_0})^\mu N$. Moreover, the new $(n,k)$ MDS array code has the $\delta_{[0:m)}$-optimal access property for all nodes if the base code has the $\delta_0$-optimal access property for all nodes.
	\end{Theorem}
	\begin{proof}
		According to Theorem \ref{Thr all conditions}, to obtain the desired MDS array code from Algorithm \ref{Alg_repair_for_all_nodes},  it is sufficient to show that for any $s\in [0:\mu)$, the nodes in $\mathcal{J}_s$ of code $\mathbbmss{Q}_s$ satisfy C1-C3 by setting the key matrices $K_{t,i,v}^{(s)}, t\in [0:r),i\in \mathcal{J}_s, v\in [0:\delta_{m-1}-\delta_0)$ 
		in \eqref{Eqn key matrix for alg2}. By recursively applying Lemma \ref{Theorem the recursion result}, we only need to varify that the $N\times N'$ matrices $K_{t,i,v}^{(0)},~t\in [0:r),~i\in \mathcal{J}_s, ~v\in [0:\delta_{m-1}-\delta_0)$ are the key matrices of the nodes in $\mathcal{J}_s$ of code $\mathbbmss{Q}_0$ such that they satisfy C1-C3, which is guaranteed by   Definition \ref{def T-MDS array code}.
	\end{proof}
	
	In the following, we provide an example of Algorithm \ref{Alg_repair_for_all_nodes}.
\begin{Example}
Applying  Algorithm \ref{Alg_repair_for_all_nodes} to the base code $\mathbbmss{C}_0$ in Example \ref{Exp set Piz for (1,2,3)}, we obtain the codes $\mathbbmss{Q}_1,\mathbbmss{Q}_2,\cdots,\mathbbmss{Q}_8$ through  eight rounds of the construction method  where  in round $s\in [1:8]$, the set $\{2s-2,2s-1\}$ is chosen as the set $\mathcal{G}$. Let $\mathbf{g}_{s,i}$ of length $3^s N$ be the data stored at node $i\in [0:16)$ of code $\mathbbmss{Q}_s$, where $s\in [1:8]$. For convenience, the  data $\mathbf{g}_{s,i}$  is always represented by 
			\begin{eqnarray*}
				\mathbf{g}_{s,i}=\left(\begin{array}{c}
					\mathbf{f}_i^{(0)} \\
					\mathbf{f}_i^{(1)}\\
					\vdots\\
					\mathbf{f}_i^{(3^s-1)}
				\end{array}\right), & 0\le i<16,\,s\in [1:8],
			\end{eqnarray*}
			where $\mathbf{f}_i^{(a)}$ is a column vector of length $N$. 
			Note that the PCGs of code $\mathbbmss{Q}_1$ has been shown in Example \ref{Exp (16,10) MDS array code C3 (2,3)}, i.e., the code $\mathbbmss{C}_2$ in Example \ref{Exp (16,10) MDS array code C3 (2,3)} is the code $\mathbbmss{Q}_1$.
			
			In what follows, we give the $t$-th PCG of the code $\mathbbmss{Q}_2$, while those of $\mathbbmss{Q}_3,\mathbbmss{Q}_4,\cdots,\mathbbmss{Q}_8$ can be obtained similarly. 
			By means of  Algorithm \ref{Alg the set Pij}, the sets $\mathcal{P}_{i,1}$ and $\mathcal{P}_{i,1}^{(a)}$, $a\in [0:2)$ of node $i\in \mathcal{G}=[2:4)$ for the second round are 
			\begin{eqnarray*}
			    &&\mathcal{P}_{i,1}=\{\mathbf{g}_{1,i}^{(2)}[0],\mathbf{g}_{1,i}^{(2)}[1]\}=\left\{\left(\begin{array}{c} \mathbf{f}_i^{(6)}[0]\\\mathbf{f}_i^{(7)}[0]\\\mathbf{f}_i^{(8)}[0] \end{array}\right),\left(\begin{array}{c} \mathbf{f}_i^{(6)}[1]\\\mathbf{f}_i^{(7)}[1]\\\mathbf{f}_i^{(8)}[1] \end{array}\right)\right\},\\
				&&\mathcal{P}_{i,1}^{(0)}=\{\mathbf{g}_{1,i}^{(2)}[0]\}=\left\{\left(\begin{array}{c} \mathbf{f}_i^{(6)}[0]\\\mathbf{f}_i^{(7)}[0]\\\mathbf{f}_i^{(8)}[0] \end{array}\right)\right\},\mathcal{P}_{i,1}^{(0)}=\{\mathbf{g}_{1,i}^{(2)}[1]\}=\left\{\left(\begin{array}{c} \mathbf{f}_i^{(6)}[1]\\\mathbf{f}_i^{(7)}[1]\\\mathbf{f}_i^{(8)}[1] \end{array}\right)\right\}.
			\end{eqnarray*}
			Through the generic construction method, the $t$-th PCG of the code $\mathbbmss{Q}_2$ are given as
			\begin{eqnarray*}\setlength\arraycolsep{0.6pt}
				\begin{small}
					\hspace{-3mm}\left(\hspace{-1.5mm}\begin{array}{l}
						A_{t,0}\mathbf{f}_0^{(0)}+\zeta_0^tV_{0,0}^\top\mathbf{f}_0^{(2)}[0]\\
						A_{t,0}\mathbf{f}_0^{(1)}+\zeta_0^tV_{0,0}^\top\mathbf{f}_0^{(2)}[1]\\
						A_{t,0}\mathbf{f}_0^{(2)}\\
						A_{t,0}\mathbf{f}_0^{(3)}+\zeta_0^tV_{0,0}^\top\mathbf{f}_0^{(5)}[0]\\
						A_{t,0}\mathbf{f}_0^{(4)}+\zeta_0^tV_{0,0}^\top\mathbf{f}_0^{(5)}[1]\\
						A_{t,0}\mathbf{f}_0^{(5)}\\
						A_{t,0}\mathbf{f}_0^{(6)}+\zeta_0^tV_{0,0}^\top\mathbf{f}_0^{(8)}[0]\\
						A_{t,0}\mathbf{f}_0^{(7)}+\zeta_0^tV_{0,0}^\top\mathbf{f}_0^{(8)}[1]\\
						A_{t,0}\mathbf{f}_0^{(8)}
					\end{array}\hspace{-1mm}\right)+\left(\hspace{-1.5mm}\begin{array}{l}
						A_{t,1}\mathbf{f}_1^{(0)}+\zeta_0^tV_{0,1}^\top\mathbf{f}_1^{(2)}[0]\\
						A_{t,1}\mathbf{f}_1^{(1)}+\zeta_0^tV_{0,1}^\top\mathbf{f}_1^{(2)}[1]\\
						A_{t,1}\mathbf{f}_1^{(2)}\\
						A_{t,1}\mathbf{f}_1^{(3)}+\zeta_0^tV_{0,1}^\top\mathbf{f}_1^{(5)}[0]\\
						A_{t,1}\mathbf{f}_1^{(4)}+\zeta_0^tV_{0,1}^\top\mathbf{f}_1^{(5)}[1]\\
						A_{t,1}\mathbf{f}_1^{(5)}\\
						A_{t,1}\mathbf{f}_1^{(6)}+\zeta_0^tV_{0,1}^\top\mathbf{f}_1^{(8)}[0]\\
						A_{t,1}\mathbf{f}_1^{(7)}+\zeta_0^tV_{0,1}^\top\mathbf{f}_1^{(8)}[1]\\
						A_{t,1}\mathbf{f}_1^{(8)}
					\end{array}\hspace{-1mm}\right)+	 \left(\hspace{-1.5mm}\begin{array}{c}
						A_{t,2}\mathbf{f}_2^{(0)}+\zeta_0^tV_{1,0}^\top\mathbf{f}_2^{(6)}[0]\\
						A_{t,2}\mathbf{f}_2^{(1)}+\zeta_0^tV_{1,0}^\top\mathbf{f}_2^{(7)}[0]\\
						A_{t,2}\mathbf{f}_2^{(2)}+\zeta_0^tV_{1,0}^\top\mathbf{f}_2^{(8)}[0]\\
						A_{t,2}\mathbf{f}_2^{(3)}+\zeta_0^tV_{1,0}^\top\mathbf{f}_2^{(6)}[1]\\
						A_{t,2}\mathbf{f}_2^{(4)}+\zeta_0^tV_{1,0}^\top\mathbf{f}_2^{(7)}[1]\\
						A_{t,2}\mathbf{f}_2^{(5)}+\zeta_0^tV_{1,0}^\top\mathbf{f}_2^{(8)}[1]\\
						A_{t,2}\mathbf{f}_2^{(6)}\\
						A_{t,2}\mathbf{f}_2^{(7)}\\
						A_{t,2}\mathbf{f}_2^{(8)}
					\end{array}\hspace{-1mm}\right)+	\left(\hspace{-1.5mm}\begin{array}{c}
						A_{t,3}\mathbf{f}_3^{(0)}+\zeta_0^tV_{1,1}^\top\mathbf{f}_3^{(6)}[0]\\
						A_{t,3}\mathbf{f}_3^{(1)}+\zeta_0^tV_{1,1}^\top\mathbf{f}_3^{(7)}[0]\\
						A_{t,3}\mathbf{f}_3^{(2)}+\zeta_0^tV_{1,1}^\top\mathbf{f}_3^{(8)}[0]\\
						A_{t,3}\mathbf{f}_3^{(3)}+\zeta_0^tV_{1,1}^\top\mathbf{f}_3^{(6)}[1]\\
						A_{t,3}\mathbf{f}_3^{(4)}+\zeta_0^tV_{1,1}^\top\mathbf{f}_3^{(7)}[1]\\
						A_{t,3}\mathbf{f}_3^{(5)}+\zeta_0^tV_{1,1}^\top\mathbf{f}_3^{(8)}[1]\\
						A_{t,3}\mathbf{f}_3^{(6)}\\
						A_{t,3}\mathbf{f}_3^{(7)}\\
						A_{t,3}\mathbf{f}_3^{(8)}
					\end{array}\hspace{-1mm}\right)+\sum\limits_{i=4}^{15}	\left(\hspace{-1.5mm}\begin{array}{c}
						A_{t,i}\mathbf{f}_i^{(0)}\\
						A_{t,i}\mathbf{f}_i^{(1)}\\
						A_{t,i}\mathbf{f}_i^{(2)}\\
						A_{t,i}\mathbf{f}_i^{(3)}\\
						A_{t,i}\mathbf{f}_i^{(4)}\\
						A_{t,i}\mathbf{f}_i^{(5)}\\
						A_{t,i}\mathbf{f}_i^{(6)}\\
						A_{t,i}\mathbf{f}_i^{(7)}\\
						A_{t,i}\mathbf{f}_i^{(8)}
					\end{array}\hspace{-1mm}\right)=\mathbf{0},
				\end{small}
			\end{eqnarray*}
		where $0\le t<6$.
\end{Example}

	\subsection{An $(n,k)$ MDS array code $\mathbbmss{G}$ by Applying  Algorithm \ref{Alg_repair_for_all_nodes} to VBK code in \cite{kumar}}
	In this subsection, we generate an MDS array code $\mathbbmss{G}$ by applying Algorithm 2 to the $(n,k)$ VBK code which has the $\delta_0$-optimal access property for all nodes and sub-packetization level $N=\delta_0^\tau$, where $\tau={\lceil \frac{n}{\delta_0}\rceil}$, $\delta_0 \in \{2,3,4\}$ and $r=n-k > \delta_0$. In what follows, we first visit the definition of the VBK code. 
	
	Let $\varepsilon \not\in \{0,1\}$ be an element in the field $\mathbb{F}_q$. For $x\in [0:\tau)$ and $\delta_0=2,3,4$, respectively define $\delta_0\times \delta_0$ matrix $\Theta_x$ as
\begin{equation*}
\Theta_x= \left\{\begin{array}{ll}
				\left(\begin{array}{cc}
			\vartheta_{0,x} & \varepsilon \vartheta_{1,x}\\
			\vartheta_{1,x} & \vartheta_{0,x}
		\end{array}\right), & \textrm{if }\delta_0=2, \\
		\left(\begin{array}{ccc}
			\vartheta_{0,x} & \varepsilon \vartheta_{1,x} & \varepsilon \vartheta_{2,x}\\
			\vartheta_{1,x} & \vartheta_{0,x} & \varepsilon \vartheta_{3,x}\\
			\vartheta_{2,x} & \vartheta_{3,x} & \vartheta_{0,x}
		\end{array}\right), & \textrm{if }\delta_0=3,\\
	\left(\begin{array}{cccc}
			\vartheta_{0,x} & \varepsilon \vartheta_{1,x} & \varepsilon \vartheta_{2,x} & \varepsilon \vartheta_{3,x}\\
			\vartheta_{1,x} & \vartheta_{0,x} & \varepsilon \vartheta_{3,x} & \varepsilon\vartheta_{2,x}\\
			\vartheta_{2,x} & \vartheta_{3,x} & \vartheta_{0,x} & \varepsilon \vartheta_{1,x}\\
			\vartheta_{3,x} & \vartheta_{2,x} & \vartheta_{1,x} & \vartheta_{0,x}
		\end{array}\right),	&\textrm{if }\delta_0=4,
		\end{array}\right.    
\end{equation*}	
	where  $\{\vartheta_{i,x},\varepsilon\vartheta_{i,x},\vartheta_{0,x}|i \in \{1,2,3\}, x\in [0:\tau)\}$ is a collection of distinct elements in $\mathbb{F}_q$  with 
	\begin{eqnarray}\label{Eqn the finite field size of VBK code}
		q\ge \left\{\begin{array}{ll}
			6\lceil \frac{n}{2} \rceil+2, & \textrm{if }\delta_0=2, \\
			18\lceil\frac{n}{\delta_0}\rceil+2, & \textrm{if }\delta_0=3,4.\\
		\end{array}\right.
	\end{eqnarray}

Let $s$ be any given positive integer,	for any $0\le x<s$, $u\in [0:\delta_0)$ and given $a=(a_{s-1},a_{s-2},\cdots,a_0)\in [0:\delta_0^s)$, define 
	\begin{eqnarray}\label{Eqn a(i,u)}
		\pi_{s}(a,x,u)=(a_{s-1},\cdots,a_{x+1},u,a_{x-1},\cdots,a_0),
	\end{eqnarray}  
	i.e., replace the $x$-digit $a_x$ of the vector $a=(a_{s-1},\cdots,a_{x+1},u,a_{x-1},\cdots,a_0)$ by $u$. 
	For $x\in [0:\tau)$, $y\in [0:\delta_0)$, and $t=0,1,\cdots,r-1$, define an $N \times N$ matrix $A_{t,\delta_0x+y}$ as
	\begin{eqnarray}\label{eqn the parity-matrix of VBK code}
		A_{t,\delta_0x+y}=\sum\limits_{a=0}^{N-1}\lambda_{\delta_0x+y,a_x}^te_a^\top e_a+\sum\limits_{a=0,a_x=y}^{N-1}\sum\limits_{u=0,u\ne y}^{\delta_0-1}\varepsilon_{u,y}\lambda_{\delta_0x+y,u}^te_a^\top e_{\pi_\tau(a,x,u)},
	\end{eqnarray}
	where $\{e_a|a\in a\in [0:N)\}$ is the standard basis of $\mathbb{F}_q^N$ defined in \eqref{Eqn_SB},  
	\begin{eqnarray*}\label{eqn coefficients lambda for VBK code}
		\lambda_{\delta_0x+y,v}=\Theta_x(v,y) \textrm{~for~} 0\le x<\tau, 0\le v,y<\delta_0
	\end{eqnarray*}
	 and
	\begin{eqnarray}\label{eqn coefficients gamma for VBK code}
		\varepsilon_{u,y}=\left\{\begin{array}{ll}
			\varepsilon, & \mathrm{if~}u<y,\\
			1, & \mathrm{if~}u>y.
		\end{array}\right.
	\end{eqnarray}
	Then, the $(n,k)$ VBK code is defined by \eqref{Eqn Code C} with parity-check matrices $(A_{t,i})_{t\in [0:r),i\in [0:n)}$ given in \eqref{eqn the parity-matrix of VBK code},  $\delta_0$-repair matrices $R_{i,\delta_0}$ and $\delta_0$-select matrices $S_{i,\delta_0}$  given as
	\begin{eqnarray}\label{eqn the repair, select matrix of VBK code}
		R_{i,\delta_0}=S_{i,\delta_0}=V_{\lfloor\frac{i}{\delta_0}\rfloor,i\%\delta_0}, i\in [0:n).
	\end{eqnarray}

	In the sequel, we show that the VBK code is a TMDS array code with the sets $\mathcal{J}_0,\mathcal{J}_1,\cdots,\mathcal{J}_{\mu-1}$  in Definition \ref{def T-MDS array code} being
	\begin{eqnarray}\label{eqn the set J of VBK code}
		\mathcal{J}_s=\left\{\begin{array}{ll}
			\{s\delta_0,s\delta_0+1,\cdots,s\delta_0+\delta_0-1\}, & \mbox{if~} 0\le s<\mu-1,\\
			\{s\delta_0,s\delta_0+1,\cdots,n-1\}, & \mbox{if~}s=\mu-1,
		\end{array}\right.	
	\end{eqnarray}
where $\mu=\tau$ and the matrix $K_{t,i,v}^{(0)}$ in \eqref{Eqn key matrix for alg2} is set as
	\begin{eqnarray}\label{eqn the key matrix of VBK code}
		K_{t,i,v}^{(0)}=\zeta_v^t R_{i,\delta_0}^\top,& t\in [0:r),i\in [0:n),v\in [0:\delta_{m-1}-\delta_0),
	\end{eqnarray}
	with $\zeta_0,\zeta_1,\cdots,\zeta_{\delta_{m-1}-\delta_0-1}$ being $\delta_{m-1}-\delta_0$ distinct elements in $\mathbb{F}_q\backslash\{\lambda_{i,u}|i\in [0:n),u\in [0:\delta_0)\}$.

	 In what follows, we check that for any $x\in [0:\mu)$, the nodes is set $\mathcal{J}_x$ of VBK code satisfy C1-C2 for $\alpha=1$, and C3 for $\alpha=1,z=0$. For convenience,  let $\epsilon_0,\epsilon_1,\cdots, \epsilon_{N'-1}$ be the standard basis of $\mathbb{F}_q^{N'}$ defined in \eqref{Eqn_SB} from now on,
	where $N'=\frac{N}{\delta_0}=\delta_0^{\tau-1}$. 
	
	First of all, we verify that for any $x\in [0:\mu)$, the nodes in set $\mathcal{J}_x$  of VBK code satisfy C1 for $\alpha=1$ and C3 for $\alpha=1,z=0$  with the help of Lemma \ref{lem the first helpful lemma for the proof of Thm. 10},  whose proof is given in Appendix \ref{appen proof of lems 5,6,8}.
	
	\begin{Lemma}\label{lem the first helpful lemma for the proof of Thm. 10}
		For any $0\le \tilde{x}\ne x <\tau$ and $0\le u,v,h<\delta_0$, 
		\begin{itemize}
			\item [(i)] $V_{x,u}V_{x,u}^\top=I_{N'}$ and $V_{x,u}V_{x,v}^\top=\mathbf{0}$ if $u\ne v$; and
			\item [(ii)] $V_{x,u}(V_{\tilde{x},v}^\top\Delta_h)=T_{x,\tilde{x},v,h}V_{x,u}$ for some $N'\times N'$ matrices $T_{x,\tilde{x},v,h}$ with
			\begin{eqnarray}\label{eqn the row vector of matrix Tijuvh}
				T_{x,\tilde{x},v,h}(a,:)=\left\{\begin{array}{ll}
					\epsilon_{(h,a_{\tau-2},\cdots,a_{\tilde{x}},a_{\tilde{x}-2},\cdots,a_0)}, & \mbox{if~~}0\le x<\tilde{x}<\tau,a_{\tilde{x}-1}=v,\\
					\epsilon_{(h,a_{\tau-2},\cdots,a_{\tilde{x}+1},a_{\tilde{x}-1},\cdots,a_0)}, & \mbox{if~~}0\le \tilde{x}<x<\tau,a_{\tilde{x}}=v,\\
					\mathbf{0}, & \mbox{otherwise},
				\end{array}\right.
			\end{eqnarray}
			where $a=(a_{\tau-2},a_{\tau-3},\cdots,a_0)\in [0:N')$ and  $\Delta_h$ is the $N'\times N$ matrix defined in \eqref{eqn the matrix Delta}.
		\end{itemize}
	\end{Lemma}

	\begin{Theorem}\label{Thr VBK code satisfy C2}
		By setting the $N\times N'$ key matrix $K_{t,i,v}^{(0)}$  of node $i$ of VBK code  as in \eqref{eqn the key matrix of VBK code},   the nodes with indices in $\mathcal{J}_x$ ($x\in [0:\mu)$)  of VBK code satisfy C1 for $\alpha=1$  and C3 for $z=0,\alpha=1$.
	\end{Theorem}
	\begin{proof}
		For any $i,j\in [0:n)$, let $u=i\%\delta_0$ and $u'=j\%\delta_0$.  
		
		Firstly, consider $i,j\in \mathcal{J}_x$ with $i\ne j$ for $x\in [0:\mu)$. According to \eqref{eqn the set J of VBK code}, we have $\lfloor \frac{i}{\delta_0}\rfloor=\lfloor \frac{j}{\delta_0}\rfloor=x$ and $u\ne u'$. Thus by (i) of Lemma \ref{lem the first helpful lemma for the proof of Thm. 10} and \eqref{eqn the repair, select matrix of VBK code},
		\begin{eqnarray*}
			S_{i,\delta_0}K_{t,j,v}^{(0)}=S_{i,\delta_0}\cdot\zeta_v^tR_{j,\delta_0}^\top=\zeta_v^tV_{x,u}V_{x,u'}^\top=\mathbf{0}, t\in [0:r),v\in [0:\delta_{m-1}-\delta_0),
		\end{eqnarray*}
		which means that the nodes with indices in $\mathcal{J}_x$ of VBK code satisfy C1.
		
		 Next we show that any node  $j\in \mathcal{J}_x$ satisfy C3 for $z=0,\alpha=1$. For any $i\in \mathcal{J}_{\tilde{x}}$ with $0\le \tilde{x}\ne x<\mu=\tau$, then 
		\begin{eqnarray*}
			\textrm{Rank}\left(\left(\begin{array}{c}
				R_{i,\delta_0}\\
				S_{i,\delta_0}(K_{t,j,v}^{(0)}\Phi_{1,h})
			\end{array}\right)\right)
			&=&\textrm{Rank}\left(\left(\begin{array}{c}
				V_{\tilde{x},u}\\
				\zeta_v^tV_{\tilde{x},u}(V_{x,u'}^\top \Delta_h)
			\end{array}\right)\right)\\
			&=&\textrm{Rank}\left(\left(\begin{array}{c}
				V_{\tilde{x},u}\\
				\zeta_v^tT_{\tilde{x},x,u',h}V_{\tilde{x},u}
			\end{array}\right)\right)\\
                        &=&\textrm{Rank}(V_{\tilde{x},u})\\
			&=&{N\over \delta_0}
		\end{eqnarray*}
		for $h\in [0:\delta_0)$, $t\in [0:r)$ and $v\in [0:\delta_{m-1}-\delta_0)$, where the first equality holds due to \eqref{eqn the matrix Phi}, \eqref{eqn the repair, select matrix of VBK code} and \eqref{eqn the key matrix of VBK code}, and the second equality follows from (ii) of Lemma \ref{lem the first helpful lemma for the proof of Thm. 10}.  Then, the nodes with indices in set $\mathcal{J}_x$ satisfy C3 for $z=0,\alpha=1$, 
	\end{proof}
	
	Next, we show that for any $s\in [0:\mu)$, the nodes with indices in set  $\mathcal{J}_s$ of VBK code satisfy C2 for $\alpha=1$. That is, we need to verify that the matrix $M_{i,\mathcal{D}_z}$ defined in \eqref{eqn matrix MijDz} with $\alpha=1$ is nonsingular for any given $i\in [0:n)$, $z\in [1:m)$ and $\mathcal{D}_z=\{j_0,j_1,\cdots,j_{r-\delta_z-1}\}\subset [0:n)\backslash\{i\}$. According to the definition of matrix $M_{i,\mathcal{D}_z}$, the verification of its  nonsingularity requires  to determine the form of $S_{i,\delta_0}A_{t,i}$ and $\tilde{A}_{t,j,i,\delta_0}$ for $0\le i \ne j<n$, which will be ensured by Lemma \ref{lem alignment interference for VBK code}. In addition, Lemmas \ref{lem important for def_P4 of C2} and \ref{lem the helpful lemma for the proof of Thm. 10}  are also 
		critical to proving the invertibility of the matrix $M_{i,\mathcal{D}_z}$. Lemmas \ref{lem important for def_P4 of C2} can be proved similar to the proof of MDS property of VBK code  in \cite{kumar}, thus we omit it here. Whereas,  the proofs of Lemmas \ref{lem alignment interference for VBK code} and \ref{lem the helpful lemma for the proof of Thm. 10} are given in Appendix \ref{appen proof of lems 5,6,8}.

	\begin{Lemma}\label{lem alignment interference for VBK code}
			For any  $i=\delta_0\tilde{x}+\tilde{y}$ and $j=\delta_0x+y\in [0:n)\backslash\{i\}$, where $0\le x,\tilde{x}<\tau$ and $0\le y,\tilde{y}<\delta_0$, 
			\begin{itemize}
				\item [(i)] $S_{i,\delta_0}A_{t,i}=\lambda_{i,\tilde{y}}^tV_{\tilde{x},\tilde{y}}+\sum\limits_{u=0,u\ne \tilde{y}}^{\delta_0-1}\varepsilon_{u,\tilde{y}}\lambda_{i,u}^tV_{\tilde{x},u}$;
				\item [(ii)] The matrix $\tilde{A}_{t,j,i,\delta_0}$  in \eqref{repair node requirement3} is of the form
				\begin{eqnarray}\label{eqn matrix Atji for VBK code}
					\tilde{A}_{t,j,i,\delta_0}=\left\{\begin{array}{ll}
						\sum\limits_{a=0}^{N'-1}\lambda_{j,a_x}^t\epsilon_a^\top \epsilon_a+\sum\limits_{a=0,a_x=y}^{N'-1}\sum\limits_{u=0,u\ne y}^{\delta_0-1}\varepsilon_{u,y}\lambda_{j,u}^t\epsilon_a^\top \epsilon_{\pi_{\tau-1}(a,x,u)}, & \textrm{if~}x<\tilde{x},\\[6pt]
					\lambda_{j,\tilde{y}}^tI_{N'}, & \textrm{if~}x=\tilde{x},\\[6pt]
						\sum\limits_{a=0}^{N'-1}\lambda_{j,a_{x-1}}^t\epsilon_a^\top \epsilon_a+\sum\limits_{a=0,a_{x-1}=y}^{N'-1}\sum\limits_{u=0,u\ne y}^{\delta_0-1}\varepsilon_{u,y}\lambda_{j,u}^t\epsilon_a^\top \epsilon_{\pi_{\tau-1}(a,x-1,u)}, & \textrm{if~}x>\tilde{x},
					\end{array}\right.
				\end{eqnarray}
			\end{itemize}
			where $A_{t,i}$ and $S_{i,\delta_0}$ are defined in \eqref{eqn the parity-matrix of VBK code} and \eqref{eqn the repair, select matrix of VBK code}, respectively.
		\end{Lemma}

	\begin{Lemma}\label{lem important for def_P4 of C2}
		For any given $i\in [0:n)$ and $\{j_0,j_1,\cdots,j_{s-1}\}\subset [0:n)\backslash\{i\}$ with $1\le s<r$, when $\delta_0=2,3,4$,  the block matrix
		\begin{eqnarray*}
			\left(\begin{array}{cccc}
				\tilde{A}_{0,j_0,i,\delta_0} & \tilde{A}_{0,j_1,i,\delta_0} & \cdots & \tilde{A}_{0,j_{s-1},i,\delta_0}\\
				\tilde{A}_{1,j_0,i,\delta_0} & \tilde{A}_{1,j_1,i,\delta_0} & \cdots & \tilde{A}_{1,j_{s-1},i,\delta_0}\\
				\vdots & \vdots & \vdots & \vdots\\
				\tilde{A}_{s-1,j_0,i,\delta_0} & \tilde{A}_{s-1,j_1,i,\delta_0} & \cdots & \tilde{A}_{s-1,j_{s-1},i,\delta_0}\\
			\end{array}\right)
		\end{eqnarray*}
		of order $sN'$ is nonsingular over $\mathbb{F}_{q}$, where the $N'\times N'$ matrix $\tilde{A}_{t,j,i,\delta_0}$ is given by \eqref{eqn matrix Atji for VBK code}.
	\end{Lemma}

	\begin{Lemma}\label{lem the helpful lemma for the proof of Thm. 10}
		Let $\beta_0,\beta_1,\cdots,\beta_{r-1}$ be the elements in $\mathbb{F}_{q}$. For any given $i=\delta_0\tilde{x}+\tilde{y}$ and any $0\le s\le p<r$, define a $(r-s)N'\times (r-s)N'$ matrix $H_{i,p,s}$ as
		\begin{eqnarray}\label{eqn the matrix Hzys}
			H_{i,p,s}=\left(\begin{array}{cccc;{2pt/2pt}ccc}
				I_{N'} & I_{N'} & \cdots & I_{N'} &  \tilde{A}_{0,j_0,i,\delta_0} & \cdots & \tilde{A}_{0,j_{r-p-1},i,\delta_0}\\
				\beta_s I_{N'} &  \beta_{s+1}I_{N'} & \cdots & \beta_{p-1}I_{N'} &  \tilde{A}_{1,j_0,i,\delta_0} & \cdots & \tilde{A}_{1,j_{r-p-1},i,\delta_0}\\
				\vdots & \vdots & \vdots & \vdots  & \vdots & \vdots & \vdots\\
				\beta_s^{r-s-1} I_{N'} &  \beta_{s+1}^{r-s-1}I_{N'} & \cdots & \beta_{p-1}^{r-s-1}I_{N'} & \tilde{A}_{r-s-1,j_0,i,\delta_0} & \cdots & \tilde{A}_{r-s-1,j_{r-p-1},i,\delta_0}\\
			\end{array}\right)
		\end{eqnarray}
		where $j_0,j_1,\cdots,j_{r-p-1}\in [0:n)\backslash\{i\}$. Then for any $0\le s<p$,  $|H_{i,p,s}| \ne 0$ if 
		\begin{itemize}
			\item [(i)] $|H_{i,p,s+1}|\ne 0$;
			\item [(ii)] $\beta_u\ne \beta_v$ and $\beta_u\ne \lambda_{j,\tilde{y}}$ for any $0\le u\ne v<p$ and $j\in [0:n)\backslash\{i\}$ with $\lfloor \frac{j}{\delta_0}\rfloor=\lfloor \frac{i}{\delta_0}\rfloor$;
			\item [(iii)] $\beta_u\ne \lambda_{j,v}$ for any $0\le u<p$, $0\le v<\delta_0$ and $j\in [0:n)\backslash\{i\}$ with $\lfloor \frac{j}{\delta_0}\rfloor\ne \lfloor \frac{i}{\delta_0}\rfloor$.
		\end{itemize}
	\end{Lemma}
	
	\begin{Theorem}\label{Thr VBK code satisfy C3}
		  By setting the $N\times N'$ key matrix $K_{t,i,v}^{(0)}$  of node $i$ of VBK code to be the form  in \eqref{eqn the key matrix of VBK code}, the matrix $M_{i,\mathcal{D}_z}$  in \eqref{eqn matrix MijDz} is nonsingular over $\mathbb{F}_q$ for any given $i=\delta_0\tilde{x}+\tilde{y}\in [0:n)$ and $\mathcal{D}_z=\{j_0,j_1,\cdots,j_{r-\delta_z-1}\}\subset [0:n)\backslash\{i\}$ with $z\in [1:m)$, where $q$ is determined in \eqref{Eqn the finite field size of VBK code}.
	\end{Theorem}
	
	\begin{proof}
	For simplicity, let 
		\begin{eqnarray}\label{eqn the coefficient beta}
			\beta_u=\left\{\begin{array}{ll}
				\lambda_{i,u}, & \textrm{if }u\in [0:\delta_0),\\
				\zeta_{u-\delta_0}, & \textrm{if }u\in [\delta_0:\delta_{m-1}),
			\end{array}\right.	
		\end{eqnarray}
		where it is noting  from the definitions of $\lambda_{i,u},\zeta_{u-\delta_0}$  that 
		\begin{itemize}
			\item [\textbf{(i)}] $\beta_u\ne \beta_v$, $\beta_u\ne \lambda_{j,a}$ and $\beta_u \ne \lambda_{j',\tilde{y}}$ for any $0\le u\ne v<\delta_z$, $a\in [0:\delta_0)$, $j\in [0:n)$ with $\lfloor\frac{j}{\delta_0}\rfloor \ne \tilde{x}$ and $j'\in [0:n)\backslash\{i\}$ with $\lfloor\frac{j'}{\delta_0}\rfloor= \tilde{x}$. 
		\end{itemize}
		
		By replacing $K_{t,i,v}$ in \eqref{eqn matrix MijDz} with $K_{t,i,v}^{(0)}$, we first calculate
		\begin{align}\label{eqn the form of MiDz(t,:)}
			&\hspace{1.5em}\left(\begin{array}{ccccccc}
	S_{i,\delta_0}A_{t,i} & \tilde{A}_{t,j_0,i,\delta_0} & \cdots & \tilde{A}_{t,j_{r-\delta_z-1},i,\delta_0} & S_{i,\delta_0}K_{t,i,0}^{(0)} & \cdots & S_{i,\delta_0}K_{t,i,\delta_z-\delta_0-1}^{(0)}
\end{array}\right)\notag\\
			&=\left(\hspace{-1mm}\begin{array}{ccccccc}
	\lambda_{i,\tilde{y}}^tV_{\tilde{x},\tilde{y}}+\sum\limits_{u=0,u\ne \tilde{x}}^{\delta_0-1}\varepsilon_{u,\tilde{y}}\lambda_{i,u}^tV_{\tilde{x},u} & \tilde{A}_{t,j_0,i,\delta_0} & \cdots & \tilde{A}_{t,j_{r-\delta_z-1},i,\delta_0} & \zeta_0^tV_{\tilde{x},\tilde{y}}V_{\tilde{x},\tilde{y}}^\top & \cdots & \zeta_{\delta_z-\delta_0-1}^tV_{\tilde{x},\tilde{y}}V_{\tilde{x},\tilde{y}}^\top
\end{array}\hspace{-1mm}\right)\notag\\
			&=\left(\hspace{-1mm}\begin{array}{ccccccccc}
	  \lambda_{i,0}^tI_{N'} & \cdots & \lambda_{i,\delta_0-1}^tI_{N'} & \tilde{A}_{t,j_0,i,\delta_0} & \cdots & \tilde{A}_{t,j_{r-\delta_z-1},i,\delta_0} & \zeta_0^tI_{N'} & \cdots & \zeta_{\delta_z-\delta_0-1}^tI_{N'}
\end{array}\hspace{-1mm}\right)B_{\mathrm{VBK}}\notag\\
			&=\left(\hspace{-2mm}\begin{array}{ccccccccc}
				\beta_{0}^tI_{N'} & \cdots & \beta_{\delta_0-1}^tI_{N'} & \tilde{A}_{t,j_0,i,\delta_0} & \cdots & \tilde{A}_{t,j_{r-\delta_z-1},i,\delta_0} & \beta_{\delta_0}^t I_{N'} & \cdots & \beta_{\delta_z-1}^t I_{N'}
			\end{array}\hspace{-2mm}\right)B_{\mathrm{VBK}}
		\end{align} 		
where the first equality follows from Lemma \ref{lem alignment interference for VBK code}, \eqref{eqn the repair, select matrix of VBK code} and \eqref{eqn the key matrix of VBK code}, the second equality comes from (i) of Lemma \ref{lem the first helpful lemma for the proof of Thm. 10}, the third equality follows from \eqref{eqn the coefficient beta}, and the $rN'\times rN'$ matrix
		\begin{eqnarray*}
			B_{\mathrm{VBK}}\triangleq\left(\begin{array}{cc}
				\begin{array}{c}
				w_0 V_{\tilde{x},0} \\
				\vdots \\
				w_{\delta_0-1}V_{\tilde{x},\delta_0-1} 
				\end{array} & \\
				 & I_{(r-\delta_0)N'}
			\end{array}\right)\textrm{~with~}w_u=\left\{\begin{array}{ll}
				\varepsilon_{u,\tilde{y}}, & \mbox{if~}u\ne \tilde{y},\\
				1, & \mbox{otherwise}.
			\end{array}\right.
		\end{eqnarray*}
It is easy to see that the block matrix $B_{\mathrm{VBK}}$ is nonsingular. 
		
		\textit{Case 1.} If $\delta_z=r$, i.e., $\mathcal{D}_z=\emptyset$, then by \eqref{eqn matrix MijDz}  and \eqref{eqn the form of MiDz(t,:)}, we have that the matrix $M_{i,\mathcal{D}_z}$ is of the form
		\begin{eqnarray}\label{eqn the form of MiDz for dz=r}
			M_{i,\mathcal{D}_z}=\left(\begin{array}{cccccc;{2pt/2pt}ccc}
				I_{N'} & I_{N'} & \cdots & I_{N'} \\
				\beta_0 I_{N'} & \beta_1 I_{N'} & \cdots & \beta_{r-1}I_{N'}\\
				\vdots & \vdots & \ddots & \vdots \\
				\beta_0^{r-1}  & \beta_{1}^{r-1}I_{N'} & \cdots & \beta_{r-1}^{r-1}I_{N'}\\
			\end{array}\right)B_{\mathrm{VBK}},
		\end{eqnarray}
Note that the first block Vandermond matrix on RHS of \eqref{eqn the form of MiDz for dz=r} is nonsingular according to (i), so is the matrix $M_{i,\mathcal{D}_z}$.
		
		\textit{Case 2.} If $\delta_z<r$, we have that the matrix $M_{i,\mathcal{D}_z}$ is of the form
		\begin{eqnarray*}\setlength\arraycolsep{2pt}
			M_{i,\mathcal{D}_z}=\underbrace{\left(\hspace{-1mm}\begin{array}{ccc;{2pt/2pt}ccc;{2pt/2pt}ccc}
					I_{N'}  & \cdots & I_{N'} & \tilde{A}_{0,j_0,i,\delta_0} & \cdots & \tilde{A}_{0,j_{r-\delta_z-1},i,\delta_0} & I_{N'}  & \cdots & I_{N'} \\
					\beta_0 I_{N'}  & \cdots & \beta_{\delta_0-1}I_{N'} &  \tilde{A}_{1,j_0,i,\delta_0} & \cdots & \tilde{A}_{1,j_{r-\delta_z-1},i,\delta_0} & \beta_{\delta_0} I_{N'}  & \cdots & \beta_{\delta_z-1}I_{N'} \\
					\vdots & \vdots  & \vdots  & \vdots & \vdots & \vdots & \vdots & \vdots & \vdots\\
					\beta_0^{r-1} I_{N'} & \cdots & \beta_{\delta_0-1}^{r-1}I_{N'} & \tilde{A}_{r-1,j_0,i,\delta_0} & \cdots & \tilde{A}_{r-1,j_{r-\delta_z-1},i,\delta_0} & \beta_{\delta_0}^{r-1} I_{N'}  & \cdots & \beta_{\delta_z-1}^{r-1}I_{N'}\\
				\end{array}\hspace{-1mm}\right)}_{A_{\mathrm{VBK}}} B_{\mathrm{VBK}}.
		\end{eqnarray*}
		according to \eqref{eqn matrix MijDz} and \eqref{eqn the form of MiDz(t,:)}. Then, the matrix $M_{i,\mathcal{D}_z}$ is nonsingular if the block matrix $A_{\mathrm{VBK}}$ is nonsingular. By switching some block columns of $A_{\mathrm{VBK}}$, we then get $|A_{\mathrm{VBK}}|\ne 0$ if and only if $|H_{i,\delta_z,0}| \ne 0$, where $H_{i,\delta_z,0}$ is defined in \eqref{eqn the matrix Hzys}.  By (i) and Lemma \ref{lem the helpful lemma for the proof of Thm. 10},  $|H_{i,\delta_z,0}|\ne 0$ if $|H_{i,\delta_z,\delta_z}|\ne 0$,  where the invertibility of $H_{i,\delta_z,\delta_z}$
 follows from Lemma \ref{lem important for def_P4 of C2}.
		
		Collecting the above two cases, we can conclude that the matrix $M_{i,\mathcal{D}_z}$ with $\alpha=1$ is noningular over $\mathbb{F}_{q}$ for any $i\in[0:n)$ and $\mathcal{D}_z\in [0:n)\backslash\{i\}$ with $z\in [1:m)$. This finishes the proof.
	\end{proof}
	
	By combining Theorems \ref{Thr VBK code satisfy C2} and \ref{Thr VBK code satisfy C3}, and the $\delta_0$-optimal repair property of VBK code, we have the following theorem.
	\begin{Theorem}\label{Thm VBK is a TMDS code}
		The $(n,k)$ VBK code is a TMDS array code over $\mathbb{F}_{q}$ with $q$  in \eqref{Eqn the finite field size of VBK code}, and  
		\begin{itemize}
			\item The sets $\mathcal{J}_0,\mathcal{J}_1,\cdots,\mathcal{J}_{\mu-1}$  in Definition \ref{def T-MDS array code} are given by \eqref{eqn the set J of VBK code}, where $\mu=\tau=\lceil \frac{n}{\delta_0}\rceil$,
			\item For $i\in [0:n),t\in [0:r)$ and $v\in [0:\delta_{m-1}-\delta_0)$, the $N\times N'$ key matrix $K_{t,i,v}^{(0)}$ in \eqref{Eqn key matrix for alg2} is  $\zeta_v^t R_{i,\delta_0}^{\top}$, where $\zeta_0,\zeta_1,\cdots,\zeta_{\delta_{m-1}-\delta_0-1}$ are all distinct elements in $\mathbb{F}_q\backslash\{\lambda_{i,u}|i\in [0:n),u\in [0:\delta_0)\}$.
		\end{itemize}
	\end{Theorem}

	The following theorem immediately follows from  Theorems \ref{sec:modify adaptive:Thr the requirement for resultant code} and  \ref{Thm VBK is a TMDS code}.   
	\begin{Theorem}\label{theorem G1 obtained by algorithm}
		By choosing VBK code as base code in Algorithm \ref{Alg_repair_for_all_nodes}, an $(n,k)$ MDS array code $\mathbbmss{G}$ with the $\delta_{[0:m)}$-optimal access property for all nodes over $\mathbb{F}_q$ can be obtained, where $q\ge 6\lceil \frac{n}{2} \rceil+2$ if $\delta_0=2$ and $q\ge 18\lceil\frac{n}{\delta_0}\rceil+2$ if $\delta_0=3,4$. Especially, the sub-packetization level of the MDS array code $\mathbbmss{G}$ is $\delta^{\lceil\frac{n}{\delta_0}\rceil}$ with $\delta_0=2,3,4$, where $\delta=\mathrm{lcm}(\delta_0,\delta_1,\cdots,\delta_{m-1})$.
	\end{Theorem}
	
	\begin{Remark}
		Note that the generic construction method and Algorithm \ref{Alg_repair_for_all_nodes} have wide potential applications. For example, it can be verified that both the YB codes 1 and 2 in \cite{Barg_1} are TMDS array codes and can be chosen as the base code. However, the resultant codes are not as good as the code $\mathbbmss{G}$ particularly in terms of the sub-packetization level, because the sub-packetization levels of YB codes 1 and 2 in \cite{Barg_1} are much larger than that of the VBK code in \cite{kumar}, i.e, the base code of $\mathbbmss{G}$. Therefore, we do not present the two resultant codes in this paper.
	\end{Remark}

	\section{Comparisons}
	In this section, we give  comparisons  of some key parameters among the  proposed MDS array code $\mathbbmss{G}$  and some existing  notable  MDS  codes with $\delta_{[0:m)}$-optimal repair property for all nodes, where $\delta_{[0:m)}=\{\delta_0,\delta_1,\cdots,\delta_{m-1}\}$. 
	
	Table \ref{Table comp} compares the details of these codes, while Table \ref{Table comp1} - \ref{Table comp3} compare the new MDS array code $\mathbbmss{G}$, YB codes 3 and 4 in terms of the sub-packetization level,  the smallest possible size of field  with characteristic two, and the storage capacity ($N\log q$) for $\delta_0=2,3$ and $4$, respectively. From these tables, we see that the proposed MDS array code $\mathbbmss{G}$ has the following advantages:
	\begin{itemize}
		\item The new MDS array code $\mathbbmss{G}$ has the $\delta_{[0:m)}$-optimal access property for all nodes.
	
		\item Compared with YB code 3, the new $(n,k)$ MDS array code $\mathbbmss{G}$ has a smaller finite field size under the same parameters $n,k$ and set $\delta_{[0:m)}$, but do not possess the optimal update property.
		
		\item In contrast to YB codes 3 and 4 with the same $n,k$ and set $\delta_{[0:m)}$, the sub-packetization level of code $\mathbbmss{G}$ is much smaller than that of YB codes 3 and 4. More precisely,
		    \begin{enumerate}
		    	\item For $\delta_0=2,3,4$, the sub-packetization level of code $\mathbbmss{G}$ is decreased by a factor of $\delta^{n-\lceil\frac{n}{\delta_0}\rceil}$ in contrast to YB codes 3 and 4;
		    	\item For $\delta_0>4$, consider the new code $\mathbbmss{G}$ with the $(\{4\}\cup \delta_{[0:m)})$-optimal access property for all nodes, its sub-packetization level $N=(\mathrm{lcm}(4,\delta))^{\lceil\frac{n}{4} \rceil}$ is decreased by a factor of $\eta$ in contrast to YB codes 3 and 4, where
		    	\begin{eqnarray*}
		    		\eta=\left\{\begin{array}{ll}
		    			\delta^{n-\lceil\frac{n}{4}\rceil}, & \mathrm{if~} 4\mid \delta\\
		    			\frac{\delta^{n-\lceil\frac{n}{4}\rceil}}{2^{\lceil \frac{n}{4}\rceil}}, & \mathrm{if~} 2\mid \delta \mathrm{~and~}4\nmid \delta\\
		    			\frac{\delta^{n-\lceil\frac{n}{4}\rceil}}{4^{\lceil \frac{n}{4}\rceil}}, & \mathrm{otherwise}
		    		\end{array}\right.
		    	\end{eqnarray*}
		    	Moreover, it supports one more repair degree than YB codes 3 and 4 in this case.
		    \end{enumerate}
		\item Compared with YB codes 3 and 4, the field size of new code $\mathbbmss{G}$ is smaller than that of YB code 3 but at most $6$ times larger than that of YB code 4. Since the sub-packetization level of code $\mathbbmss{G}$ is decreased logarithmically with the code length $n$ and value $\delta_0$, thus the total storage ($N \log q$ bits) at each node of our new code $\mathbbmss{G}$  is much smaller than those of YB codes 3 and 4 under the same condition that all of them are constructed over the smallest possible finite field $\mathbb{F}_q$, as shown in Tables \ref{Table comp1}-\ref{Table comp3}.
		
	\end{itemize}

	\begin{table}[htbp]
		\begin{center}
			\caption{A comparison of some key parameters among the $(n,k)$ MDS array codes $\mathbbmss{G}$, YB codes 3 and 4.}\label{Table comp}
			\begin{tabular}{|c|c|c|c|c|}
				\hline
				&Sub-packetization& \multirow{2}{*}{Field size}   & \multirow{2}{*}{Remark} & \multirow{2}{*}{Reference} \\
				&level $N$  &    & &\\
				\hline\hline
				\tabincell{c}{New MDS \\array code $\mathbbmss{G}$}  & $
				\left\{\begin{array}{ll}
				    \delta^{\lceil\frac{n}{\delta_0}\rceil}, & \mathrm{if~}\delta_0 \in [2:5)\\
				    (\mathrm{lcm}(4,\delta))^{\lceil\frac{n}{4} \rceil}, & \mathrm{if~}\delta_0\ge 5
				\end{array}\right.$ & $q\ge \left\{\begin{array}{ll}
					6\lceil\frac{n}{2}\rceil+2, & \mathrm{if~}\delta_0=2\\
					18 \lceil\frac{n}{\delta_0}\rceil+2, & \mathrm{if~}\delta_0=3,4\\
					18 \lceil\frac{n}{4}\rceil+2, & \mathrm{if~}\delta_0 \ge 5\\
				\end{array}\right.$  & Optimal access & Theorem \ref{theorem G1 obtained by algorithm}\\
				\hline\hline
				\tabincell{c}{YB code 3}  & $\delta^n$ & $q \ge \delta n$ & Optimal update &  \cite{kumar} \\
				\hline
				\tabincell{c}{YB code 4}  & $\delta^n$ & $q\ge n+1$ & Optimal access & \cite{kumar}\\
				\hline
			\end{tabular}
		\end{center}
	\end{table}
	
	\begin{table}[htbp]
		\begin{center}
			\caption{A comparison of some parameters among the $(24,20)$ MDS array codes $\mathbbmss{G}$, YB codes 3 and 4 for $\delta_0=2$, where we set the finite field size as the power of 2}\label{Table comp1}
			\begin{tabular}{|c|c|c|c|c|c|}
				\hline
				&  \multirow{2}{*}{Set of $\delta_{[0:m)}$} & Sub-packetization  &  The finite field   &  Storage capacity  & \multirow{2}{*}{$\frac{\textrm{Storage capacity}}{\textrm{Storage capacity of YB code 4}}$} \\
				& & level $N$ & size $q$ & ($N\log q$) & \\
				\hline\hline
				New code $\mathbbmss{G}$  & \multirow{3}{*}{$\{2,3\}$} &  $6^{12}$ & $2^7$ & $7\times 6^{12}$ & $6.43\times 10^{-10}$\\
				YB code 3 & & $6^{24}$ & $2^8$ & $8\times 6^{24}$ & 1.6\\
				YB code 4 & & $6^{24}$ & $2^5$ & $5\times 6^{24}$ & 1\\
				\hline
				New code $\mathbbmss{G}$  & \multirow{3}{*}{$\{2,4\}$} & $4^{12}$ & $2^7$ & $7\times 4^{12}$ & $8.34\times 10^{-8}$\\
				YB code 3 & & $4^{24}$ & $2^7$ & $7\times 4^{24}$  & 1.4\\
				YB code 4 & & $4^{24}$ & $2^5$ & $5\times 4^{24}$  & 1\\
				\hline
				New code $\mathbbmss{G}$ & \multirow{3}{*}{$\{2,3,4\}$} & $12^{12}$ & $2^7$ & $7\times 12^{12}$ & $1.57\times 10^{-13}$ \\
				YB code 3  & & $12^{24}$ & $2^9$ & $9\times 12^{24}$ & 1.8\\
				YB code 4 & & $12^{24}$ & $2^5$ & $5\times 12^{24}$  & 1\\
				\hline
			\end{tabular}
		\end{center}
	\end{table}
	
	\begin{table}[htbp]
		\begin{center}
			\caption{A comparison of some parameters among the $(24,19)$ MDS array codes $\mathbbmss{G}$, YB codes 3 and 4 for $\delta_0=3$, where we set the finite field size as the power of 2}\label{Table comp2}
			\begin{tabular}{|c|c|c|c|c|c|}
				\hline
				&  \multirow{2}{*}{Set of $\delta_{[0:m)}$} & Sub-packetization  &  The finite field   &  Storage capacity  & \multirow{2}{*}{$\frac{\textrm{Storage capacity}}{\textrm{Storage capacity of YB code 4}}$} \\
				& & level $N$ & size $q$ & ($N\log q$) & \\
				\hline\hline
				New code $\mathbbmss{G}$  & \multirow{3}{*}{$\{3,4\}$} &  $12^{8}$ & $2^8$ & $7\times 12^{8}$ & $7.57\times 10^{-18}$ \\
				YB code 3 & &  $12^{24}$ & $2^9$ & $9\times 12^{24}$ & 1.8\\
				YB code 4 & &  $12^{24}$ & $2^5$ & $5\times 12^{24}$ &  1\\
				\hline
				New code $\mathbbmss{G}$ & \multirow{3}{*}{$\{3,5\}$} & $15^{8}$ & $2^8$ & $8\times 15^{8}$ & $2.44\times 10^{-19}$\\
				YB code 3 & & $15^{24}$ & $2^9$ & $9\times 15^{24}$ & 1.8\\
				YB code 4 & & $15^{24}$ & $2^5$ & $5\times 15^{24}$ & 1\\
				\hline			
				New code $\mathbbmss{G}$ & \multirow{3}{*}{$\{3,4,5\}$} & $60^{8}$ & $2^8$ & $8\times 60^{8}$ & $5.67\times 10^{-29}$\\
				YB code 3 & & $60^{24}$ & $2^{11}$ & $11\times 60^{24}$ & 2.2\\
				YB code 4 & & $60^{24}$ & $2^5$ & $5\times 60^{24}$ & 1\\
				\hline
			\end{tabular}
		\end{center}
	\end{table}
	
	\begin{table}[htbp]
		\begin{center}
			\caption{A comparison of some parameters among the $(24,18)$ MDS array codes $\mathbbmss{G}$, YB codes 3 and 4 for $\delta_0=4$, where we set the finite field size as the power of 2}\label{Table comp3}
			\begin{tabular}{|c|c|c|c|c|c|}
				\hline
				&  \multirow{2}{*}{Set of $\delta_{[0:m)}$} & Sub-packetization  &  The finite field   &  Storage capacity  & \multirow{2}{*}{$\frac{\textrm{Storage capacity}}{\textrm{Storage capacity of YB code 4}}$} \\
				& & level $N$ & size $q$ & ($N\log q$) & \\
				\hline\hline
				New code $\mathbbmss{G}$  & \multirow{3}{*}{$\{4,5\}$} &  $20^{6}$ & $2^7$ & $7\times 20^{6}$ & $5.34\times 10^{-24}$\\
				YB code 3 & & $20^{24}$ & $2^9$ & $9\times 20^{24}$ & 1.8\\
				YB code 4 & &  $20^{24}$ & $2^5$ & $5\times 20^{24}$ & 1 \\
				\hline
				New code $\mathbbmss{G}$  & \multirow{3}{*}{$\{4,6\}$} & $12^{6}$ & $2^7$ & $7\times 12^{6}$ & $5.26\times 10^{-20}$\\
				YB code 3 & & $12^{24}$ & $2^9$ & $9\times 12^{24}$ & 1.8\\
				YB code 4 & & $12^{24}$ & $2^5$ & $5\times 12^{24}$ & 1\\
				\hline
				New code $\mathbbmss{G}$  & \multirow{3}{*}{$\{4,5,6\}$} & $60^{6}$ & $2^7$ & $7\times 60^{6}$ & $1.38\times 10^{-32}$ \\
				YB code 3 & & $60^{24}$ & $2^{11}$ & $11\times 60^{24}$ & 2.2 \\
				YB code 4 & & $60^{24}$ & $2^5$ & $5\times 60^{24}$ & 1\\
				\hline
			\end{tabular}
		\end{center}
	\end{table}

	\section{Conclusion}
	In this paper, we provided a generic construction method and further proposed an algorithm that can transform an existing TMDS array code with $\delta_0$-optimal repair property for all nodes into a new MDS array code with all nodes having $\delta_{[0:m)}$-optimal repair property, where $1<\delta_0<\delta_1<\cdots<\delta_{m-1}\le r$. A new explicit construction of high-rate MDS array code $\mathbbmss{G}$ is obtained by directly applying the algorithm to VBK code, where each node of the new code $\mathbbmss{G}$ has the $\delta_{[0:m)}$-optimal access property. The comparisons show that the new code $\mathbbmss{G}$ outperforms existing MDS array codes (i.e., YB codes 3 and 4) in terms of the field size and/or the sub-packetization level under the same parameters $n,k$ and subset $\delta_{[0:m)}$ of $[2:r]$. Extending our generic construction method and specific algorithm to any MDS array codes with $\delta_0$-optimal repair property for all nodes is part of our ongoing work.

	\appendices	
	
	\section{Proof of Property \ref{Pro the number of elements of Piz}}\label{appen the proof of property 1}
	
	All the three properties rely on a fact from Lines 2  and 6 of Algorithm \ref{Alg the set Pij} that
	\begin{eqnarray}\label{Eqn_Pw}
		\mathcal{P}_{i,j}=\bigcup\limits_{a=0}^{l_j-1}\mathcal{P}_{i,j}^{(a)}=\{\mathbf{f}_{i}^{(a)}[u]|a\in [l_j:l_{j-1}),u\in [0:\delta_0)\}\cup\bigcup\limits_{a=l_j}^{l_{j-1}-1}(\mathcal{P}_{i,1}^{(a)}\cup\cdots\cup\mathcal{P}_{i,j-1}^{(a)}), ~j\in [1:m).
	\end{eqnarray}
	
	Firstly, we prove P1 by induction. 
	
	i) If $j=1$, then P1 is a direct consequence of Lines 2  and 4 of Algorithm \ref{Alg the set Pij}. 
	
	ii) Suppose that  P1 holds for $j=w$, where $1\le w<m-1$, i.e.,
	\begin{equation}\label{Eqn_P1_j=w}
		\mathcal{P}_{i,1}\cup \mathcal{P}_{i,2} \cdots \cup \mathcal{P}_{i,w}=\bigcup\limits_{a=0}^{l_w-1}(\mathcal{P}_{i,1}^{(a)}\cup\mathcal{P}_{i,2}^{(a)}\cup\cdots\cup\mathcal{P}_{i,w}^{(a)})=\{\mathbf{f}_i^{(a)}[u]|a\in  [l_w:l_0),u\in [0:\delta_0)\}.
	\end{equation}
	Then, it follows from \eqref{Eqn_Pw} and \eqref{Eqn_P1_j=w} that
	\begin{eqnarray*}
		\mathcal{P}_{i,1}\cup \mathcal{P}_{i,2} \cdots \cup \mathcal{P}_{i,w+1}&=&\{\mathbf{f}_{i}^{(a)}[u]|a\in [l_{w+1}:l_{w}),u\in [0:\delta_0)\}\cup \{\mathbf{f}_i^{(a)}[u]|a\in  [l_w:l_0),u\in [0:\delta_0)\}\\
		&=&\{\mathbf{f}_{i}^{(a)}[u]|a\in [l_{{w+1}}:l_{0}),u\in [0:\delta_0)\}
	\end{eqnarray*}
	and
	\begin{eqnarray*}
		&~&\bigcup\limits_{a=0}^{l_{w+1}-1}(\mathcal{P}_{i,1}^{(a)}\cup\mathcal{P}_{i,2}^{(a)}\cup\cdots\cup\mathcal{P}_{i,{w+1}}^{(a)})\\
		&=&(\bigcup\limits_{a=0}^{l_{w+1}-1} \mathcal{P}_{i,{w+1}}^{(a)}) \cup \bigcup\limits_{a=0}^{l_{w+1}-1}(\mathcal{P}_{i,1}^{(a)}\cup\mathcal{P}_{i,2}^{(a)}\cup\cdots\cup\mathcal{P}_{i,{w}}^{(a)})\\
		&=&\{\mathbf{f}_{i}^{(a)}[u]|a\in [l_{{w+1}}:l_{w}),u\in [0:\delta_0)\}\cup \bigcup\limits_{a=l_{w+1}}^{l_{w}-1}(\mathcal{P}_{i,1}^{(a)}\cup\cdots\cup\mathcal{P}_{i,w}^{(a)})\cup \bigcup\limits_{a=0}^{l_{w+1}-1}(\mathcal{P}_{i,1}^{(a)}\cup\mathcal{P}_{i,2}^{(a)}\cup\cdots\cup\mathcal{P}_{i,w}^{(a)})\\
		&=&\{\mathbf{f}_{i}^{(a)}[u]|a\in [l_{{w+1}}:l_{w}),u\in [0:\delta_0)\}\cup \bigcup\limits_{a=0}^{l_{w}-1}(\mathcal{P}_{i,1}^{(a)}\cup\mathcal{P}_{i,2}^{(a)}\cup\cdots\cup\mathcal{P}_{i,w}^{(a)})\\
		&=&\{\mathbf{f}_{i}^{(a)}[u]|a\in [l_{{w+1}}:l_{0}),u\in [0:\delta_0)\},
	\end{eqnarray*}
	i.e., P1 also holds for $j=w+1$.
	
	Secondly, we prove P2. When $z=j-1$, P2 is obvious from \eqref{Eqn_Pw}. For $z\le j-2$, by applying \eqref{Eqn_Pw} $j-z$ times we have
	\begin{eqnarray*}
		\mathcal{P}_{i,j}&=&\{\mathbf{f}_{i}^{(a)}[u]|a\in [l_{{j}}:l_{j-1}),u\in [0:\delta_0)\}\cup \bigcup\limits_{a=l_{j}}^{l_{j-1}-1}(\mathcal{P}_{i,1}^{(a)}\cup\mathcal{P}_{i,2}^{(a)}\cup\cdots\cup\mathcal{P}_{i,j-1}^{(a)})\\
		&\subseteq& \{\mathbf{f}_{i}^{(a)}[u]|a\in [l_{{j}}:l_{j-1}),u\in [0:\delta_0)\}\cup \mathcal{P}_{i,j-1}\cup \bigcup\limits_{a=l_{j}}^{l_{j-1}-1}(\mathcal{P}_{i,1}^{(a)}\cup\mathcal{P}_{i,2}^{(a)}\cup\cdots\cup\mathcal{P}_{i,j-2}^{(a)})\\
		& = & \{\mathbf{f}_{i}^{(a)}[u]|a\in [l_{{j}}:l_{j-2}),u\in [0:\delta_0)\}\cup \bigcup\limits_{a=l_{j}}^{l_{j-2}-1}(\mathcal{P}_{i,1}^{(a)}\cup\mathcal{P}_{i,2}^{(a)}\cup\cdots\cup\mathcal{P}_{i,j-2}^{(a)})\\
		& \vdots &\\
		& \subseteq & \{\mathbf{f}_{i}^{(a)}[u]|a\in [l_j:l_{z+1}),u\in [0:\delta_0)\}\cup \mathcal{P}_{i,z+1}\cup \bigcup\limits_{a=l_{j}}^{l_{z+1}-1}(\mathcal{P}_{i,1}^{(a)}\cup\mathcal{P}_{i,2}^{(a)}\cup\cdots\cup\mathcal{P}_{i,z}^{(a)})\\
		& =& \{\mathbf{f}_{i}^{(a)}[u]|a\in [l_j:l_z),u\in [0:\delta_0)\}\cup \bigcup\limits_{a=l_{j}}^{l_z-1}(\mathcal{P}_{i,1}^{(a)}\cup\mathcal{P}_{i,2}^{(a)}\cup\cdots\cup\mathcal{P}_{i,z}^{(a)}).
	\end{eqnarray*}
	
	Thirdly, we prove P3 for $1\le j<m$ by induction. 
	
i)	
If $j=1$, then P3 follows from Lines 2 and 4 of Algorithm \ref{Alg the set Pij}. 

ii) Suppose that  P3 holds for all $j=1,2,\cdots,w$ where $1\le w<m-1$, i.e.,
$|\mathcal{P}_{i,j} |=(\delta_j-\delta_{j-1})l_j$ and $|\mathcal{P}_{i,j}^{(a)} |=\delta_j-\delta_{j-1}$  for $a\in [0:l_j)$ and $j=1,2,\cdots,w$.
We next prove that P3 holds for $j=w+1$. 

Given $i\in \mathcal{G}$ and $w\in [1: m)$,  by P1 and the hypothesis, we have
	\begin{eqnarray}
	  \delta_0(l_0-l_w)&=&|\{\mathbf{f}_i^{(a)}[u]|a\in  [l_w:l_0),u\in [0:\delta_0)\}|\notag \\
	    &=& |\bigcup\limits_{a=0}^{l_w-1}(\mathcal{P}_{i,1}^{(a)}\cup \mathcal{P}_{i,2}^{(a)}\cup \cdots \cup \mathcal{P}_{i,w}^{(a)})|\notag\\
	    &\le &\sum\limits_{a=0}^{l_w-1}(|\mathcal{P}_{i,1}^{(a)}|+|\mathcal{P}_{i,2}^{(a)}|+\cdots+|\mathcal{P}_{i,w}^{(a)}|)\notag\\
	    &=&l_w[(\delta_1-\delta_0)+(\delta_2-\delta_1)+\cdots+(\delta_w-\delta_{w-1})]\notag\\
	    &=&l_w(\delta_w-\delta_0)\label{Eqn useful for proving P3}\\
	    &=&\delta_0(l_0-l_w)\notag,
	\end{eqnarray}
which implies 
\begin{equation*}
|\bigcup\limits_{a=0}^{l_w-1}(\mathcal{P}_{i,1}^{(a)}\cup \mathcal{P}_{i,2}^{(a)}\cup \cdots \cup \mathcal{P}_{i,w}^{(a)})|=\sum\limits_{a=0}^{l_w-1}(|\mathcal{P}_{i,1}^{(a)}|+|\mathcal{P}_{i,2}^{(a)}|+\cdots+|\mathcal{P}_{i,w}^{(a)}|),    
\end{equation*}
i.e., 
	\begin{equation}\label{Eqn_disjointP}
\mathcal{P}_{i,s}^{(a)}\cap  \mathcal{P}_{i,s'}^{(a')}=\emptyset \mbox{~for~}s,s'\in [1:w]\mbox{~and~}a,a'\in [0:l_w) \mbox{~with~}	(s,a)\ne (s',a').    
	\end{equation}
	
	Then, by  \eqref{eqn notation l},  \eqref{Eqn_Pw} and \eqref{Eqn_disjointP}, similarly to \eqref{Eqn useful for proving P3} we obtain
	\begin{eqnarray*}
		|\mathcal{P}_{i,w+1}|
		&=&|\{\mathbf{f}_{i}^{(a)}[u]|u\in [0:\delta_0),a\in [l_{w+1}:l_w)\}|+(l_w-l_{w+1})(\delta_w-\delta_0)\\
		&=&\delta_0(l_w-l_{w+1})+(l_w-l_{w+1})(\delta_w-\delta_0)\\
		&=&\delta_w(l_w-l_{w+1})\\
		&=&(\delta_{w+1}-\delta_w)l_{w+1},
	\end{eqnarray*}
	which together with Line 7 of Algorithm \ref{Alg the set Pij} shows P3 holds for $j=w+1$. This finishes the proof.

	\section{Proof of Lemma \ref{lem the important for the repair of remainder node}}\label{appen the proof of lem 3}
	
	Given $a\in [l_{w+1}:l_w)$ with $1\le w<m$, according to P3, we set
	\begin{eqnarray*}
		\mathcal{P}_{j,s}^{(a)}=\left\{\mathbf{f}_j^{(b_{\delta_{s-1}-\delta_0})}[u_{\delta_{s-1}-\delta_0}],\cdots,\mathbf{f}_j^{(b_{\delta_s-\delta_0-1})}[u_{\delta_s-\delta_0-1}]\right\} \textrm{ for all } s\in [1:w], j\in \mathcal{G},
	\end{eqnarray*}
	which together with \eqref{eqn_specific_form_of_P}  gives
	\begin{eqnarray}\label{eqn roughly form of P}
		S_{i,\delta_z}\mathbf{P}_{t,j}^{(a)}=\sum\limits_{v=0}^{\delta_w-\delta_0-1}S_{i,\delta_z} K_{t,j,v}\mathbf{f}_j^{(b_v)}[u_v]=\sum\limits_{v=0}^{\delta_w-\delta_0-1}S_{i,\delta_z} K_{t,j,v}\Phi_{\alpha,u_v}\mathbf{f}_j^{(b_v)}
	\end{eqnarray}
	for $i\in [0:n)\backslash\mathcal{G}$, $j\in \mathcal{G}$, $t\in [0:r)$ and $z\in [0:m)$, where $\{(b_p,u_p)|p\in [0:\delta_w-\delta_0)\}\subset [0:l_0)\times [0:\delta_0)$, and the second equality holds because of \eqref{eqn vector b[u]}. Moreover, it follows from P1 and Lines 2, 7 of Algorithm \ref{Alg the set Pij} that 
	\begin{eqnarray*}
		\mathcal{P}_{j,s}^{(a)} \subseteq  \{\mathbf{f}_j^{(b)}[u]|b\in [l_w:l_0),u \in [0:\delta_0)\}, s\in [1:w],
	\end{eqnarray*}
	which implies that $b_p\in [l_w:l_0)$ for any $p\in [0:\delta_w-\delta_0)$.

	Then, for any $a\in [l_{w+1}:l_w)$ with $w\in [1:m)$, any $z\in [0:m)$, $t\in [0:r)$, $i\in [0:n)\backslash\mathcal{G}$ and $j\in \mathcal{G}$, according to \eqref{eqn roughly form of P}, we are able to compute the vector $S_{i,\delta_z}\mathbf{P}_{t,j}^{(a)}$ from the data in set $\{R_{i,\delta_z}\mathbf{f}_j^{(b)}|b\in [l_w:l_0)\}$ based on C3.
	
	\section{Proof of Lemma \ref{Theorem the recursion result}}\label{appen the proof of lem 4}
	Our task is to prove that the matrices $(K'_{t,i,v})_{t\in [0:r),i\in \mathcal{G}',v\in [0:\delta_{m-1}-\delta_0)}$ given in \eqref{eqn key matrix for recursion} are the key matrices such that the nodes in set $\mathcal{G}'$ of new code $\mathbbmss{C}_2$ satisfy C1-C3.

	Firstly, we verify that the nodes in set $\mathcal{G}'$ of new code $\mathbbmss{C}_2$ satisfy C1. For any $t\in [0:r)$, $v\in [0:\delta_{m-1}-\delta_0)$ and $i,j\in \mathcal{G}'$ with $i\ne j$, by  \eqref{eqn repair,select matrix of remainder nodes} and \eqref{eqn key matrix for recursion}, we then have
	\begin{eqnarray*}
		S_{i,\delta_0}'K_{t,j,v}'=\mbox{blkdiag}(S_{i,\delta_0}K_{t,j,v},S_{i,\delta_0}K_{t,j,v},\cdots,S_{i,\delta_0}K_{t,j,v})_{l_0}.
	\end{eqnarray*}
	Recall that the nodes in set $\mathcal{G}'$ of base code $\mathbbmss{C}_0$ satisfy C1, thus the nodes in set $\mathcal{G}'$ of new code $\mathbbmss{C}_2$ also satisfy C1.

	Secondly,  we check that  the nodes in set $\mathcal{G}'$ of new code $\mathbbmss{C}_2$ satisfy C2.  Given a node $j\in [0:n)$ of code $\mathbbmss{C}_2$,
	let $A_{t,j}'~(t\in [0:r))$ be its parity-check matrix and $\mathbf{g}_j=((\mathbf{f}_j^{(0)})^\top,(\mathbf{f}_j^{(1)})^\top,\cdots,(\mathbf{f}_j^{(l_0-1)})^\top)^\top$ be the 
	data stored at node $j$. Then, for any $j\in \mathcal{G}'$, by \eqref{eqn t-PCG of code C3},
	\begin{eqnarray*}
		A_{t,j}'\mathbf{g}_j=\left\{\begin{array}{ll}
			\mbox{blkdiag}(A_{t,j},A_{t,j},\cdots,A_{t,j})_{l_0}\mathbf{g}_j+\left(\begin{array}{c}
				\mathbf{P}_{t,j}^{(0)}\\
				\mathbf{P}_{t,j}^{(1)}\\
				\vdots\\
				\mathbf{P}_{t,j}^{(l_0-1)}\\
			\end{array}\right), & \mbox{if~}j\in \mathcal{G},\\
			\mbox{blkdiag}(A_{t,j},A_{t,j},\cdots,A_{t,j})_{l_0}\mathbf{g}_j,& \mbox{otherwise}.
		\end{array}\right.
	\end{eqnarray*}
	Thus, for any $i\in \mathcal{G}'$ and $j\in [0:n)\backslash\{i\}$,  by   \eqref{repair node requirement3}, \eqref{eqn the repair,select matrix of C3 for goal nodes} and \eqref{eqn repair,select matrix of remainder nodes},  we have
	\begin{eqnarray}\label{eqn the third equation for prove Thm. 5}
		S_{i,\delta_0}'A_{t,i}'=\mbox{blkdiag}(S_{i,\delta_0}A_{t,i},S_{i,\delta_0}A_{t,i},\cdots,S_{i,\delta_0}A_{t,i})_{l_0}  
	\end{eqnarray}
	and 
	\begin{eqnarray}\label{eqn the fourth equation for prove Thm. 5}
		\tilde{A}_{t,j,i,\delta_0}'=\left\{\begin{array}{ll}
			\left(\begin{array}{cccc}
				\tilde{A}_{t,j,i,\delta_0} & \# & \cdots & \#\\
				& \tilde{A}_{t,j,i,\delta_0} & \cdots & \# \\
				& & \ddots & \vdots\\
				& & & \tilde{A}_{t,j,i,\delta_0}\\
			\end{array}\right), & \mbox{if~}j\in \mathcal{G},\\
			\mbox{blkdiag}(\tilde{A}_{t,j,i,\delta_0},\tilde{A}_{t,j,i,\delta_0},\cdots,\tilde{A}_{t,j,i,\delta_0})_{l_0}, &\mbox{otherwise},
		\end{array}\right.
	\end{eqnarray}
	where the case of $j\in \mathcal{G}$ follows from P0 and Lemma \ref{lem the important for the repair of remainder node}, $\tilde{A}_{t,j,i,\delta_0}'$ is the matrix defined in \eqref{repair node requirement3}, and symbol $\#$ denotes some matrices which we do not care about the exact expression.

	Let $M_{i,\mathcal{D}_z}'$ be the matrix defined in \eqref{eqn matrix MijDz},  where one should note that the symbols $A,K,S$ in \eqref{eqn matrix MijDz} are replaced by $A',K',S'$, respectively. According to  \eqref{eqn the third equation for prove Thm. 5} and  \eqref{eqn the fourth equation for prove Thm. 5}, by exachanging some block rows and block columns of matrix $M_{i,\mathcal{D}_z}'$, we obtain  
	\begin{eqnarray*}
		\mbox{Rank}(M_{i,\mathcal{D}_z}')=\mbox{Rank}\left(\left(\begin{array}{cccc}
			M_{i,\mathcal{D}_z} & \# & \cdots & \#\\
			& M_{i,\mathcal{D}_z} & \cdots & \#\\
			& & \ddots & \vdots\\
			& & & M_{i,\mathcal{D}_z}
		\end{array}\right)\right),
	\end{eqnarray*}
	which  finishes the proof of C2, together with the fact that the nodes in $\mathcal{G}'$ of base code satisfy C2 as well.

	Finally, we show that the nodes in $\mathcal{G}'$ of new code $\mathbbmss{C}_2$ satisfy C3. Note that the sub-packetization level of new code $\mathbbmss{C}_2$ is $l_0 \alpha N$, let $\alpha'=l_0 \alpha$. Then by \eqref{eqn the matrix Phi}, 
	\begin{eqnarray}\label{eqn the matrix for Phi with a=a'}
		\Phi_{\alpha',u}=\mbox{blkdiag}(\Delta_u,\Delta_u,\cdots,\Delta_u)_{\alpha'}=\mbox{blkdiag}(\Phi_{\alpha,u},\Phi_{\alpha,u},\cdots,\Phi_{\alpha,u})_{l_0}
	\end{eqnarray} 
	due to $\alpha'=l_0 \alpha$. Thus by  \eqref{eqn repair,select matrix of remainder nodes}, \eqref{eqn key matrix for recursion}  and \eqref{eqn the matrix for Phi with a=a'}, for any $j\in \mathcal{G}'$, $i\in [0:n)\backslash \mathcal{G}'$, $t\in [0:r)$, $v\in [0:\delta_{m-1}-\delta_0)$ and $u\in [0:\delta_0)$, we get
	\begin{eqnarray}\label{eqn the second equation for prove Thm. 5}
		S_{i,\delta_z}'K_{t,j,v}'\Phi_{\alpha',u}=\left\{\begin{array}{ll}
			\left(\begin{array}{c;{2pt/2pt}c}
				\underbrace{\begin{array}{ccc}
						S_{i,\delta_0}K_{t,j,v}\Phi_{\alpha,u} & &\\
						& \ddots &\\
						& & S_{i,\delta_0}K_{t,j,v}\Phi_{\alpha,u}
				\end{array}}_{l_z\times l_z} & \underbrace{\begin{array}{ccc}
						\mathbf{0}_{\alpha N'\times \alpha N} & \cdots & \mathbf{0}_{\alpha N'\times \alpha N}\\
						\vdots & \vdots & \vdots\\
						\mathbf{0}_{\alpha N'\times \alpha N} & \cdots & \mathbf{0}_{\alpha N'\times \alpha N}
				\end{array}}_{l_z\times (l_0-l_z)}
			\end{array}\right), & \textrm{ if }i\in \mathcal{G},\\
			\mbox{blkdiag}(S_{i,\delta_z}K_{t,j,v}\Phi_{\alpha,u},S_{i,\delta_z}K_{t,j,v}\Phi_{\alpha,u},\cdots,S_{i,\delta_z}K_{t,j,v}\Phi_{\alpha,u})_{l_0}, & \textrm{otherwise},
		\end{array}	\right.
	\end{eqnarray}
	where one should note that \eqref{eqn the second equation for prove Thm. 5} holds for $i\not\in \mathcal{G}$ if and only if the node $i$ has $\delta_z$-optimal repair property in new code $\mathbbmss{C}_2$.

	Combining \eqref{eqn the repair,select matrix of C3 for goal nodes},  \eqref{eqn repair,select matrix of remainder nodes}  and \eqref{eqn the second equation for prove Thm. 5}, we obtain
	\begin{eqnarray*}
		\mbox{Rank}\left(\left(\begin{array}{c}
			R_{i,\delta_z}'\\
			S_{i,\delta_z}'K_{t,j,v}'\Phi_{\alpha',u}
		\end{array}\right)\right)=l_z\cdot \textrm{Rank}\left(\left(\begin{array}{c}
			R_{i,\delta_0}\\
			S_{i,\delta_0}K_{t,j,v}\Phi_{\alpha,u}\\
		\end{array}\right)\right)=l_z\cdot \frac{\alpha N}{\delta_0}=\frac{l_0\alpha N}{\delta_z} & \textrm{ if }i\in \mathcal{G}
	\end{eqnarray*}
	and
	\begin{eqnarray*}
		\mbox{Rank}\left(\left(\begin{array}{c}
			R_{i,\delta_z}'\\
			S_{i,\delta_z}'K_{t,j,v}'\Phi_{\alpha',u}
		\end{array}\right)\right)=l_0\cdot \textrm{Rank}\left(\left(\begin{array}{c}
			R_{i,\delta_z}\\
			S_{i,\delta_z}K_{t,j,v}\Phi_{\alpha,u}\\
		\end{array}\right)\right)=\frac{l_0\alpha N}{\delta_z} & \textrm{ if }i\in [0:n)\backslash( \mathcal{G}\cup \mathcal{G}')
	\end{eqnarray*}
	since the nodes in $\mathcal{G}'$ of base code $\mathbbmss{C}_0$ satisfy C3, where we make use of the fact $\frac{l_z}{\delta_0}=\frac{l_0}{\delta_z}$ from \eqref{eqn notation l}.
	That is,  the nodes in $\mathcal{G}'$ of new code $\mathbbmss{C}_2$ satisfy C3.

	\section{Proofs of Lemmas \ref{lem the first helpful lemma for the proof of Thm. 10}, \ref{lem alignment interference for VBK code} and
	 \ref{lem the helpful lemma for the proof of Thm. 10}
	}\label{appen proof of lems 5,6,8}
	Before proving those three lemmas, let us introduce some necessary notations. Note that $N'=\frac{N}{\delta_0}=\delta_0^{\tau-1}$. For a given $a=(a_{\tau-2},a_{\tau-3},\cdots,a_0)\in [0:N')$, define
	\begin{eqnarray}\label{eqn function phi}
		\varphi(a,x,u)=(a_{\tau-2},\cdots,a_x,u,a_{x-1},\cdots,a_0)
	\end{eqnarray}
	for any $x\in [0:\tau)$, $u\in [0:\delta_0)$, i.e., insert the value $u$ between the $x$-digit and $(x-1)$-digit of the vector $a$ if $x<\tau-1$, and insert the value $u$ before the $(\tau-1)$-digit if $x=\tau-1$. Then by \eqref{Eqn Vt}, we  easily get 
	\begin{eqnarray}\label{Eqn Vt for the equiv form}
		V_{x,u}(a,:)=e_{\varphi(a,x,u)}, & a\in [0:N').
	\end{eqnarray}
	
	By \eqref{Eqn_SB}, \eqref{Eqn a(i,u)} and \eqref{eqn function phi}, we have the following simple facts.
	\begin{Fact}
	1) For $0\le a=(a_{\tau-1},a_{\tau-2},\cdots,a_0),b=(b_{\tau-1},b_{\tau-2},\cdots,b_0)<N$,
		\begin{eqnarray}\label{Eqn_fact_EB}
		e_ae_b^{\top}=\left\{\begin{array}{ll}
		1, & \mathrm{if~} a=b,\\
		0, &	\mathrm{otherwise}.
		\end{array}
		\right.
	\end{eqnarray}
	
		2) For  $a=(a_{\tau-2},a_{\tau-3},\cdots,a_0)\in [0:N')$, $0\le x, \tilde{x}<\tau$ and $0\le u,v<\delta_0$, 
		\begin{eqnarray}\label{Eqn_fact_varphi}
		(\varphi(a,\tilde{x},u))_{x}=\left\{\begin{array}{ll}
				a_x, & \mathrm{if~}x<\tilde{x},\\
				u, & \mathrm{if~}x=\tilde{x},\\
				a_{x-1}, & \mathrm{if~}x>\tilde{x},\\
			\end{array}\right.
		\end{eqnarray}
	and	
		\begin{eqnarray}\label{Eqn_fact_varphi_pi}
                   \pi_\tau(\varphi(a,\tilde{x},u),x,v)=\left\{\begin{array}{ll}
				\varphi(\pi_{\tau-1}(a,x,v),\tilde{x},u), & \mathrm{if~}x<\tilde{x},\\
				\varphi(a,\tilde{x},v), & \mathrm{if~}x=\tilde{x},\\
				\varphi(\pi_{\tau-1}(a,x-1,v),\tilde{x},u), & \mathrm{if~}x>\tilde{x}.\\
			\end{array}\right.
		\end{eqnarray}

	\end{Fact}

	\textbf{\textit{Proof of Lemma \ref{lem the first helpful lemma for the proof of Thm. 10}}}
	
	Clearly, (i)  is true because of \eqref{Eqn Vt} and \eqref{Eqn_fact_EB}.

	Next, we prove (ii) for  $0\le u,v,h<\delta_0$ and $0\le x \ne \tilde{x}<\tau$. Given $a=(a_{\tau-2},a_{\tau-3},\cdots,a_0)\in [0:N')$, on one hand, 
	\begin{eqnarray}
		V_{x,u}(a,:)\cdot(V_{\tilde{x},v}^\top\Delta_h)&=& e_{\varphi(a,x,u)}\cdot(e_{\varphi(0,\tilde{x},v)}^\top,e_{\varphi(1,\tilde{x},v)}^\top,\cdots,e_{\varphi(N'-1,\tilde{x},v)}^\top)\left(\begin{array}{c}
				e_{h\cdot N'}\\
				e_{h\cdot N'+1}\\
				\vdots\\
				e_{h\cdot N'+N'-1}
			\end{array}\right)\notag\\
		&=& e_{\varphi(a,x,u)}\sum\limits_{b=0}^{N'-1}e_{\varphi(b,\tilde{x},v)}^\top e_{h\cdot N'+b}\notag\\
		&=&\sum\limits_{b=0}^{N'-1}(e_{\varphi(a,x,u)}e_{\varphi(b,\tilde{x},v)}^\top)e_{h\cdot N'+b}\notag\\
		&=&\sum_{\varphi(b,\tilde{x},v)=\varphi(a,x,u)} e_{h\cdot N'+b} \label{Eqn ehs}
	\end{eqnarray}
	where the first identity follows from \eqref{eqn the matrix Delta} and \eqref{Eqn Vt for the equiv form}, and  the fourth identity comes from 
	\eqref{Eqn_fact_EB}.
	
For $x\ne \tilde{x}$,  by \eqref{eqn function phi}, we have  $\varphi(a,x,u)=\varphi(b,\tilde{x},v)$ if and only if
	\begin{eqnarray}\label{Eqn_abphi}
		b=\left\{\begin{array}{ll}
			(a_{\tau-2},\cdots,a_{\tilde{x}},a_{\tilde{x}-2},\cdots,a_{x},u,a_{x-1},\cdots,a_0), & \mbox{if~}0\le x<\tilde{x}<\tau,a_{\tilde{x}-1}=v,\\
			(a_{\tau-2},\cdots,a_{x},u,a_{x-1},\cdots,a_{\tilde{x}+1},a_{\tilde{x}-1},\cdots,a_0), & \mbox{if~}0\le \tilde{x}<x<\tau, a_{\tilde{x}}=v.
		\end{array}\right.
	\end{eqnarray}

	On the other hand, applying \eqref{eqn the row vector of matrix Tijuvh}, \eqref{Eqn Vt for the equiv form} and  \eqref{Eqn_abphi}, we get
	\begin{eqnarray}\label{eqn a*V}
		T_{x,\tilde{x},v,h}(a,:) \cdot V_{x,u}&=&\left\{\begin{array}{ll}
			\epsilon_{\bar{a}}V_{x,u}, & \textrm{if~} 0\le x<\tilde{x}<\tau,a_{\tilde{x}-1}=v \textrm{~or~}0\le \tilde{x}<x<\tau,a_{\tilde{x}}=v\\
			\mathbf{0}, & \textrm{otherwise}
		\end{array}\right.\notag\\
		&=&\left\{\begin{array}{ll}
			e_{\varphi(\bar{a},x,u)}, & \textrm{if~} 0\le x<\tilde{x}<\tau,a_{\tilde{x}-1}=v \textrm{~or~}0\le \tilde{x}<x<\tau,a_{\tilde{x}}=v\\
			\mathbf{0}, & \textrm{otherwise}
		\end{array}\right.\notag\\
		&=&\left\{\begin{array}{ll}
			e_{h\cdot N'+b}, & \textrm{if~}\varphi(b,\tilde{x},v)=\varphi(a,x,u),\\
			\mathbf{0}, & \textrm{otherwise},
		\end{array}\right.
	\end{eqnarray}
	where 
	\begin{eqnarray*}
		\bar{a}=\left\{\begin{array}{ll}
			(h,a_{\tau-2},\cdots,a_{\tilde{x}},a_{\tilde{x}-2},\cdots,a_0), & \textrm{if }0\le x<\tilde{x}<\tau,a_{\tilde{x}-1}=v\\
			(h,a_{\tau-2},\cdots,a_{\tilde{x}+1},a_{\tilde{x}-1},\cdots,a_0), & \textrm{if }0\le \tilde{x}<x<\tau,a_{\tilde{x}}=v
		\end{array}\right.
	\end{eqnarray*}
	Collecting \eqref{Eqn ehs} and \eqref{eqn a*V}, we complete the proof.
	
	\vspace{3mm}
	In what follows, we give the proofs of Lemma \ref{lem alignment interference for VBK code} and \ref{lem the helpful lemma for the proof of Thm. 10}, in which we always let $i=\delta_0\tilde{x}+\tilde{y}$ and $j=\delta_0 x+y$, where $0\le \tilde{x},x<\tau$ and $0\le \tilde{y},y<\delta_0$.
	
	\textbf{\textit{Proof of Lemma \ref{lem alignment interference for VBK code}}}
	
For any given $a=(a_{\tau-2},a_{\tau-3},\cdots,a_0)\in [0:N')$, according to \eqref{eqn the parity-matrix of VBK code}, \eqref{eqn coefficients gamma for VBK code} and \eqref{Eqn Vt for the equiv form}, we have
	\begin{eqnarray}\label{eqn e_aQti}
		V_{\tilde{x},\tilde{y}}(a,:)\cdot A_{t,j}&=&e_{\varphi(a,\tilde{x},\tilde{y})}\cdot A_{t,j}\notag\\
		&=&e_{\varphi(a,\tilde{x},\tilde{y})}(\sum\limits_{b=0}^{N-1}\lambda_{j,b_x}^te_b^\top e_b+\sum\limits_{b=0,b_x=y}^{N-1}\sum\limits_{u=0,u\ne y}^{\delta_0-1}\varepsilon_{u,y}\lambda_{j,u}^te_b^\top e_{\pi_\tau(b,x,u)})\nonumber\\
		&=&\left\{\begin{array}{ll}
				\lambda_{j,y}^te_{\varphi(a,\tilde{x},\tilde{y})}+\sum\limits_{u=0,u\ne y}^{\delta_0-1}\varepsilon_{u,y}\lambda_{j,u}^te_{\pi_\tau(\varphi(a,\tilde{x},\tilde{y}),x,u)}, &\mbox{if~}(\varphi(a,\tilde{x},\tilde{y}))_x= y,\\
				\lambda_{j,(\varphi(a,\tilde{x},\tilde{y}))_x}^te_{\varphi(a,\tilde{x},\tilde{y})}, & \mbox{otherwise},
			\end{array}\right.
	\end{eqnarray}
	due to \eqref{Eqn_fact_EB}.  

	
	Particularly, when $x=\tilde{x}$, \eqref{eqn e_aQti} becomes
	\begin{equation*}\label{eqn VtQti for VBK code eq 1}
		V_{\tilde{x},\tilde{y}}(a,:)A_{t,j}=\left\{\begin{array}{ll}
			\lambda_{i,\tilde{y}}^tV_{\tilde{x},\tilde{y}}(a,:)+\sum\limits_{u=0,u\ne \tilde{y}}^{\delta_0-1}\varepsilon_{u,\tilde{y}}\lambda_{i,u}^tV_{\tilde{x},u}(a,:), & \textrm{if~} i=j,\\
			\lambda_{j,\tilde{y}}^tV_{\tilde{x},\tilde{y}}(a,:), & \textrm{if } i\ne j \textrm{ and } x=\tilde{x}
			\end{array}\right.
	\end{equation*}
	since $(\varphi(a,\tilde{x},\tilde{y}))_{\tilde{x}}=\tilde{y}$ and $e_{\pi_\tau(\varphi(a,\tilde{x},\tilde{y}),\tilde{x},u)}=e_{\varphi(a,\tilde{x},u)}=V_{\tilde{x},u}(a,:)$ according to   \eqref{Eqn Vt for the equiv form}, \eqref{Eqn_fact_varphi} and  \eqref{Eqn_fact_varphi_pi}. That is,
	\begin{eqnarray*}\label{eqn VtQti for VBK code}
		V_{\tilde{x},\tilde{y}}A_{t,j}=
		\left\{\begin{array}{ll}
			\lambda_{i,\tilde{y}}^tV_{\tilde{x},\tilde{y}}+\sum\limits_{u=0,u\ne \tilde{y}}^{\delta_0-1}\varepsilon_{u,\tilde{y}}\lambda_{i,u}^tV_{\tilde{x},u}, & \mbox{if~} i=j,\\
			\lambda_{j,\tilde{y}}^tV_{\tilde{x},\tilde{y}}, & \textrm{if } i\ne j \textrm{ and }x=\tilde{x},
		\end{array}\right.
	\end{eqnarray*}
	which together with \eqref{eqn the repair, select matrix of VBK code} implies  
	\begin{eqnarray*}
		S_{i,\delta_0}A_{t,i}=\lambda_{i,\tilde{y}}^tV_{\tilde{x},\tilde{y}}+\sum\limits_{u=0,u\ne \tilde{y}}^{\delta_0-1}\varepsilon_{u,\tilde{y}}\lambda_{i,u}^tV_{\tilde{x},u} 
	\end{eqnarray*}
	and 
	\begin{equation*}
		S_{i,\delta_0}A_{t,j}=\lambda_{j,\tilde{y}}^tR_{i,\delta_0} \mbox{~for~} 0\le i\ne j<n \mbox{~with~} x=\tilde{x},  
	\end{equation*}
	i.e., (i) is true and (ii) holds for $0\le i\ne j<n$ with $x=\tilde{x}$.
	
	Next we prove this lemma for $x\ne \tilde{x}$. Herein we only check the case of $x<\tilde{x}$ since the case of $x>\tilde{x}$ can be proved in a similar manner. In this case, i.e., $(\varphi(a,\tilde{x},\tilde{y}))_x=a_x$ by \eqref{Eqn_fact_varphi}, then \eqref{eqn e_aQti} turns into
	\begin{eqnarray}\label{Eqn the product between V[a] and Atj}
	V_{\tilde{x},\tilde{y}}(a,:)A_{t,j}
		&=&\left\{\begin{array}{ll}
			\lambda_{j,y}^te_{\varphi(a,\tilde{x},\tilde{y})}+\sum\limits_{u=0,u\ne y}^{\delta_0-1}\varepsilon_{u,y}\lambda_{j,u}^te_{\varphi(\pi_{\tau-1}(a,x,u),\tilde{x},\tilde{y})}, &\mbox{if~} a_x= y,\\
			\lambda_{j,a_x}^te_{\varphi(a,\tilde{x},\tilde{y})}, & \mbox{otherwise},
		\end{array}\right.\notag\\
		&=&\left\{\begin{array}{ll}
			(\lambda_{j,y}^tV_{\tilde{x},\tilde{y}}(a,:)+\sum\limits_{u=0,u\ne y}^{\delta_0-1}\varepsilon_{u,y}\lambda_{j,u}^tV_{\tilde{x},\tilde{y}}(\pi_{\tau-1}(a,x,u),:), &\mbox{if~}a_x= y,\\
			\lambda_{j,a_x}^t V_{\tilde{x},\tilde{y}}(a,:), & \mbox{otherwise},
		\end{array}\right.\notag\\
		&=&\left\{\begin{array}{ll}
			(\lambda_{j,y}^t\epsilon_a+\sum\limits_{u=0,u\ne y}^{\delta_0-1}\varepsilon_{u,y}\lambda_{j,u}^t\epsilon_{\pi_{\tau-1}(a,x,u)})V_{\tilde{x},\tilde{y}}, &\mbox{if~}a_x= y,\\
			\lambda_{j,a_x}^t\epsilon_a\cdot V_{\tilde{x},\tilde{y}}, & \mbox{otherwise},
		\end{array}\right.\notag\\
		&=&\epsilon_a(\sum\limits_{b=0}^{N'-1}\lambda_{j,b_x}^t\epsilon_b^\top\epsilon_b+\sum\limits_{b=0,b_x=y}^{N'-1}\sum\limits_{u=0,u\ne y}^{\delta_0-1}\varepsilon_{u,y}\lambda_{j,u}^t\epsilon_b^\top \epsilon_{\pi_{\tau-1}(b,x,u)})V_{\tilde{x},\tilde{y}}\notag\\
		&=&\epsilon_a \tilde{A}_{t,j,i,\delta_0}\cdot V_{\tilde{x},\tilde{y}}\notag\\
		&=&\tilde{A}_{t,j,i,\delta_0}(a,:)\cdot V_{\tilde{x},\tilde{y}}
	\end{eqnarray}
where the first equality follows from \eqref{Eqn_fact_varphi_pi}, the second equality comes from \eqref{Eqn Vt for the equiv form},  and the fourth equality  can be  derived similarly to the third equality in \eqref{eqn e_aQti}. Applying \eqref{eqn the repair, select matrix of VBK code} and \eqref{Eqn the product between V[a] and Atj}, we have $S_{i,\delta_0}A_{t,j}=\tilde{A}_{t,j,i,\delta_0}R_{i,\delta_0}$, which finishes the proof.

	\vspace{3mm}
	
	\textbf{\textit{Proof of Lemma  \ref{lem the helpful lemma for the proof of Thm. 10}}}

	Hereafter we only check the case of $0\le s<p-1$ since the case of $s=p-1$ can be verified  similarly. 	For $i,j\in [0:n)$ with $i\ne j$, define
	\begin{eqnarray}\label{eqn matrix Upsilon}
		 \Upsilon_{j,i}=\left\{\begin{array}{ll}
			\sum\limits_{a=0}^{N'-1}\lambda_{j,a_x}\epsilon_a^\top \epsilon_a, & \textrm{if }x<\tilde{x},\\
			[6pt]
			\lambda_{j,\tilde{y}}I_{N'}, & \textrm{if }x=\tilde{x}\\
			[6pt]
			\sum\limits_{a=0}^{N'-1}\lambda_{j,a_{x-1}}\epsilon_a^\top \epsilon_a, & \textrm{if }x>\tilde{x},
		\end{array}\right.
	\end{eqnarray}
which together with \eqref{eqn matrix Atji for VBK code} implies 	
\small{
\begin{eqnarray}\label{Eqn_Bt}
		\tilde{A}_{t,j,i,\delta_0}\Upsilon_{j,i}&=&\left\{\begin{array}{ll}
			(\sum\limits_{a=0}^{N'-1}\lambda_{j,a_x}^t\epsilon_a^\top \epsilon_a+\sum\limits_{a=0,a_x=y}^{N'-1}\sum\limits_{u=0,u\ne y}^{\delta_0-1}\varepsilon_{u,y}\lambda_{j,u}^t\epsilon_a^\top \epsilon_{\pi_{\tau-1}(a,x,u)})\cdot\sum\limits_{b=0}^{N'-1}\lambda_{j,b_x}\epsilon_b^\top \epsilon_b, & \textrm{if~}x<\tilde{x},\\[6pt]
			(\lambda_{j,\tilde{y}}^tI_{N'})\cdot \lambda_{j,\tilde{y}}I_{N'}, & \textrm{if~}x=\tilde{x},\\[6pt]
			(\sum\limits_{a=0}^{N'-1}\lambda_{j,a_{x-1}}^t\epsilon_a^\top \epsilon_a+\sum\limits_{a=0,a_{x-1}=y}^{N'-1}\sum\limits_{u=0,u\ne y}^{\delta_0-1}\varepsilon_{u,y}\lambda_{j,u}^t\epsilon_a^\top \epsilon_{\pi_{\tau-1}(a,x-1,u)})\cdot \sum\limits_{b=0}^{N'-1}\lambda_{j,b_{x-1}}\epsilon_b^\top \epsilon_b, & \textrm{if~}x>\tilde{x},
		\end{array}\right.\notag\\
		&=&\left\{\begin{array}{ll}
			\sum\limits_{a=0}^{N'-1}(\lambda_{j,a_x}^t\epsilon_a^\top \epsilon_a\sum\limits_{b=0}^{N'-1}\lambda_{j,b_x}\epsilon_b^\top \epsilon_b)+\sum\limits_{a=0,a_x=y}^{N'-1}(\sum\limits_{u=0,u\ne y}^{\delta_0-1}\varepsilon_{u,y}\lambda_{j,u}^t\epsilon_a^\top \epsilon_{\pi_{\tau-1}(a,x,u)}\cdot\sum\limits_{b=0}^{N'-1}\lambda_{j,b_x}\epsilon_b^\top \epsilon_b), & \textrm{if~}x<\tilde{x},\\[6pt]
			\lambda_{j,\tilde{y}}^{t+1}I_{N'}, & \textrm{if~}x=\tilde{x},\\[6pt]
			\sum\limits_{a=0}^{N'-1}(\lambda_{j,a_{x-1}}^t\epsilon_a^\top \epsilon_a\sum\limits_{b=0}^{N'-1}\lambda_{j,b_{x-1}}\epsilon_b^\top \epsilon_b)+\sum\limits_{a=0,a_{x-1}=y}^{N'-1}(\sum\limits_{u=0,u\ne y}^{\delta_0-1}\varepsilon_{u,y}\lambda_{j,u}^t\epsilon_a^\top \epsilon_{\pi_{\tau-1}(a,x-1,u)}\cdot \sum\limits_{b=0}^{N'-1}\lambda_{j,b_{x-1}}\epsilon_b^\top \epsilon_b), & \textrm{if~}x>\tilde{x},
		\end{array}\right.\notag\\
		&=&\left\{\begin{array}{ll}
			\sum\limits_{a=0}^{N'-1}\lambda_{j,a_x}^{t+1}\epsilon_a^\top \epsilon_a+\sum\limits_{a=0,a_x=y}^{N'-1}\sum\limits_{u=0,u\ne y}^{\delta_0-1}\varepsilon_{u,y}\lambda_{j,u}^{t+1}\epsilon_a^\top \epsilon_{\pi_{\tau-1}(a,x,u)}, & \textrm{if~}x<\tilde{x},\\[6pt]
			\lambda_{j,\tilde{y}}^{t+1}I_{N'}, & \textrm{if~}x=\tilde{x},\\[6pt]
			\sum\limits_{a=0}^{N'-1}\lambda_{j,a_{x-1}}^{t+1}\epsilon_a^\top \epsilon_a+\sum\limits_{a=0,a_{x-1}=y}^{N'-1}\sum\limits_{u=0,u\ne y}^{\delta_0-1}\varepsilon_{u,y}\lambda_{j,u}^{t+1}\epsilon_a^\top \epsilon_{\pi_{\tau-1}(a,x-1,u)}, & \textrm{if~}x>\tilde{x},
			\end{array}\right.\notag\\
		&=&\tilde{A}_{t+1,j,i,\delta_0}
	\end{eqnarray}
	}
	for $0\le i \ne j<n$ and $0\le t<r-1$. Let us define a block lower triangular matrix $\Psi_s$ of order $(r-s)N'$ as
	\begin{eqnarray*}  \label{eqn def Psij}
		\Psi_s=\underbrace{\left(\begin{array}{cccc}
				I_{N'} & & &\\
				\beta_s I_{N'} & -I_{N'} & & \\
				& \ddots & \ddots &  \\
				& & \beta_s I_{N'} & -I_{N'}
			\end{array}\right)}_{(r-s)\times(r-s)}.
	\end{eqnarray*} 
	 	 Then by \eqref{Eqn_Bt}, multiplying $H_{i,p,s}$ on the left by $\Psi_s$ we obtain			
	\begin{eqnarray*}
		\Psi_s H_{i,p,s}=\left(\begin{array}{c;{2pt/2pt}c;{2pt/2pt}c}
			I_{N'} & W_1 & W_2\\
			\hdashline
			\mathbf{0}_{(r-s-1)N'\times N'} & W_3Q_{s,0} & W_4Q_{s,1}
		\end{array}\right)
	\end{eqnarray*}
	where $W_1=(\underbrace{I_{N'},\cdots,I_{N'}}_{p-s-1})$, $W_2=\left(\begin{array}{ccc}
		\tilde{A}_{0,j_0,i,\delta_0} & \cdots & \tilde{A}_{0,j_{r-p-1},i,\delta_0}
	\end{array}\right)$, 
	\begin{eqnarray*}
		W_3=\left(\begin{array}{ccc}
				I_{N'} & \cdots & I_{N'} \\
				\vdots & \ddots & \vdots \\
				\beta_{s+1}^{r-s-2}I_{N'} & \cdots & \beta_{p-1}^{r-s-2}I_{N'}\\
			\end{array}\right),~W_4=\left(\begin{array}{ccc}
				\tilde{A}_{0,j_0,i,\delta_0} & \cdots & \tilde{A}_{0,j_{r-p-1},i,\delta_0}\\
				\vdots & \ddots & \vdots\\
				\tilde{A}_{r-s-2,j_0,i,\delta_0} & \cdots & \tilde{A}_{r-s-2,j_{r-p-1},i,\delta_0}\\
			\end{array}\right),   
	\end{eqnarray*}
	and
	\begin{eqnarray*}
	Q_{s,0}=\left(\begin{array}{ccc}	
			(\beta_s-\beta_{s+1})I_{N'} & & \\
			&\ddots &\\
			& & (\beta_s-\beta_{p-1})I_{N'}
		\end{array}\right),~	Q_{s,1}=\left(\begin{array}{ccc}
			\beta_s I_{N'}-\Upsilon_{j_0,i} & &\\
			& \ddots &\\
			& & \beta_s I_{N'}-\Upsilon_{j_{r-p-1},i}
		\end{array}\right).
	\end{eqnarray*}	
	
Then we obtain
	\begin{eqnarray*}  \label{eqn T-MDS_4}
		|\Psi_s||H_{i,p,s}|&=&|\Psi_s H_{i,p,s}|\notag\\
		&=&|I_{N'}|	\left|\left(\begin{array}{c;{2pt/2pt}c}
			W_3Q_{s,0} & W_4 Q_{s,1}\\
		\end{array}\right)\right|\notag\\&=&
		\left|\left(\begin{array}{c;{2pt/2pt}c}
			W_3 & W_4\\
		\end{array}\right)\left(\begin{array}{cc}
			Q_{s,0} & \\
			& Q_{s,1}
		\end{array}\hspace{-1.5mm}\right) \right|\notag\\
		&=&\left|\left(\hspace{-1.5mm}\begin{array}{cccc;{2pt/2pt}ccc}
			I_{N'} & I_{N'} & \cdots & I_{N'} &  \tilde{A}_{0,j_0,i.\delta_0} & \cdots & \tilde{A}_{0,j_{r-p-1},i,\delta_0}\\
			\beta_{s+1} I_{N'} &  \beta_{s+2}I_{N'} & \cdots & \beta_{p-1}I_{N'} &  \tilde{A}_{1,j_0,i.\delta_0} & \cdots & \tilde{A}_{1,j_{r-p-1},i,\delta_0}\\
			\vdots & \vdots & \vdots & \vdots  & \vdots & \vdots & \vdots\\
			\beta_{s+1}^{r-s-2} I_{N'} &  \beta_{s+2}^{r-s-2}I_{N'} & \cdots & \beta_{p-1}^{r-s-2}I_{N'} & \tilde{A}_{r-s-2,j_0,i.\delta_0} & \cdots & \tilde{A}_{r-s-2,j_{r-p-1},i,\delta_0}\\
		\end{array}\hspace{-1.5mm}\right)\hspace{-1.5mm}\left(\hspace{-1.5mm}\begin{array}{cc}
			Q_{s,0} & \\
			& Q_{s,1}
		\end{array}\hspace{-1.5mm}\right) \right|\notag\\
		&=&|\Psi_s||H_{i,p,s+1}||Q_{s,0}||Q_{s,1}|,
	\end{eqnarray*}	
where the last equality follows from  the definition of matrix $H_{i,p,s}$ in \eqref{eqn the matrix Hzys}.
	It is clear that $|Q_{s,0}|\ne 0$ since $\beta_u \ne \beta_v$ for any $0\le u\ne v<p$. Additionally, for any $j\in [0:n)\backslash\{i\}$, by \eqref{eqn matrix Upsilon} we have
	\begin{eqnarray*}
		\beta_u I_{N'}-\Upsilon_{j,i}=\left\{\begin{array}{ll}
			\sum\limits_{a=0}^{N'-1}(\beta_u-\lambda_{j,a_x})\epsilon_a^\top \epsilon_a, & \textrm{if~}x<\tilde{x},\\
			(\beta_u-\lambda_{j,\tilde{y}})I_{N'}, & \textrm{if~}x=\tilde{x},\\
			\sum\limits_{a=0}^{N'-1}(\beta_u-\lambda_{j,a_{x-1}})\epsilon_a^\top \epsilon_a, & \textrm{if~}x>\tilde{x},
		\end{array}\right.
	\end{eqnarray*}
	which is nonsingular according to conditions (ii) and (iii) in this lemma.  Then,  we arrive at the desired conclusion.


\begin{thebibliography}{99}
		
		\bibitem{array codes} M. Blaum, P.G. Farell, and H. van Tilborg, ``Array codes," \textit{Handbook of Coding Theory}, V. Pless and W. C. Huffman, Eds. Elsevier Science, 1998, vol. II, ch. 22, pp. 1855-1909.
		
		\bibitem{HDFS} D. Borthakur, ``HDFS Architecture Guide,"  in  \textit{Hadoop Apache Project,} 2008. [Online]. Available:
		http://hadoop.apache.org/common/docs/current/hdfs design.pdf
		
		\bibitem{total} R. Bhagwan, K. Tati, Y.-C. Cheng, S. Savage, and G.M. Voelker,
		``Total recall: System support for automated availability management,"
		in \textit{Proc. 1st Symposium on Networked Systems Design and Implementation (NSDI)}, San Francisco, CA, Mar. 2004.
		
		\bibitem{Dimakis} A.G. Dimakis, P. Godfrey, Y. Wu, M. Wainwright, and K. Ramchandran, ``Network coding for distributed storage systems," \textit{IEEE Trans. Inform. Theory,} vol. 56, no. 9, pp. 4539-4551, Sep. 2010.
		
		\bibitem{Dhash} F. Dabek, J. Li, E. Sit, J. Robertson, M. Kaashoek, and R. Morris, ``Designing
		a DHT for low latency and high throughput," in \textit{Proc. 1st Symposium on Networked Systems Design and Implementation (NSDI)}, San Francisco, CA, Mar. 2004.
		
		\bibitem{Micro} C. Huang, H. Simitci, Y. Xu, A. Ogus, B. Calder, P. Gopalan, J. Li, and
		S. Yekhanin, ``Erasure coding in Windows Azure storage," in \textit{Proc. 2012 USENIX Annual Technical Conference}, Boston, MA, pp. 1-12, Jun. 2012.
		
		\bibitem{repair_parity_zigzag} J. Li and X.H. Tang, ``Optimal exact repair strategy for the parity nodes of the $(k+2,k)$ Zigzag code," \textit{IEEE Trans. Inform. Theory,} vol. 62, no. 9, pp. 4848-4856, Sep. 2016.
		
		\bibitem{new_modified_zigzag} J. Li, X.H. Tang, and W. Xiang, ``A New Construction of $(k+2,k)$ Minimal Storage Regenerating Code  Over $\mathbf{F}_3$ with Optimal Access Property for All Nodes,"  \textit{IEEE Communications Letters,} vol. 20, no. 7, pp. 1289-1292, Jul. 2016.
		
		\bibitem{invariant_subspace} J. Li, X.H.  Tang, and U. Parampalli, ``A framework of constructions of minimal storage regenerating codes
		with the optimal access/update property," \textit{IEEE Trans. Inform. Theory,} vol. 61, no. 4, pp. 1920-1932, Apr. 2015.
		
		\bibitem{all_nodes} J. Li, X.H. Tang, and C. Tian , ``A Generic Transformation for Optimal Repair Bandwidth and Rebuilding Access in MDS codes,"   \textit{Proc. IEEE Int. Symp. Inform. Theory,} Aachen, Germany, pp. 1623-1627, Jun. 2017.
		
		\bibitem{t-transformation}  Y. Liu, J. Li, and X.H. Tang, ``A Generic Transformation to Generate MDS Codes with $\delta$-Optimal Access Property," \textit{arxiv preprint 	arXiv:2107.07733v2}, 2021.
		
		\bibitem{RES} N. Raviv, S. Natalia, and E. Tuvi, ``Constructions of high-rate minimum storage regenerating codes over small fields," \textit{IEEE Trans. Inform. Theory},  vol. 63, no. 4, pp. 2015-2038. Apr. 2017
		
		\bibitem{RS_codes} I. Reed and G. Solomon, ``Polynomial codes over certain finite fields,"  \textit{J. Soc.
			Ind. Appl. Math.,} vol. 8, no. 2, pp. 300-304, Jun. 1960.
		
		\bibitem{ocean} S. Rhea, C. Wells, P. Eaton, D. Geels, B. Zhao, H. Weatherspoon, and J. Kubiatowicz, ``Maintenance-free global data storage," \textit{IEEE Internet Comput.,} vol. 5, no. 5, pp. 40-49, Sep.-Oct. 2001.
		
		\bibitem{Sasidharan_Kumar} B. Sasidharan, G.K. Agarwal, and P.V. Kumar, ``A high-rate MSR code with polynomial sub-packetization level,"  \textit{Proc. IEEE Int. Symp. Inform. Theory,} Hong Kong, China, pp. 2051-2055, Jun. 2015.
		
		\bibitem{coupled_layer} B. Sasidharan, V. Myna, and P.V. Kumar, ``An explicit, coupled-layer construction of a high-rate MSR code with low sub-packetization level, small field size and all-node repair," \textit{arXiv preprint arXiv:1607.07335}.
		
		\bibitem{coupled} B. Sasidharan, V. Myna, and P.V. Kumar, ``An explicit, coupled-layer construction of a high-rate MSR code with low sub-packetization level, small field size and $d<(n-1)$," \textit{Proc. IEEE Int. Symp. Inform. Theory,} Aachen, Germany, pp. 2048-2052, Jun. 2017.
				
		\bibitem{Hadamard_strategy} X.H. Tang, B. Yang, J. Li, and H.D.L. Hollmann, ``A new repair strategy for the hadamard minimum storage regenerating codes for distributed storage systems," \textit{IEEE Trans. Inform. Theory,}  vol. 61, no. 10, pp. 5271-5279, Oct. 2015.
		
		\bibitem{zigzag} T. Tamo, Z. Wang, and J. Bruck, ``Zigzag codes: MDS array codes
		with optimal rebuilding," \textit{IEEE Trans. Inform. Theory,} vol. 59, no. 3, pp. 1597-1616, Mar. 2013.
		
		\bibitem{kumar}  M. Vajha, B.S. Babu, and P.V. Kumar, ``Explict MSR Codes with Optimal Access, Optimal Sub-packetization and Small Field Size for $d=k+1,k+2,k+3$," \textit{arxiv preprint arxiv:1804.00598}, 2018.
		
		\bibitem{extended_zigzag} Z. Wang, I. Tamo, and J. Bruck, ``On codes for optimal rebuilding access," in \textit{Proc. 49th Annu. Allerton Conf. Commun., Control, Comput.,} Monticello, IL, pp. 1374-1381, Sep. 2011.
		
		\bibitem{Long_arxiv} Z. Wang, T. Tamo, and J. Bruck, ``Explicit minimum storage regenerating codes,"  \textit{IEEE Trans. Inform. Theory,}  vol. 62, no. 8, pp. 4466-4480, Aug. 2016.	
		
		\bibitem{Barg_1} M. Ye and A. Barg, ``Explicit constructions of high-rate MDS array codes with optimal repair bandwidth," \textit{IEEE Trans. Inform. Theory,} vol. 63, no. 4, pp. 2001-2014, Apr. 2017.
		
		\bibitem{Barg_2} M. Ye and A. Barg, ``Explicit constructions of optimal-access MDS codes with nearly optimal sub-packetization," \textit{IEEE Trans. Inform. Theory,} vol. 63, no. 10, pp. 6307-6317, Oct. 2017.
		
		
	\end{thebibliography}
\end{document}